\documentclass[aos,preprint]{imsart}

\RequirePackage[OT1]{fontenc}

\RequirePackage{latexsym,amsmath,amsthm,amsfonts,mathtools,multirow,caption,float,graphicx,mathrsfs,algorithm,subcaption,hyperref,url,color}

\RequirePackage[authoryear]{natbib}
\RequirePackage[psamsfonts]{amssymb}
\usepackage[noend]{algpseudocode}

\arxiv{arXiv:0000.0000}

\startlocaldefs
\numberwithin{equation}{section}
\newtheorem{conj}{Conjecture}
\newtheorem{thm}[conj]{Theorem}
\newtheorem{cor}[conj]{Corollary}
\newtheorem{prop}[conj]{Proposition}
\newtheorem{lemma}[conj]{Lemma}

\def\eps{\varepsilon}
\def\wh{\widehat}

\def\to{\rightarrow}

\def\EE{\mathbb{E} }
\def\PP{\mathbb{P} }
\def\RR{\mathbb{R}}
\def\E{\mathcal{E}}
\def\C{\mathcal{C}}
\def\Z{\mathcal{Z}}
\def\iff{\Longleftrightarrow}
\def\wt{\widetilde}
\def\rs{\!\!\!}
\DeclareMathOperator*{\argmin}{arg\,min}
\DeclarePairedDelimiter{\ceil}{\lceil}{\rceil}
\DeclarePairedDelimiter{\floor}{\lfloor}{\rfloor}
\def\qm{q\wedge m}
\makeatletter
\def\BState{\State\hskip-\ALG@thistlm}
\makeatother

\endlocaldefs

\begin{document}
	
	\begin{frontmatter}
		\title{Adaptive estimation of the rank of the coefficient matrix in high dimensional multivariate response regression models}
		\runtitle{Adaptive estimation of the reduced rank}
		
		\begin{aug}
			\author{\fnms{Xin} \snm{Bing}\corref{} \ead[label=e1]{xb43@cornell.edu}}
			\and
			\author{\fnms{Marten H.} \snm{Wegkamp}	\ead[label=e2]{marten.wegkamp@cornell.edu}}
			
			\runauthor{X. Bing and M.H. Wegkamp}
			
			\affiliation{Cornell University}
			
			\address{
				Xin Bing\\
				Department of Statistical Science\\ 
				Cornell University\\
				301D Malott Hall\\
				Ithaca, New York 14853-3801\\
				United States of America\\
				\printead{e1}}
			
			\address{
				Marten H. Wegkamp\\
				Department of Mathematics \&\\
				Department of Statistical Science\\
				Cornell University\\
				432 Malott Hall\\ 
				Ithaca, New York 14853-3801\\
				United States of America\\
				\printead{e2}}
		\end{aug}
		
		\begin{abstract}
			We consider the multivariate response regression problem with a regression coefficient matrix of  low, unknown rank. In this setting, we analyze a new criterion for selecting the optimal reduced rank. This criterion differs notably from the one proposed in \cite{BSW} in that it does not require  estimation of the unknown variance of the noise, nor does it depend on a delicate choice of a tuning parameter. 
			We develop an iterative, fully data-driven procedure, that adapts to the optimal signal-to-noise ratio. This procedure finds the true rank in a few steps with overwhelming probability.
			At each step, our estimate increases, while at the same time  it does not exceed the true rank. Our finite sample results  hold for any sample size and any dimension, even when the number of responses and of covariates grow much faster than the number of observations. We perform an extensive simulation study that confirms our theoretical findings. The new  method performs better and is more stable than the procedure of   \cite{BSW} in both low-  and high-dimensional settings.		\end{abstract}
		
	\end{frontmatter}
	
	\section{Introduction}
	\subsection{Background}
	We study the multivariate response regression model
	\[ Y = XA+E \in\RR^{n\times m}\]
	with $X\in \RR^{n\times p}$ of   $\text{rank}(X)=q$ and $A\in \RR^{p\times m}$ of unknown $\text{rank}(A)=r$.
	We assume that the entries $E_{ij}$ of  $E$ are i.i.d. $N(0,\sigma^2)$ distributed with $\sigma^2<\infty$.
	Section \ref{sec_extension} discusses extensions to general, heavy tailed distributions of the errors $E_{ij}$. 
	
	Standard least squares estimation is tantamount to regressing each response on the predictors separately, thus  ignoring the multivariate nature of the possibly correlated responses. 
	In large dimensional settings ($m$ and $p$ are large relative to the sample size $n$), it is desirable to achieve a dimension reduction  in the coefficient matrix $A$.
	One popular way of achieving this goal, is to find a common subset of $s\le p$ covariates that are relevant for prediction, using 
	penalized least squares with a $\ell_1/\ell_2$ (group lasso) type penalty on the regression coefficients, see, for instance, \cite{BSW2012, bg,lounici, ob,yuanlin}
	to recover the support of the set of $s$ rows for which $A$ is non-zero. 
	
	Reduced rank regression is a different approach to achieve necessary dimension reduction.	The main premise is that $A$ has low rank, so that we can write $A=A_1 A_2$ for a $p\times r$ matrix $A_1$ and a $r\times m$ matrix $A_2$. Then only few linear combinations $X^*= XA_1$ of $X$ are needed to explain the variation of $Y$. \cite{i75} coined  the term {\em reduced-rank regression} for this class of models, but its history  dates back to  \cite{a51}. 
	There are many works on this topic in the classical setting of {\em fixed} dimensions $m$ and $p$, and sample size $n\to\infty$, see \cite{a99,r78,r73,r74} and more recently, \cite{a02}. A comprehensive overview on reduced rank regression is given by \cite{rv98}.
	Only recently, the high-dimensional case has been discussed: \cite{BSW,BSW2012,Giraud,GiraudBook,nw,rt}.
	
	The main topic of this paper is the estimation of the unknown rank. Determination of the rank of the coefficient matrix is the first key step for the estimation of $A$. 
	For known rank $r$, \cite{a99} derives the asymptotic distribution of the reduced rank regression coefficient matrix estimator $\wh A_r$
	in the asymptotic setting with $m,p$ fixed and $n\to\infty$. The estimator $\wh A_k$ is the matrix corresponding to minimizing the squared Frobenius or Hilbert-Schmidt norm  $\| Y - XB\|^2$
	over all $p\times m$ matrices $B$ of rank $k$ and has a closed form, due to the Eckart-Young theorem \citep{EckartY,s}.
	It is crucial to have the true rank $k=r$ for obtaining a good fit for both $\| XA - X \wh A_k\|^2$ and $\| A- \wh A_k\|^2$.
	In general, however, the rank $r$ is unknown {\em a priori.}
	The classical approach to estimate the  rank $r$ uses the likelihood ratio test, see \cite{a51}. 
	An elementary calculation shows that this statistic  coincides with 
	Bartlett's test statistic as a consequence of  the relation between reduced rank regression and canonical correlation analysis, see \cite{a51,rv}.
	Our main goal in this study is to 
	develop a non-asymptotic method to estimate $r$ that is easy to compute,  adaptively from the data, and valid for any values of $m, n$ and $p$, especially when the number of predictors $p$ and the number of responses $m$ are large.  The resulting estimator of $A$ can then be used  to construct a possibly  much smaller number of new transformed predictors, or  the most important canonical variables based on the original $X$ and $Y$, see \cite[Chapter 6]{i08}
	for a  historical account.
	Under weak assumptions on the signal, our estimate of $r$ can be shown to be equal to $r$ with overwhelming probability, to wit, $1-\exp(-\theta_1 mn)-\exp(-\theta_2(m+q))$, for some positive, finite constants $\theta_1,\theta_2$, so that the selection error is small compared to the overall error in estimating $A$.
	
	\subsection{Recent developments}
	
	\cite{BSW,BSW2012} proposed
	\begin{eqnarray}
	\label{een}
	\min_A \left\{ \| Y-XA\| ^2 + \mu\cdot \text{rank}(A) \right\}
	\end{eqnarray}
	and recommended the choice $\mu=C(\sqrt{m}+\sqrt{q})^2\sigma^2$ with constant $C>1$ for the tuning parameter $\mu$. 
	In particular, 
	\cite{BSW}  established a convenient closed form for the rank$(\wh A)= \wh r$ of the matrix $\wh A$
	that minimizes criterion (\ref{een}). They gave  sufficient conditions on the level of the smallest non-zero singular value of $XA$ to guarantee  that $\wh r$  consistently estimates
	$r $. 
	The disadvantage of this method is that a value for $\sigma^2$, in addition to the tuning parameter $C$,  is required for $\mu$. \cite{BSW} proposed to use the unbiased estimator 
	\begin{equation}\label{sigmaTilde}
	\widetilde{\sigma}^2 := \frac{\|Y-PY\|^2}{nm-qm}
	\end{equation} based on the projection $PY$ of $Y$ onto the range space of $X$.
	However, this becomes problematic when $(n-q)m$ is not large enough, or even infeasible when $n=q$.	\cite{Giraud} introduces another estimation scheme,  that does not require estimation of $\sigma^2$.   Unfortunately,  a closed form for the minimizer as in \cite{BSW} is lacking, and rank consistency in fact fails, as our simulations reveal in Appendix F.2 
	in the supplementary materials. Moreover,
	the procedures in both \cite{BSW} and \cite{Giraud} are rather sensitive to the choices of their respective tuning parameters involved. We emphasize that  \cite{Giraud} studies the  error  $\|X\wh A - XA\|^2$ and  not the rank of his estimator $\wh A$.

	\subsection{Proposed Research}
	This paper studies a third criterion, 
	\begin{eqnarray}
	\label{vier} 
	\wh{ \sigma}_k^2:=
	\frac{ \| Y- (PY)_k \| ^2}{ nm -\lambda k} .
	\end{eqnarray}
	Here $(PY)_k= \sum_{j=1}^k d_j(PY) u_j v_j^T$ 
	is the truncated singular value decomposition  of $PY$
	based on the (decreasing) singular values $d_j(PY)$  of the projection $PY$ and their corresponding singular vectors $u_j, v_j$.
	The range over which we minimize (\ref{vier}) is $\{0,1,\ldots, K \}$ with
	$ K= K_\lambda:=\floor{(nm-1)/\lambda} \wedge q\wedge m$ to avoid a non-positive denominator. 
	The purpose of this paper is to show that this new criterion produces a consistent estimator of the rank. It turns out that the choice of the optimal tuning parameter $\lambda$ involves a delicate trade-off. On the one hand, $\lambda$  should be large enough to prevent overfitting, that is, prevent selecting a rank that is larger than the true rank $r$. On the other hand, if one takes $\lambda$  too large, the selected rank will typically be smaller than $r$ as the procedure will not be able to distinguish the unknown singular values $d_j (XA)$ from the noise for $j>s$, for some $s=s(\lambda)<r$. 
	To effectively deal with this situation, we refine our initial procedure using  our new criterion (\ref{vier}) by an iterative procedure in Section \ref{sec_SRS} that provenly finds the optimal value of $\lambda$ and  consequently of the  estimate of $r$. This method does not require any data-splitting and our simulations show that it is very stable, even for general, heavy tailed error distributions. To our knowledge, it is rare feat to have a feasible algorithm that finds the optimal tuning parameter in a fully data-driven way, without data-splitting, and with mathematically proven guarantees.\\

	While our main interest is to  provide consistent estimators of the rank, 
	we briefly address estimation of the mean $XA$, which
	is the principal problem in \cite{BSW,BSW2012,Giraud}.   
	Our
	selected rank $\wh k$ automatically yields the estimate  $X\wh A:= (PY)_{\wh k}$. 
	We prove in Theorem 10 
	in Appendix D of the supplement
	that  on the event $\wh k=r$,  the inequality
	$\| X\wh A- XA\|^2 \le 4 r d_1^2(PE)$ holds, for our selected rank $\wh k$. This provides a direct link between rank consistency and optimal estimation of the mean, because  $ d_1^2(PE) \le 2   (m+q)\sigma^2 $ with overwhelming probability, see (\ref{d_0}) and (\ref{d_1}) below.
	(In fact, this bound continues to hold, up to a multiplicative constant,  for  subGaussian errors, using Theorem 5.39 of \cite{rv_rand_mat}.)
	Hence we can estimate $XA$ at the rate $r(m+q)$, which is  proportional to the number of parameters in the low rank model and
	minimax optimal \citep{BSW,BSW2012,Giraud}.
	Simulations in Appendix F.3 of the supplement show that our procedures  in fact provide better estimates of $XA$ than their competitors, even  in approximately low-rank models.\\

	
	The paper is organized as follows. Section \ref{sec_rank} shows that the minimizer of (\ref{vier}) has a closed form. The main results are discussed in Sections \ref{sec_GRS} and \ref{sec_SRS}. It obtains rank consistency in case of no signal ($XA=0$) and in case of sufficient signal. For the latter, we develop a key notion of signal-to-noise ratio that is required for rank consistency. A sufficient, easily interpretable condition will be presented that corresponds to a computable value (estimate) of the tuning parameter $\lambda$. We develop in Section \ref{sec_SRS} an iterative, fully automated procedure, which has a guaranteed recovery of the true rank (with overwhelming probability) under \emph{increasingly milder} conditions on the signal. The first step uses the potentially suboptimal estimate $\wh k_0$ developed in Section \ref{sec_GRS}, but which  is less than $r$, with overwhelming probability. This value $\wh k_0$ is used to update the tuning parameter $\lambda$. Then we minimize (\ref{vier}) again, and obtain  a new estimate $\wh k_1$, which in turn is used to update $\lambda$. The procedure produces each time a smaller $\lambda$, thereby selecting a larger rank $k$ than the previous one, whilst each time we can guarantee that the selected rank doesn't exceed the true rank $r$. This is a major mathematical challenge and its proof relies on highly non-trivial monotonicity arguments. The procedure stops when the selected rank does not change after an iteration.   Our results hold with high probability (exponential in $mn$ and $m+q$) which translates into extremely   accurate estimates under a weak signal condition. 
	
	Section \ref{sec_extension}  describes several extensions of the developed theory, allowing for non-Gaussian errors $E_{ij}$.
	
	A large simulation study is reported in Section \ref{sec_sim}.  It confirms our theoretical findings, and shows that our method improves upon  the methods proposed in  \cite{BSW}.
	
	The  proofs are deferred to \ref{suppA}.  
	\subsection{Notation}
	
	For any matrix $A$, we will use $A_{k\ell}$ to denote the $k,\ell$th element of $A$ (i.e., the entry on the $k$th row and $\ell$th column of $A$), and we write $d_1(A)\ge d_2(A) \ge \cdots$ to denote its ordered singular values.
	
	The Frobenius or Hilbert-Schmidt inner product $\left<\cdot,\cdot\right>$ on the space of matrices is defined as $\left<A,B\right>=\text{tr}(A^T B)$ for commensurate  matrices $A,B$.  The corresponding  norm is denoted by $\| \cdot \|^2$, and we recall that
	$\| A\|^2 = \text{tr}(A^TA)=\sum_j d_j^2(A)$ for any matrix $A$. 
	Moreover, it is known (\cite{s,EckartY}, see also the review by \cite{stew}) that minimizing 
	$ \| A- B\|^2$ over $B$ with $ r(B)\le r$ is achieved for $B=(A)_r=UD_rV^T$ based on the singular value decomposition of $A=UDV^T$ where $D_r$ denotes the diagonal matrix with $[D]_{ii}=d_i(A)$ for $i=1,\dots,r$. Hence,  $  \min_{B:\ r(B)=r} \| A- B\|^2 = \sum_{j>r} d_j^2(A)$.
	
	For other norms on matrices, we use $\|\cdot\|_2$ to denote the operator norm and $\|\cdot\|_*$ the nuclear norm (i.e., the sum of singular values). We have the inequalities $<A,B>=\text{tr}(A^TB) \le \| A\|_2 \| B\|_{*}$ and $\| A\|_{*} \le \sqrt{ \text{rank}(A) } \| A\|$.
	
	For two positive sequences $a_n$ and $b_n$, we denote by $a_n = O(b_n)$ if there exists constant $C>0$ such that $\lim_{n\to \infty}a_n/b_n \le C$. If $\lim_{n\to \infty}a_n/b_n\to 0$, we write $a_n = o(b_n)$.
	
	For general $m\times n$ matrices $A$ and $B$, Weyl's inequality \citep{Weyl} implies that
	$d_{i+j-1}(A+B) \le d_i(A)+ d_j(B)$ for $1\le i,j,\le q$ and $i+j\le q+1$ with $q=\min\{m,n\}$.
	
	We denote the projection matrix onto the column space of $X$ by $P$ and we write $q:=\text{rank}(X)$ and $N:=\text{rank}(PY) = \qm$. We set
	$\wh \sigma^2_0 := \| Y\| ^2 / (nm)$, by defining $ (PY) _0:=0$ and define $\wh \sigma^2 :=  \|E\|^2/(nm)$. Throughout the paper, we use $A$ to denote the true coefficient matrix and $r$ to denote its true rank.
	
	\section{Properties of the minimizer of the new criterion}\label{sec_rank}
	At first glance, it seems difficult to describe the  minimizer $\wh k$ of $k\mapsto \wh\sigma_k^2$ because
	{\em both} the numerator and denominator in $\wh\sigma_k^2$ are decreasing in $k$. However, it turns out that there is a unique minimizer with a neat explicit formula.
	First we characterize the comparison between $\wh \sigma_i$ and $\wh\sigma_j$ for $i\ne j$. 
	
	\begin{prop}\label{prop1}
		Let $i,j\in\{ 0,1,\ldots,K\}$  with $ i<j$. Then
		\begin{eqnarray}
		\label{crit1}
		\wh\sigma^2_j \le \wh \sigma_i^2\quad \iff\quad \frac{ 1}{j-i} {\sum_{k=i+1}^j d_k^2(PY) }\ge \lambda \wh\sigma_j^2.
		\end{eqnarray}
		In particular,
		\begin{eqnarray}
		\label{crit2}
		\wh\sigma^2_j \le \wh \sigma_{j-1}^2 \quad\iff\quad   d_j^2(PY) \ge \lambda \wh\sigma_j^2,
		\end{eqnarray}
		and
		\begin{eqnarray}
		\label{crit5}
		d_j^2(PY)  \le \lambda \wh \sigma_j^2 \quad\iff\quad d_j^2(PY) \le \lambda \wh\sigma^2_{j-1}.
		\end{eqnarray}
	\end{prop}
	
	This result and the monotonicity of the singular values $ d_1(PY) \ge d_2(PY) \ge \cdots $ readily yield the following statement.
	
	\begin{prop}\label{prop2}
		Let $k\in\{1,\ldots,K\}$. Then
		\begin{eqnarray}
		\label{crit3}
		\wh\sigma^2_k \le \wh \sigma_{k-1}^2  &\iff&  \wh\sigma^2 _k \le \wh\sigma_\ell^2 \quad \qquad\text{for all } \ell \le  k-1.\\
		\label{crit4}
		\wh\sigma^2_k \ge \wh \sigma_{k-1}^2 & \iff&   \wh\sigma_\ell^2 \ge  \wh\sigma_{k-1}^2 \qquad\text{for all } \ell> k-1.
		\end{eqnarray}
	\end{prop}
	
	It is clear that if $\wh \sigma_1^2 \ge \wh \sigma_0^2$, then $k=0$ minimizes $\wh\sigma_k^2$. Likewise, if $\wh\sigma_K^2 \le \wh\sigma_{K-1}^2$, then $k=K $ minimizes the criterion $\wh\sigma^2_k$.
	After a little reflexion, we see that $\wh \sigma^2_k$ is minimized at  the last $k$ for which $d_k^2(PY) \ge \lambda \wh\sigma_k^2$ holds.
	That is,
	\begin{equation}\label{hat k0}
	\wh k = \max\left\{  0\le k\le K\! : d_j^2(PY) \ge \lambda \wh \sigma_j^2\ \text{\rm for all } j\le k\ \text{\rm and } d_{k+1}^2(PY) < \lambda \wh \sigma_{k+1}^2 \right\}
	\end{equation}
	minimizes $\wh\sigma_k^2$ with the convention that the maximum of the empty set is 0.
	Properties (\ref{crit3}) and (\ref{crit4}) ensure that $d_j^2(PY) \ge \lambda \wh\sigma_j^2$ must hold automatically for all $j\le k$ as well 
	as $d_\ell^2(PY) < \lambda \wh\sigma_\ell^2$ for all $\ell>k$.
	That is,   $\wh k$ has an even more convenient closed form
	\begin{equation}
	\wh k 
	= \max\left\{  0\le k\le K :\ d_k^2(PY) \ge \lambda \wh \sigma_k^2\   \right\}= \sum_{k=1}^K  1 \left\{  d_k^2(PY) \ge \lambda \wh \sigma_k^2\  \right\}.\label{hat k}
	\end{equation}
	Summarizing, we have shown the following result.	
	
	\begin{thm}\label{closedform}
		There exists a unique  minimizer $\wh k$ of (\ref{vier}), given by (\ref{hat k}), such that $\wh\sigma_k^2$ is monotone decreasing for $k\le \wh k$, and monotone increasing for $k\ge \wh k$. 
	\end{thm}
	
	It is interesting to compare the choice $\wh k$ in (\ref{hat k})  with $\wh r$ in \cite{BSW}. In that paper, it is shown that
	(\ref{een}) is equivalent with
	\begin{eqnarray}
	\min_k \left\{ \| Y-(PY)_k\| ^2 + \mu k \right\}
	\end{eqnarray}
	based on the truncated singular value decomposition $U  D_k V^T$ of the projection $PY=U D V^T$ with $D_k=\text{diag}(D_{11},\ldots,D_{kk},0,\ldots,0)$. Furthermore, \cite{BSW} uses this formulation to derive 
	a closed form for $\wh r$, to wit,
	\begin{eqnarray}
	\label{twee}
	\wh r= \sum_{k\ge 1} 1\{ d_k^2(PY) \ge \mu\}
	\end{eqnarray}
	based on the singular values $d_1(PY)\ge d_2(PY)\ge\cdots$ of the projection $PY$.
	The main difference between (\ref{hat k}) and  (\ref{twee}) is that  $\wh k$ counts the number of singular values of $PY$ above a {\em variable} threshold, while $\wh r$ counts the number of singular values of $PY$ above a {\em fixed} threshold.
	Another difference is that the fixed threshold is proportional to the unknown variance $\sigma^2$, while the variable threshold is proportional to $\wh\sigma_k^2$, which can be thought of as an estimate of $\sigma^2$ only for $k$ close to $r$.
	
	To further illustrate the existence and uniqueness of $\wh k$, we perform one experiment to show how the $d_k(PY)$ and $\wh\sigma_k$ vary across different $k$, see Figure \ref{F1}. The plot first displays the monotone property of $\wh\sigma_k$ for $k\le\wh k$ and $k\ge \wh k$. It also justifies the definition of $\wh k$ in (\ref{hat k}) since the rank which minimizes $\wh\sigma_k$ is exactly what we defined. 
	\begin{figure}[H]
		\centering
		\includegraphics[width=0.45\textwidth]{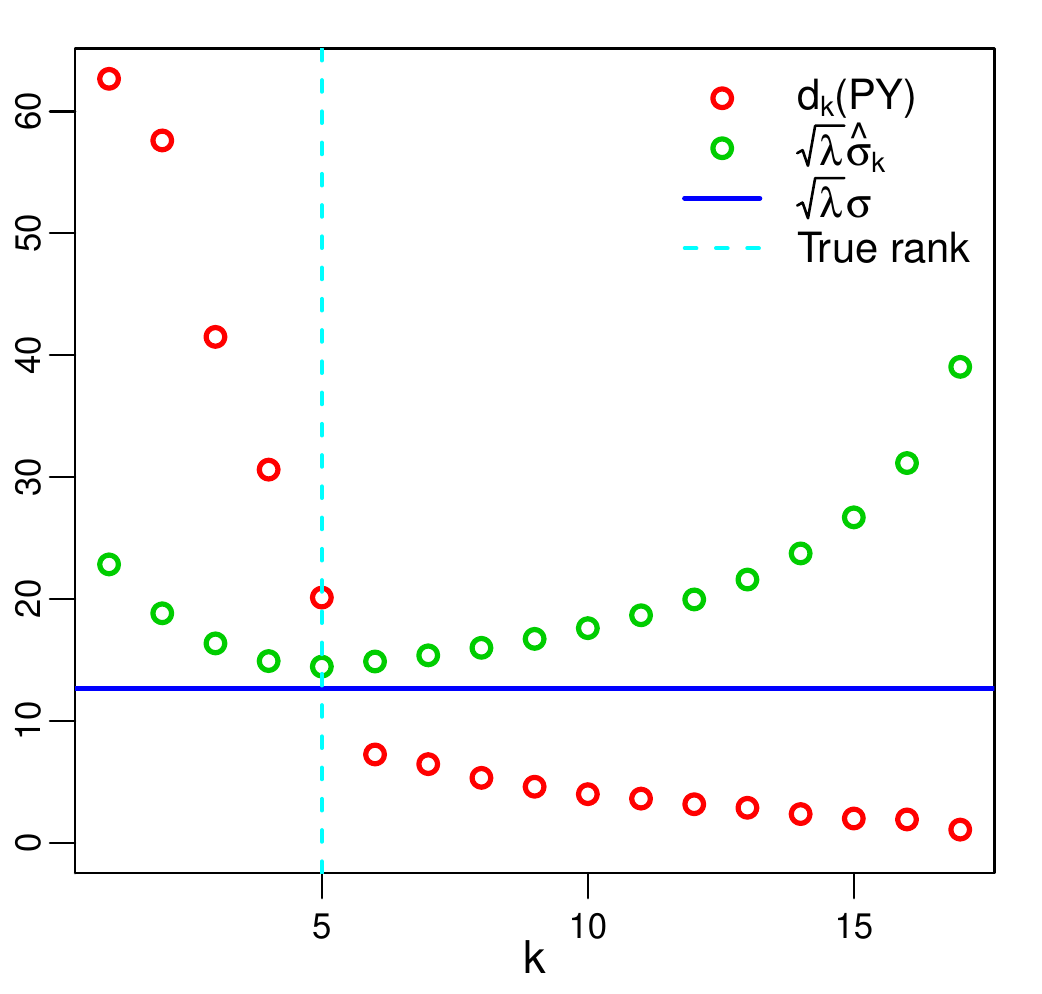}
		\caption{Plot of $d_k(PY)$ and $\sqrt{\lambda}\wh\sigma_k$ versus $k$. 
			In this experiment, we used  $n=150, m=30, p=20, q=20, r=5, \eta=0.1, b_0=0.2$ and $\sigma^2=1$, using the notation and simulation setup of Section \ref{SimSetup}.}
		\label{F1}
	\end{figure}
	
	\section{Rank consistency}\label{sec_GRS}
	\subsection{The null case $XA=0$}
	
	We treat the case $XA=0$ separately as the case $XA\neq0$ requires a lower bound on the  non-zero singular values of $XA$.
	
	\begin{thm} \label{thm:null}
		Assume $XA=0$. Then,  on the event
		$\{ d_1^2(PE) \le \lambda\wh \sigma^2\}$,
		we have $\wh k=0$. 
	\end{thm}

	We particularize Theorem \ref{thm:null} to the case where the entries of $E$ are independent $N(0,\sigma^2)$.
	In that case,
	general random matrix theory  and Borel's inequality for suprema of Gaussian processes, respectively,  give
	\begin{eqnarray}\label{d_0}\EE[ d_1(PE)] \le \sigma(\sqrt{m}+\sqrt{q})
	\end{eqnarray}
	and
	\begin{eqnarray}\label{d_1}
	\PP\{ d_1(PE) \ge \EE [ d_1(PE)] + \sigma t \} &\le& \exp(-t^2/2) \ \ \text{ for all $t>0$},
	\end{eqnarray}
	see 
	Lemma 3 in \cite{BSW}. Moreover,
	$(nm)\wh\sigma^2$ has a central $\chi^2_{mn}$ distribution, so that
	general tail bounds (\cite{john}) yield 
	\begin{eqnarray}
	\PP \left\{   \wh\sigma^2 \le \sigma^2 (1-\varepsilon) \right\} \label{chiL}
	&\le& \exp \left( -mn\varepsilon^2 /4 \right),\ \ \ \   0 \le \varepsilon<1\\
	\PP\left\{ \wh\sigma^2 \ge \sigma^2 (1+\varepsilon) \right\} & \le & \exp\left( -3mn\varepsilon^2 /16 \right), \ 0\le \varepsilon<1/2\label{chiR}
	.\end{eqnarray}
	We immediately obtain the following corollary by using (\ref{d_1}) -- (\ref{chiR}).
	\begin{cor}\label{c1}
		For any 
		$\lambda > (\sqrt{m}+\sqrt{q})^2$, we have
		$ \PP\{\wh{k} = 0\} \to 1$ exponentially fast as $nm\to\infty$ and  $ m+q\to\infty$.
	\end{cor}
	
	\subsection{The general case $XA\ne0$}
	
	The range over which we minimize (\ref{vier}) is $\{0,1,\ldots, K \}$ depends on $\lambda$ as the largest possible value is
	\begin{equation}\label{def_K}
	K = K_\lambda := \floor [\bigg]{\frac{nm-1}{\lambda}} \wedge m \wedge q
	\end{equation}
	to avoid a non-positive denominator in criterion (\ref{vier}).
	
	\begin{thm} \label{thm:RS1}
		Assume $r\le K_\lambda $.
		On the event
		\begin{eqnarray}\label{RS1}
		d_1^2(PE)   \le \lambda  \wh \sigma_{r}^2 := \frac{ \| Y- (PY)_r\|^2 }{(nm/\lambda ) -r},
		\end{eqnarray}
		we have $\wh k\le r$.
		If $r> K_\lambda$, then trivially $\wh k\le   r$ holds, with probability one.
	\end{thm}
	
	The restriction $r\le K_\lambda $   guarantees that $\wh \lambda \sigma_r^2$  is positive, that is, the event  (\ref{RS1}) is non-empty.
	If $r> K_\lambda$, then  $\wh k\le   r$,  holds trivially, with probability one, as $\wh k$ is selected from $\{0,\ldots,K_\lambda\}$.

	While the quantity $\wh\sigma_{r}^2$  is a natural one in this problem, it depends on the unknown rank $r$. It turns out that  quantifying $\wh \sigma_{r}^2$ is not trivial.
	
	\begin{prop}\label{prop:ongelijk}
		Assume $r\le K_\lambda$. On the event 
		\begin{eqnarray}\label{RS3}
		2 d_1^2(PE)   \le \lambda  \wh \sigma^2  ,
		\end{eqnarray}
		we have
		\begin{eqnarray}\label{eq_sigma_r}
		\wh \sigma^2 \le \wh \sigma_r^2 \le \frac{nm}{nm-\lambda r} \wh \sigma^2.
		\end{eqnarray}
	\end{prop}

	This result combined with Theorem \ref{thm:RS1} tells us that if $\lambda$ is chosen large enough,  we can guarantee that $\wh k \le r$. Moreover, this choice is independent of both  $r$
	and $\sigma^2$. Indeed, for matrices $E$ with independent $N(0,\sigma^2)$ Gaussian entries,
	any choice  $\lambda > 2 (\sqrt{m} +\sqrt{q})^2$ suffices:
	
	\begin{thm}\label{rankGaussian1}
		For 
		$\lambda = C(\sqrt{m}+\sqrt{q})^2$ with any numerical constant $C>2$,
		$\PP \{   \wh k\le r\}\to1$ as 
		$mn\to\infty$ and $ m+q\to\infty$.
	\end{thm}
	The convergence rate in Theorem \ref{rankGaussian1} is exponentially fast in $nm$ and $m+q$. Again, if $r> K_\lambda$, then $\PP\{ \wh k \le r \}=1$ holds trivially.
	
	Consistency of $\wh k$ can be achieved under a suitable signal to noise condition.
	
	\begin{thm} \label{thm:RS2}
		For $1\le s\le r\le K_\lambda$, on the event
		\begin{eqnarray}
		\label{RS2}
		d_s(XA) \ge d_1(PE) + \sqrt{\lambda} \wh \sigma_{r}
		\end{eqnarray} intersected with the event (\ref{RS1}),
		we further have $\wh k\in [s,r]$. 
	\end{thm}
	
	This theorem, combined with Proposition \ref{prop:ongelijk}, immediately yields the following corollary.
	
	\begin{cor}\label{cor: generalrank}
		For $1\le s\le r\le K_\lambda$, on the event $$\left\{2d_1^2(PE)\le \lambda\wh\sigma^2\right\} \cap
		\left\{d_s(XA)\ge \sqrt{\lambda}  \wh\sigma \left[ \frac{\sqrt{2} }{2} + \sqrt{\frac{nm}{nm-\lambda r}} \right] \,  \right\} 
		,$$
		we  have $\wh k \in[s,r]$.
	\end{cor}
	
	The choice of $\lambda$ impacts the possible values for $\wh k$, the  minimizer of criterion (\ref{vier}), and we see that the range $\{0,1,\ldots,K_\lambda\}$  increases as $\lambda$ decreases.
	If the true rank is rather large ($r> K_\lambda$), then no  guarantees for $\wh k$ can be made, except for the trivial, yet important observation that $\wh k\le r$.
	On the other hand, if   $r< K_\lambda$, which is arguably   the more interesting case for {\em  low rank} regression,   then consistency guarantees can be made under a suitable condition on the $r$-th singular value $d_r(XA)$ of $XA$.
	This condition becomes milder if $\lambda$ decreases.

	Let $\delta>0$. A slightly stronger restriction for the upper bound on $r$, 
	\begin{eqnarray} r < \frac{\delta}{1+\delta} \frac{nm}{\lambda}  \wedge m \wedge q \label{rankconstraint}
	\end{eqnarray}
	translates into a bound for 
	the ratio 
	\begin{eqnarray}\label{rankconstraint2}
	\frac{nm}{nm-\lambda r} \le 1+\delta
	\end{eqnarray}
	appearing in the lower bound for the {\em signal} $d_s(XA)$. We can further particularize to the Gaussian setting.
	
	\begin{thm}\label{rankGaussian}
		Let $\lambda = 2C(\sqrt{m}+\sqrt{q})^2 $ for some numerical constant $C>1$.
		Assume further that
		$r$ and $\delta$ satisfy (\ref{rankconstraint})
		and 
		\begin{eqnarray}\label{onzin}
		d_s(XA) &\geq& C' \sigma (\sqrt{m}+\sqrt{q})
		\end{eqnarray}
		for some $s\le r$ and some numerical constant $C'>\sqrt{C}(1+\sqrt{2(1+\delta)})$.
		Then $\PP \{ s\le \wh k\le r\}\to1$ as 
		$mn\to\infty$ and $ m+q\to\infty$.\\
		In particular, if (\ref{onzin}) holds for $s=r$, then $\wh k$ consistently estimates $r$.	
	\end{thm}
	The convergence rate in Theorem \ref{rankGaussian} is exponentially fast in $nm$ and $m+q$.
	This shows that the above procedure is highly accurate, which is confirmed in our simulation study.
	From the oracle inequality in the supplement, the fit $\|X\wh A_{\wh k}-XA\|^2$, for $s\le \wh k\le r$, differs only within some constant levels of the {\em noise level }$d_1^2(PE)$.

	\section{Self-tuning procedure}\label{sec_SRS}
	
	From Theorem \ref{thm:RS1} in Section \ref{sec_GRS}, we need to take $\lambda$, or in fact $\lambda \wh \sigma_r^2$, as $\wh \sigma_r^2$ depends on $\lambda$,  large enough
	to prevent overfitting, that is, to avoid  selecting a $\wh k$ larger than the true rank $r$. 
	On the other hand, we would like to keep $\lambda$  small to be able to detect  a small  signal level.
	Indeed,
	(\ref{RS1}) in Theorem \ref{thm:RS1} states that we need
	\begin{eqnarray}\label{startpoint}
	d_1(PE)\le \sqrt{\lambda}\wh\sigma_r =  \sqrt{\frac{\|Y-(PY)_r\|^2}{nm/\lambda-r}}.
	\end{eqnarray}
	The term on the right is decreasing in $\lambda$, so we should choose $\lambda$ as small as possible such that $\sqrt{\lambda}\wh \sigma_r$ is close to $d_1(PE)$.
	Without any prior knowledge on $r$, it is difficult to find the optimal choice for $\lambda$. However, 
	Theorem \ref{thm:RS1} and Proposition \ref{prop:ongelijk} tell us that  an initial $\lambda_0$ satisfying $\{2d_1^2(PE)\le \lambda _0\wh\sigma^2\}$ yields an estimated rank $\wh k_0$ with   $\wh k_0\le r$. Our idea is to use  this lower bound $\wh k_0$ for $r$ to reduce our value $\lambda_0$ to $ \lambda_1$.
	This, in turn, will yield a possibly larger estimated rank $\wh k_1$, which still obeys $\wh k_1\le r$.
	More precisely, we propose the following {\em Self-Tuning Rank Selection} (STRS) procedure. Let $Z$ be a $q\times m$ matrix with i.i.d. standard Gaussian entries and define $S_j=\EE[d_j^2 (Z)]$ with the convention $S_j := 0$ for $j>N=q\wedge m$. Moreover, let $\wh K_t := (nm/\wh \lambda_t) \wedge N$ for given $\wh \lambda_t$. For any $\eps \in (0,1)$, we define 
	\begin{eqnarray}\label{initlbd}
	\wh \lambda_0 &:=& 2(1+\eps)S_1\\
	\label{k0}
	\wh k_0 & :=& \argmin_{0 \le k\le \wh K_{0}} \frac{ \| Y- (PY)_k \| ^2}{ nm -\wh\lambda_0 k}
	\end{eqnarray}
	as starting values, and if $\wh k_0\ge 1$, for $t\ge 0$, we update
	\begin{eqnarray}\label{lbdt}
	\wh\lambda_{t+1} &: = &
	{nm \over (1-\eps)\wh R_t\ /\ \wh U_{t}+\wh k_t}
	\\
	\wh k_{t+1} & :=& \argmin_{\wh k_t \le k \le \wh K_{t+1}} \frac{ \| Y- (PY)_k \| ^2}{ nm -\wh\lambda_{t+1}k}\label{kt}
	\end{eqnarray}
	where
	\begin{equation}\label{def_rut}
	\wh R_t ~:=~ (n-q) m +  \rs\sum_{j=2\wh k_t + 1}^N\rs S_j,\qquad 
	\wh U_{ t}~: =~ S_1 \vee \left(S_{2\wh k_t+1}+ S_{2\wh k_t+2}\right).
	\end{equation}
	The procedure stops when $\wh k_{t+1} = \wh k_t$. The entire procedure is free of $\sigma^2$ and both
	$\wh R_t$ and 
	$\wh U_{t}$
	can be numerically evaluated by Monte Carlo simulations. Alternatively, we  provide an analogous procedure with analytical expressions in the supplement, but its performance in our simulations  is actually slightly inferior to the original procedure that utilizes Monte Carlo simulations. 
	
	
	Regarding the computational complexity, we emphasize that the above STRS procedure has almost \emph{the same} level of computational complexity as the methods in (\ref{een}) and (\ref{vier}). This is due to the fact that the computationally expensive singular value decomposition only needs to be computed once.  Additionally, in order to find the new rank in step (\ref{kt}), we  only need to consider values that are larger than (or equal to)  the previously selected rank. This avoids a lot of extra computation.
	
	
	The following proposition is critical for the feasibility of STRS. 
	
	\begin{prop}\label{prop:monolbd}
		We have $\wh\lambda_t > \wh\lambda_{t+1}$ and $\wh k_{t} \le \wh k_{t+1}$, for all $t \ge 0$. More importantly, $\wh k_{t}\le r$ for all $t \ge 0$, holds with probability tending to $1$ exponentially fast as $(q\vee m)\rightarrow \infty$ and $nm\rightarrow \infty$.
	\end{prop}
	
	The increasing property of $\wh k_t$  immediately yields the following theorem. 
	
	\begin{thm}\label{thm: SRS1}
		Let numerical constant $C > 2$. Let $\wt k$ be the minimizer of (\ref{vier}) using $\lambda = C(\sqrt{m}+\sqrt{q})^2$ and  $\wh k$ be the final selected rank of STRS starting from the same value $\wh\lambda_0=\lambda$. 
		Then 
		$$\PP\left\{\wt k \le \wh k \le r \right\}\to 1,$$ as $nm\to \infty$ and $(q\vee m)\to \infty$.
	\end{thm}
	
	We find that STRS  always selects a rank closer to the true rank than the (one step) Generalized Rank Selection procedure (GRS) from the previous section that uses (\ref{vier}) as its criterion.
	
	The decreasing property of $\wh \lambda_t$, stated in Proposition \ref{prop:monolbd}, implies an increasingly    milder condition on the required signal. Meanwhile, the way of updating $\lambda$ in step (\ref{lbdt}) is carefully chosen to maintain $\wh k_t \le r$. We refer to the proof for more explanations. Thus, if a proper sequence of signal-to-noise condition is met, we expect that STRS finds the rank consistently.
	The following theorem confirms this. 
	
	\begin{thm}\label{thm: signalSRS}
		Let $k_0<k_1<\cdots<k_T=r$ be a strictly increasing subsequence of $\{1,2,\ldots,r\}$ of length $T+1\le r$. 
		Define $\lambda_0$ as (\ref{initlbd}) and $\lambda_{t+1}$ obtained from (\ref{lbdt}) by using $k_t$ in lieu of $\wh k_t$, for $t=0,\ldots,T-1$. Assume $r$ and $\delta$ satisfy (\ref{rankconstraint}) for $\lambda_0$. Then, on the event
		\begin{eqnarray}\label{signalcondmsrs}
		d_{k_t}(XA)\ge C''\sigma\sqrt{\lambda_t}, \qquad   t = 0,\ldots,T
		\end{eqnarray}
		for some numerical constant $C''> 1/\sqrt{2}+\sqrt{1+\delta}$,  
		there exists $0 \le T'\le T$ such that $\wh k_{T'}  = r$, with probability tending to $1$, where $\wh k_0 \le \wh k_1 \le \cdots \le \wh k_{T'}$ are from (\ref{kt}).  	
	\end{thm}
	
	The sequence of $\{k_0, \ldots, k_T\}$ plays an important role for interpreting the above theorem. It can be regarded as a underlying sequence bridge starting from $1$ and leading towards the true rank.  The ideal case is $\{k_0, \ldots, k_T\}=\{r\}$ which leads to a one-step recovery, but requires a comparatively stronger signal-to-noise condition (\ref{signalcondmsrs}). At the other extreme, it could take $r$ steps to recover the true rank.
	We emphasize that the latter case requires the {\em mildest} signal-to-noise condition by the following two observations:
	\emph{
		\begin{enumerate}
			\item[(1)] Proposition \ref{prop:monolbd} guarantees that each time the updated $\lambda_t$ is decreasing;
			\vspace{-2mm}
			\item[(2)] The signal condition (\ref{signalcondmsrs}) is becoming milder as $\lambda_t$ gets smaller.
		\end{enumerate}
	}
	
	From   display (\ref{eq_sigma_r}) in Proposition \ref{prop:ongelijk} and from  Corollary \ref{cor: generalrank}, it is clear that we use the string of inequalities $\wh \sigma^2 \le \wh \sigma_r^2 \le nm/(nm-\lambda r)\wh\sigma^2$ to derive the required signal-to-noise condition. The second inequality becomes loose for large $r$ such that $nm-\lambda r$ is small, or equivalently, $\delta$ is large from (\ref{rankconstraint2}), which further implies a possible larger decrement of the required signal-to-noise condition by using STRS. To illustrate this phenomenon  concretely, we study to a special case where $r\ge N/2$ and $\{k_0, \ldots, k_T\} = \{\ceil{N/2}, r\}$, and show in the theorem below that the signal condition for recovering $r$ can be relaxed significantly when $\delta$ is large.
	
	\begin{thm}\label{thm: SRS-snr}
		Define $\lambda_0 = 2C(\sqrt{m}+\sqrt{q})^2$ with $C = 8/7$. Assume $r$ and $\delta$ satisfy (\ref{rankconstraint}) for $\lambda_0$ and
		\begin{eqnarray}\label{eq_srs_q/2}
		d_{\ceil{N/2}}(XA) &\ge& C'\left[1+ \sqrt{2(1+\delta)}\right]\sigma(\sqrt{m}+\sqrt{q})\\\label{eq_srs_r}
		d_r(XA) &\ge& C'\left[1 + \sqrt{1+\delta \over 1+\delta/8}\ \right]\sigma(\sqrt{m}+\sqrt{q})	
		\end{eqnarray}
		for some numerical constant $C' > 2\sqrt{2/7}$.
		Then either $\wh k_0 = r$ or $\wh k_1 = r$, with probability tending to $1$, as $N= q \wedge m\to\infty$. Here $\wh k_0$ and $\wh k_1$ are selected from (\ref{k0}) and (\ref{kt}). 
	\end{thm}
	
	The lower bound condition  (\ref{eq_srs_q/2}) is condition  (\ref{onzin}) with $s= \ceil{N/2} <r$.
	As a simple numerical  illustration,  we compare (\ref{eq_srs_q/2}) (and therefore (\ref{onzin})) with  (\ref{eq_srs_r}) for   $\delta=4 $ and $100$. If $\delta = 4$, we obtain 
	\[
	d_{\ceil{N/2}}(XA) \ge 4.17C'(\sqrt{m}+\sqrt{q})\sigma,\qquad 
	d_r(XA) \ge 2.83C'(\sqrt{m}+\sqrt{q})\sigma,	
	\]
	while
	if $\delta = 100$,  we have 
	\[
	d_{\ceil{N/2}}(XA) \ge 15.3C'(\sqrt{m}+\sqrt{q})\sigma,\qquad 
	d_r(XA) \ge 3.74C'(\sqrt{m}+\sqrt{q})\sigma.	
	\]
	As we can see, for small $\delta$, the signal-to-noise condition decreases slightly.  However, for larger $\delta$, $\sqrt{1+\delta}$ could be quite large, while $\sqrt{ (1+\delta)/  ( 1+\delta/8)}$ is always   bounded above by $2\sqrt{2}$. To further elaborate the implications of small/large $\delta$, we consider two cases by 
	recalling the rank constraint (\ref{rankconstraint}):
	\begin{enumerate}
		\item[(1)] If ${nm / \lambda_0} \ge (1+\delta)N/\delta$, the rank constraint (\ref{rankconstraint}) reduces to simply $r \le N $. From (\ref{rankconstraint2}), a smaller value for $\delta$ leads  to a smaller value for  $nm/(nm-\lambda_0 r)$, provided $r\le N$. Therefore,  $\wh\sigma^2\le\wh \sigma_r^2 \le nm/(nm-\lambda_0 r)\wh\sigma^2$ should be tight and  we expect a smaller reduction of the  signal condition  for smaller values of  $\delta$. 
		On the other hand, when $nm$ and $\lambda_0N$ are close, meaning $\delta$ is large, we expect a considerable relaxation of the signal condition for comparatively large $r$. These two points are clearly reflected in (\ref{eq_srs_q/2}) and (\ref{eq_srs_r}).
		\item[(2)] If ${nm / \lambda_0} \le (1+\delta)N/\delta$, it follows that  $r \le \{ \delta / ( 1+\delta) \} (nm /  \lambda_0)$. Then a smaller $\delta$ means a stronger restriction on $r$ which implies $\wh\sigma^2 \le \wh\sigma_r^2 \le nm/(nm-\lambda_0 r)\wh\sigma^2$ becomes tight and we expect a modest relaxation  of  the signal condition.
		If  $\delta$ is large, then  $nm-\lambda_0 r$ could be small for a comparatively large $r$. Thus $nm/(nm-\lambda_0 r)$ would explode and we expect a significant decrease in the lower bound (\ref{eq_srs_r}) for the signal. Both  these observations agree  with our results. It is worth mentioning that when $nm$ is small comparing to $\lambda_0N$, $\delta$ is likely to be large. For instance, when $m = q = n$ is moderate, taking $\lambda_0 = 2(\sqrt{m}+\sqrt{q})^2$ yields $nm /\lambda_0 = q/8$ which is not quite large already. Imposing a small $\delta$ in this case would further restrict the range of $r$.
	\end{enumerate}
	
	Recall that the range of allowable rank $\{0,\ldots,K_\lambda\}$ in (\ref{def_K})  increases as $\lambda$ decreases. This means that, after a few iterations,   the true rank could be selected even when it was out of the possible range $\{0,\ldots,K_{\lambda_0}\}$ at the beginning. This phenomenon is clearly supported by our simulations in Section \ref{sec_sim_2}. In addition, the following proposition proves that  $K_{\lambda_0}$ can be extended to $N=q\wedge m$ in some settings even when (\ref{rankconstraint}) is not met for $\lambda_0$. 
	
	
	\begin{prop}\label{prop: ex-rank}
		Let $\lambda_0 = 2C(\sqrt{m}+\sqrt{q})^2$ with $C= 8/7$ and assume $nm/\lambda_0 \ge 3N/4$. Suppose the first selected rank from (\ref{k0}) by using $\lambda_0$ satisfies $\wh k_0 \ge N/2$ and 
		\[
		d_r(XA) \ge C'(1+2\sqrt{3})(\sqrt{m}+\sqrt{q})\sigma,
		\]
		for some numerical constant $C'> 2\sqrt{2/7}$. Then,  
		we have $\PP\{\wh k_1 = r\} \to 1$ for any $N/2 \le r \le N$, as $N=q\wedge m \to \infty$.
	\end{prop}
	
	In order to be able to select among  ranks from $N/2$ to $N$ in the first step,   (\ref{rankconstraint2}) requires $nm/\lambda_0 \ge (1+\delta)N/\delta$. However, Proposition \ref{prop: ex-rank}   relaxes this to $nm/\lambda_0 \ge 3N/4$ from Proposition \ref{prop: ex-rank}. 
	
	\section{Extension to   heavy tailed error distributions}\label{sec_extension}
	
	Most results in this paper are finite sample results and apply to any matrix $E$. Only Corollary \ref{c1}, Theorem \ref{rankGaussian} and the results in Section \ref{sec_SRS} require Gaussian errors. They appeal to 
	precise concentration inequalities of $d_1(PE)$ around $\sqrt{m}+\sqrt{q}$, making use of the fact that $d_1(PE)= d_1(\Lambda U^T E)$ based on the eigen decomposition of $P=U\Lambda U^T$, and the fact that  $\Lambda U^T E$ in turn is again Gaussian. In general, if $E$ has independent entries, then the transformations $PE$ or $\Lambda U^TE$  no longer have
	independent entries, although their columns remain independent. 
	Regardless, our simulations reported in Sections \ref{sec_sim_4} and  \ref{sec_sim_5} support our conjecture that our iterative method is flexible and our results continue to hold for general distributions, such as $t$-distributions with $6$ degrees of freedom, for independent errors $E_{ij}$.
	For some important special cases we are able to formally allow for  errors with  finite fourth moments only.  
	
	\subsection{Heavy tailed errors distributions with $n / q\to 1$}
	We first consider the case $n / q\to 1$, which is  likely  to occur in   high-dimensional settings ($p>> n$). The following theorem guarantees the  rank recovery via the GRS procedure for errors with heavy tailed distributions. 
	
	\begin{thm}\label{thm: heavy tails n = q}
		Let $\lambda > 2(\sqrt{m}+\sqrt{q})^2$.
		Assume  that the entries of $E$ are i.i.d. random variables with mean zero and finite fourth moments.  Furthermore, assume  $n / q\to 1$, $r$ and $\delta$ satisfy (\ref{rankconstraint}) and 
		\begin{equation}\label{cond: signal}
		d_s(XA) \ge C\sigma (\sqrt{m}+\sqrt{q}),
		\end{equation}
		for some $s\le r$ and some numerical constant $C > 1+\sqrt{2(1+\delta)}$.
		Then, we have $\PP\{s\le \wh k\le r \} \to 1$ as $n \vee m \to \infty$, where $\wh k$ is selected from (\ref{vier}).\\ 
		In particular, if (\ref{cond: signal}) holds for $s = r$, then $\wh k$ consistently estimates $r$.
	\end{thm}
	
	For a special case, that of skinny matrices $XA$, that is, $m=O(n^\alpha)$ or $n=O(m^\alpha)$ for some $0\le \alpha <1$, we propose the following {\em Simplified Self-Tuning Rank Selection} (SSTRS) procedure. Given any $\eps \in (0,1)$, we set 
	\begin{equation}\label{initlbd_ex}
	\wh \lambda_0 ~:=~ 2(1+\eps)(m\vee q),
	\qquad
	\wh k_0  ~:=~ \argmin_{0 \le k\le \wh K_{0}} \frac{ \| Y- (Y)_k \| ^2}{ nm -\wh\lambda_0 k}
	\end{equation}
	as starting values, and if $\wh k_0\ge 1$, for $t\ge 0$, we update
	\begin{equation}\label{lbdt_ex}
	\wh\lambda_{t+1} ~:= ~
	{nm \over (1-\eps)\bigl[(m\wedge q)/2 - \wh k_t\bigr]_++\wh k_t},\quad
	\wh k_{t+1} ~:=~ \argmin_{\wh k_t \le k \le \wh K_{t+1}} \frac{ \| Y- (Y)_k \| ^2}{ nm -\wh\lambda_{t+1}k}
	\end{equation}
	where $[x]_+ := \max\{x, 0\}$ and $\wh K_t := (nm/\wh \lambda_t) \wedge q\wedge m$ for $t = 0, 1, \ldots$. The procedure stops when $\wh k_t = \wh k_{t+1}$ and we   have the following result.
	
	\begin{thm}\label{thm: heavy tails skinny}
		Assume $E_{ij}$ are i.i.d. random variables with mean zero and finite fourth moments. Suppose that $n/q\to1$ and either
		$m=O(n^\alpha)$ or $n=O(m^\alpha)$ for some $\alpha\in [0,1)$.  Let $\{k_t\}_{t=0}^{T}$ be defined as Theorem \ref{thm: signalSRS}.	Define $\lambda_0$ as (\ref{initlbd_ex}) and $\lambda_{t+1}$ obtained from (\ref{lbdt_ex}) by using $k_t$ in lieu of $\wh k_t$, for $t\ge 0$. Assume $r$ 
		and $\delta$ satisfy (\ref{rankconstraint}) for $\lambda_0$.\\
		Then, on the event
		\begin{eqnarray}
		d_{k_t}(A)\ge C\sigma\sqrt{\lambda_t}, \qquad   t = 0,\ldots,T
		\end{eqnarray}
		for some numerical constant $C> 1+\sqrt{2(1+\delta)}$,  
		there exists $0 \le T'\le T$ such that $\wh k_{T'}  = r$, with probability tending to $1$, as $n\vee m \to  \infty$. Here $\wh k_0 \le \wh k_1 \le \cdots \le \wh k_{T'}$ are given in  (\ref{initlbd_ex}) and (\ref{lbdt_ex}).  	
	\end{thm}
	
	{\sc Remark.}
	As mentioned earlier, the entries of $PE$ no longer inherit the independence from $E$ when the distribution of  the independent  entries $E_{ij}$ is not Gaussian. Nevertheless, by exploiting the independence of columns of $PE$, Theorem 5.39 in \cite{rv_rand_mat} shows that 
	$d_1(PE) \le C_E\sqrt{q}+\sqrt{m}$ with high probability, provided the entries of $E$ are independent sub-Gaussian random variables with unit variance.
	The constant $C_E$  above  unfortunately involves the unknown sub-Gaussian norm, which differs from $\sigma^2$. However, provided $q= o(m)$ and $E$ has i.i.d. sub-Gaussian entries, we simply have $d_1(PE) \le (1+o(1) )\sqrt{m}$. Hence, for this case, it should be clear that our STRS procedure based on (\ref{initlbd}) - (\ref{kt}) can be directly applied with statistical guarantees stated in Section \ref{sec_SRS}.

	\subsection{A special model: $Y=A+E$}
	We emphasize that our procedure can be applied to the important special model $Y = A+E$ where the entries of $E$ are i.i.d. random variables with mean zero and finite fourth moments. The following results guarantee that our procedure can consistently estimate the rank of $A$. They are essentially the same statements as Theorems \ref{thm: heavy tails n = q} and \ref{thm: heavy tails skinny} for the case $Y=XA+E$, but this time  {\em without} the disclaimer $n/q\to1$.
	
	\begin{thm}\label{thm: special model1}
		Assume the entries of $E\in \RR^{n\times m}$ are i.i.d. random variables with mean zero and   finite fourth moments.
		Assume further that
		$r$ and $\delta$ satisfy (\ref{rankconstraint})
		and 
		\begin{eqnarray}\label{onzin2}
		d_s(A) \geq C\sigma(\sqrt{n}+\sqrt{m})
		\end{eqnarray}
		for some $s\le r$ and some  numerical constant $C>1+\sqrt{2(1+\delta)}$.
		Then $\PP \{ s\le \wh k\le r\}\to1$ as 
		$m\vee n\to\infty$, where $\wh k$ is selected from (\ref{initlbd_ex}) by using   $\wh \lambda_0=\lambda> 2(\sqrt{n}+\sqrt{m})^2$.\\
		In particular, if (\ref{onzin2}) holds for $s=r$, then $\wh k$ consistently estimates $r$.	
	\end{thm}
	
	In particular, when $A$ is skinny, that is, $m=O(n^\alpha)$ or $n=O(m^\alpha)$ for some $0\le \alpha<1$, our newly proposed SSTRS in (\ref{initlbd_ex}) - (\ref{lbdt_ex}) maintains the rank consistency for this model. 
	
	\begin{thm}\label{thm: special model2}
		Let $E_{ij}$ be i.i.d. random variables with mean zero and finite fourth moments, 
		$m=O(n^\alpha)$ or $n=O(m^\alpha)$ for some $\alpha \in [0,1)$, and 
		assume $r$ 
		and $\delta$ satisfy (\ref{rankconstraint}) for $\lambda_0$ given in (\ref{initlbd_ex}).
		Let  $\{k_t\}_{t=0}^{T}$ be defined as Theorem \ref{thm: signalSRS}  and $\lambda_{t+1}$ obtained from (\ref{lbdt_ex}) by using $k_t$ in lieu of $\wh k_t$, for $t\ge 0$.
		Then, on the event
		\begin{eqnarray}
		d_{k_t}(A)\ge C\sigma\sqrt{\lambda_t}, \qquad   t = 0,\ldots,T
		\end{eqnarray}
		for some numerical constant $C> 1+\sqrt{2(1+\delta)}$,  
		there exists $0 \le T'\le T$ such that $\wh k_{T'}  = r$, with probability tending to $1$, as $m\vee n\to\infty$. Here  $\wh k_0 \le \wh k_1 \le \cdots \le \wh k_{T'}$ are defined in (\ref{initlbd_ex}) and (\ref{lbdt_ex}).  	
	\end{thm}

	\section{Empirical study}\label{sec_sim}
	The simulations  in  Sections \ref{sec_sim_1} and \ref{sec_sim_2}   compare the methods discussed in this paper with some existing methods. Sections \ref{sec_sim_3} -- \ref{sec_sim_5} verify our results for the proposed method. Our conclusions are summarized in Section \ref{sec_sim_conclusion}. In Section \ref{sec_sim_tightness}, we perform an additional simulation to check the tightness of signal-to-noise condition in (\ref{RS2}). In the supplement, more simulations   compare STRS   using Monte Carlo simulations with STRS using  deterministic bounds in Appendix F.1. We also present an example in Appendix F.2 to show that the method proposed by \cite{Giraud} fails to recover the rank. Finally,   simulations in Appendix F.3 compare STRS with other competing methods in terms of the errors $\|X\wh A - XA\|$, $\| \wh A-A\|$ and the selected rank.
	
	\subsection{Methods and notations}
	We first introduce the methods in our simulation. \cite{BSW} proposed the method in (\ref{een}) to select the optimal rank by using $\mu = Cd_1^2(Z)\widetilde{\sigma}^2$, where $Z$ has $q\times m$ i.i.d. $N(0,1)$ entries and $\wt{\sigma}^2$ in (\ref{sigmaTilde}) is the unbiased estimator of $\sigma^2$. The leading constant $C>1$ needs to be specified. A deterministic upper bound which could be used instead is $C(m+q)\wt \sigma^2$.  \cite{BSW} suggests to use $C = 2$ based on its overall performance. However, there is no reason for one particular choice of $C$ being globally optimal, which  was confirmed  in our  simulations.
	Another option for choosing the tuning parameter is to use $k$-fold cross-validation. However, there is no theoretical guarantee of the feasibility for cross-validation, especially if the rows $X_{i\cdot}$ of $X$ are non-i.i.d..  In contrast, our proposed procedures (with and without self-tuning) are completely
	devoid of choosing a tuning parameter and estimating $\sigma^2$.
	\vspace{3mm}
	\\
	{\sc Notation.}	We use BSW to denote the method proposed in \cite{BSW}. For those methods proposed in this paper, we denote by GRS, STRS  and SSTRS the method without self-tuning in (\ref{vier}), the one with self-tuning from (\ref{initlbd}) - (\ref{kt}) and the simpler version also with self-tuning from (\ref{initlbd_ex}) - (\ref{lbdt_ex}), respectively.  We further use {BSW-C} to denote BSW with specified leading constant $C$.
	
	\subsection{Task description}
	We divide the simulation study up into five parts. In the first part, we show that there is no optimal constant  $C$ for the  {BSW-C} method  which works for all $r$.
	In the second part, we compare the performance of STRS and {BSW-C} (for various choices of $C$). The third part demonstrates the improvement of STRS over GRS shown in Section \ref{sec_SRS}, in terms of requiring a smaller signal-to-noise ratio (SNR) and enlarging the range of possible selected ranks. The fourth part verifies the performance of GRS, STRS and SSTRS for non-Gaussian errors corresponding to our results and settings of Section \ref{sec_extension}. 
	The last part extends the fourth part and supports our conjecture that STRS continues to work for general  heavy tailed distributions
	under general settings. 
	
	\subsection{Simulation setup}\label{SimSetup}
	In general,  we consider three settings.
	The first setting is the more favorable one where the sample size $n$ is  larger than the number of covariates $p$. The other two are high-dimensional,    $n<p$, and hence more challenging, with the last setting  focussing on the worst case scenario of $n$   close to $q$ for a moderate $m$.
	
	Our experiments are inspired by those of \cite{BSW}. When $n\ge p$, the $n\times p$ design matrix $X$ is generated by independently drawing $n$ times  $p$-dimensional Gaussian vectors with mean zero and covariance matrix specified by $\Sigma_{i,j}=\eta^{|i-j|}$ with $\eta\in(0,1)$, for $i,j=1,\ldots,p$. When $n < p$, we let $X = X_1X_2\Sigma^{1/2}$ where $X_1\in \RR^{n\times q}$ and $X_2\in \RR^{q\times p}$ have i.i.d. $N(0,1)$ entries. The regression coefficient matrix $A$ is given by $A = b_0M_{p\times r}M_{r\times m}$ where the entries of $M$ are i.i.d. $N(0,1)$. As before, $r$ denotes the rank of $A$ and satisfies $r\le \qm$. Regarding the error matrix $E$, each entry is generated from $N(0,1)$ except in Sections \ref{sec_sim_4} and \ref{sec_sim_5} where we use $t_\nu$-distributions with $\nu$ degrees of freedom. 
	
	The difference with \cite{BSW} lies in the way we vary the signal-to-noise-ratio (SNR) defined as $d_r(XA)/ \EE[d_1(PE)]$. Instead of using various combinations of  $\eta$ and $b_0$, we vary the SNR by generating $A$ with different   ranks. Specifically, for given $\eta$ and $b_0$,  we first generate   $X$, and then, for each $r$ in some specified range, we generate $A$ of rank $r$.
	For each pair   $(X, A)$, we generate 200 error matrices $E$, calculate the SNR, and record the rank recovery rate and the mean selected rank for various methods in the $200$ replications.

	\subsection{Experiment 1}\label{sec_sim_1}
	We compare the rank recovery of {BSW-C} for $C$ in $\{0.7,0.9,1.1,1.3,1.5\}$ with STRS   in both low- and high-dimensional settings. In the low-dimensional case, we consider   $n=150$, $m=30$, $p=q=20$,  $r\in\{0,\ldots,20\}$ and $\eta=0.1$. For $b_0$, we choose from $\{0.15,0.20,0.25\}$ to illustrate, more clearly, the effect of $r$ on the recovery rate. The high-dimensional setting has $n=100$, $m=30$, $p=150$, $q=20$, $\eta=0.1$,  $b_0=\{0.03,0.05,0.07\}$ and $r\in\{ 0, \ldots,20\}$. The rank recovery rate and mean selected ranks for both low-  and high-dimensional cases are shown in Figures \ref{E1F1} and \ref{E1F2}, respectively.
	
	\medskip
	
	{\sc{Result}: } Both figures  demonstrate that BSW-C with a smaller $C$, say $0.7$ or $0.9$, performs better when the true rank $r$ is large, in the sense of requiring a smaller SNR, but tends to overfit for small $r$. In contrast, BSW-C with a larger $C$ does a better job in preventing overfitting for small $r$, but requires a larger SNR. The performance of {BSW-C} does not seem to depend on $r$ as we separate the effect of SNR away from this phenomenon by varying $b_0$. This   suggests that there is no optimal leading constant $C$ for BSW-C to guarantee consistent  rank estimation   for all possible $r$. On the other hand, STRS performs globally better and more stable than {BSW-C} in all settings. It  prevents overfitting for both small and large $r$ and requires a smaller SNR than BSW-1.3 and BSW-1.5.
	Finally,  the role of SNR for the rank recovery is striking. If it is too small (less than $0.8$), we completely fail to recover the  rank.
	This justifies the signal-to-noise condition in (\ref{RS2}) and is explained by the (fast) exponential tail bounds in our main results.
	
	\subsection{Experiment 2}\label{sec_sim_2}
	
	Based on the results in Section \ref{sec_sim_1}, we 
	compare BSW-1.1 and BSW-1.3 with STRS. Figures \ref{E1F1} and \ref{E1F2} show the advantage of STRS over {BSW-C} in both low- and  high-dimensional settings when $n$ is not too small compared to $q$. Here we focus on the worst case scenario of    $n\approx q$ and set  $n=150$, $m=30$, $p=200$, $\eta=0.1$, $b_0=0.011$, $q\in \{143, 145, 147, 149\}$  and $r\in\{0,\ldots,22\}$. The   recovery rates are shown in Figure \ref{E3F1}.
	
	\begin{figure}[H]
		\centering
		\begin{tabular}{cc}
			\centering
			\includegraphics[width=.45\linewidth]{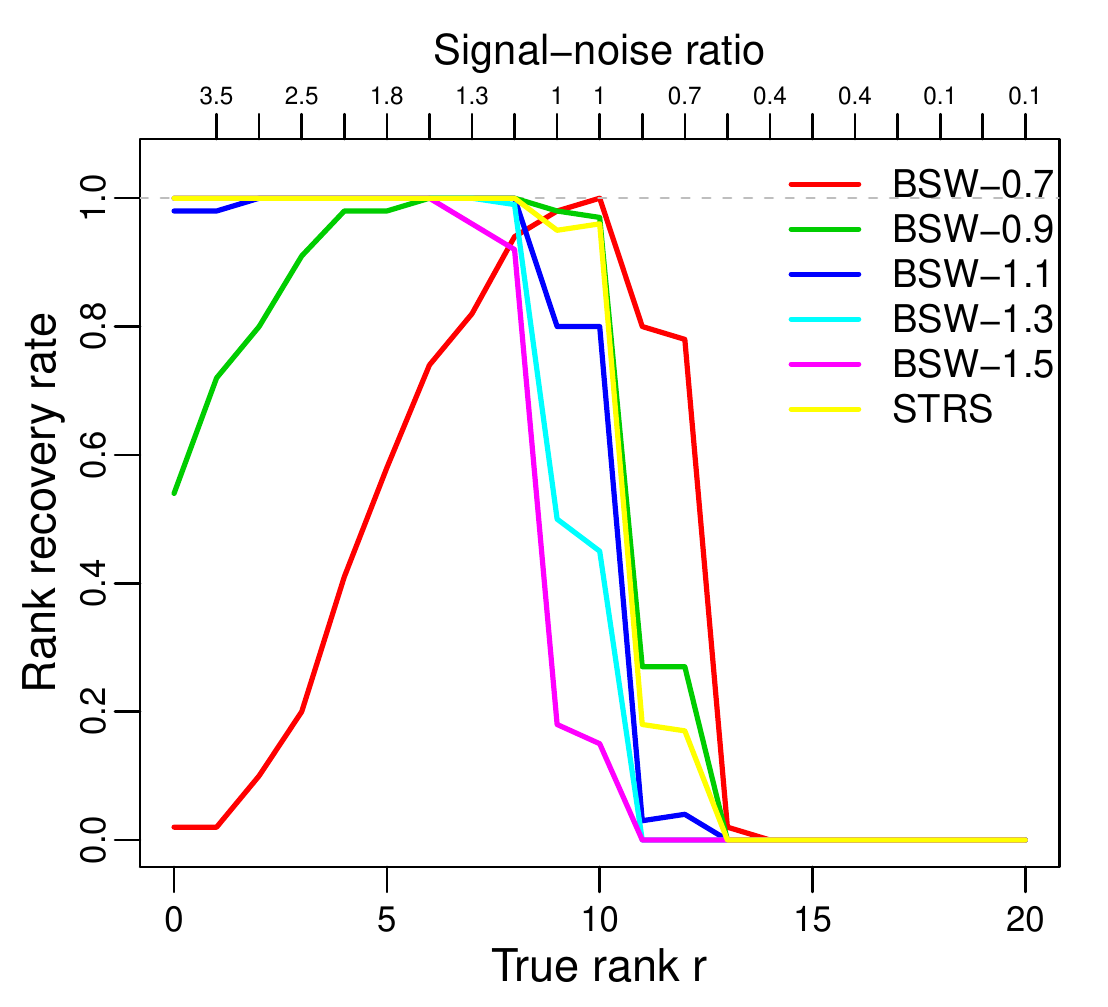}&
			\includegraphics[width=.45\linewidth]{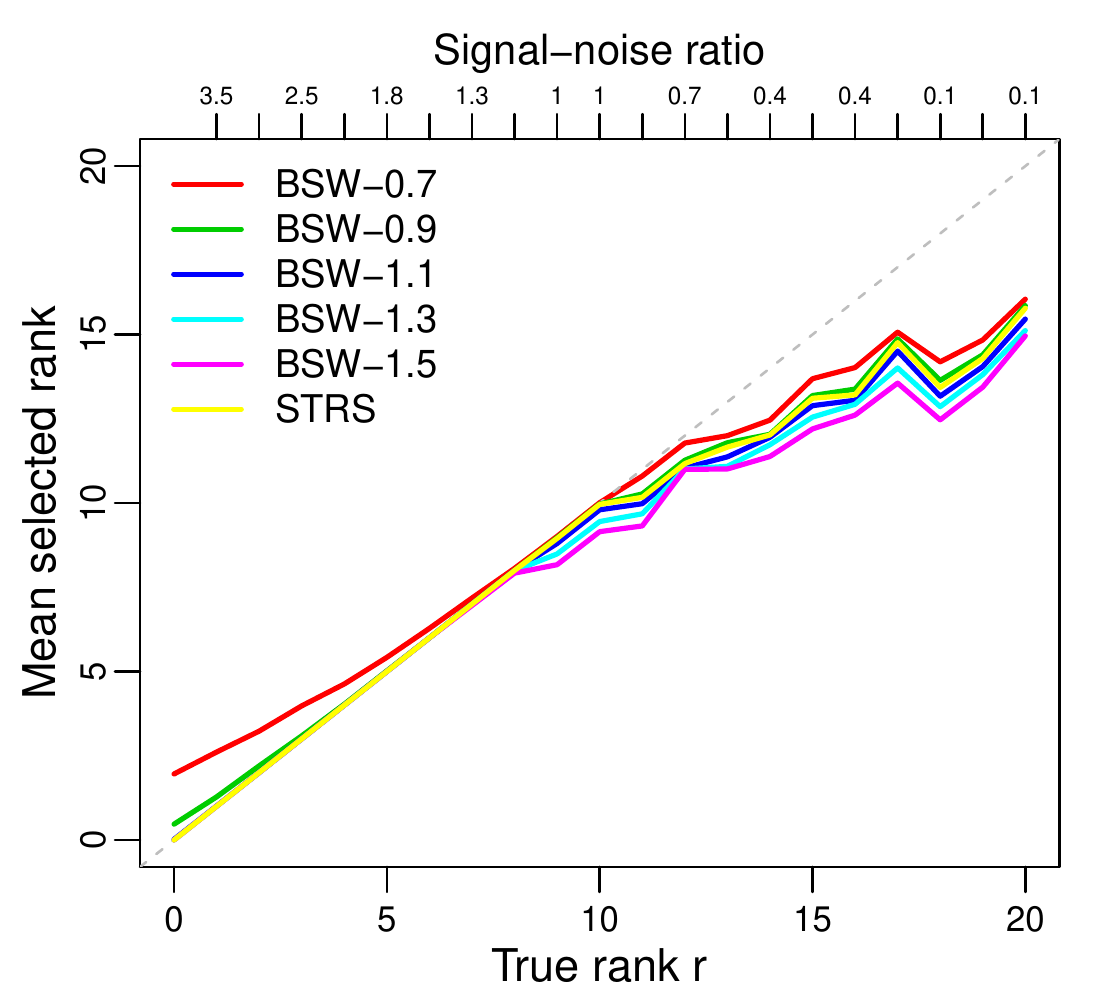}\\
			[-6pt]
			\includegraphics[width=.45\linewidth]{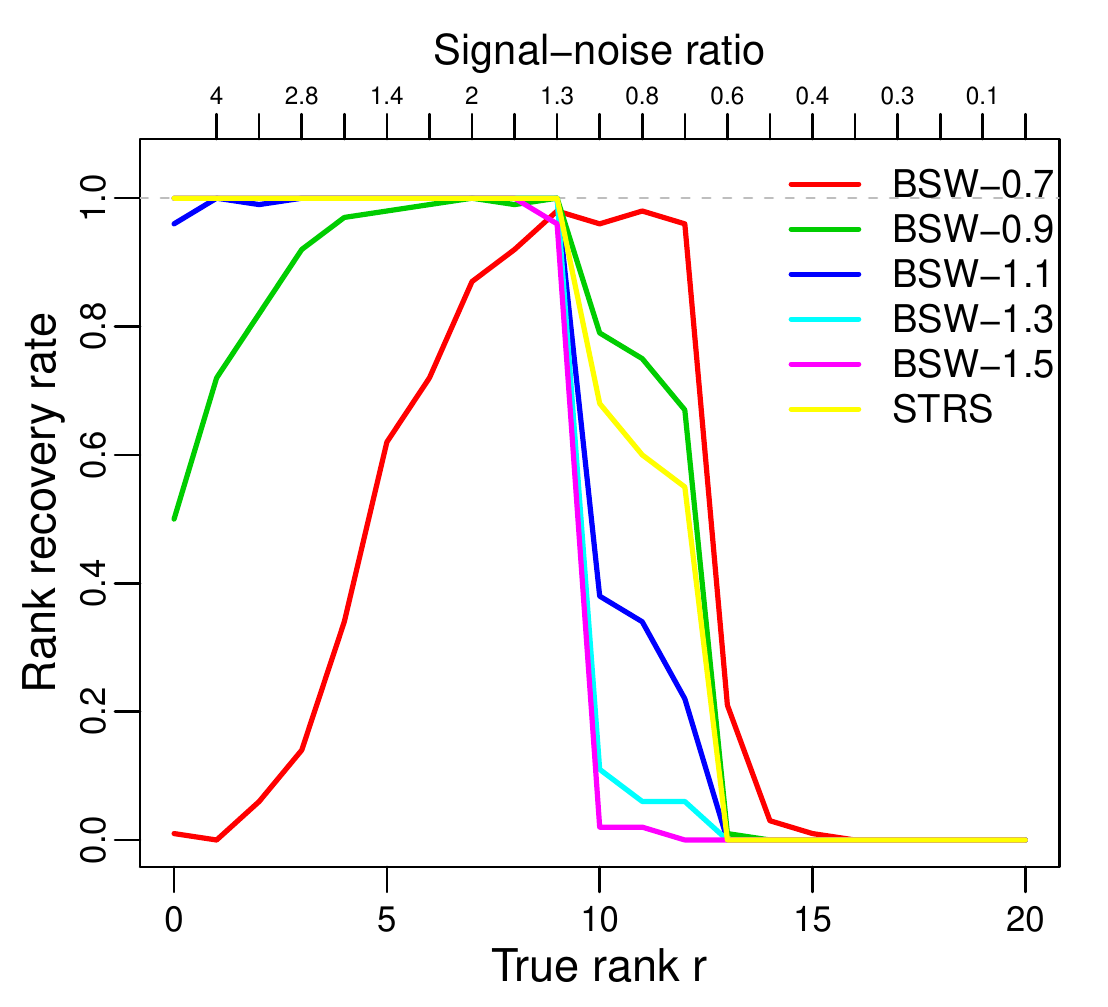} &
			\includegraphics[width=.45\linewidth]{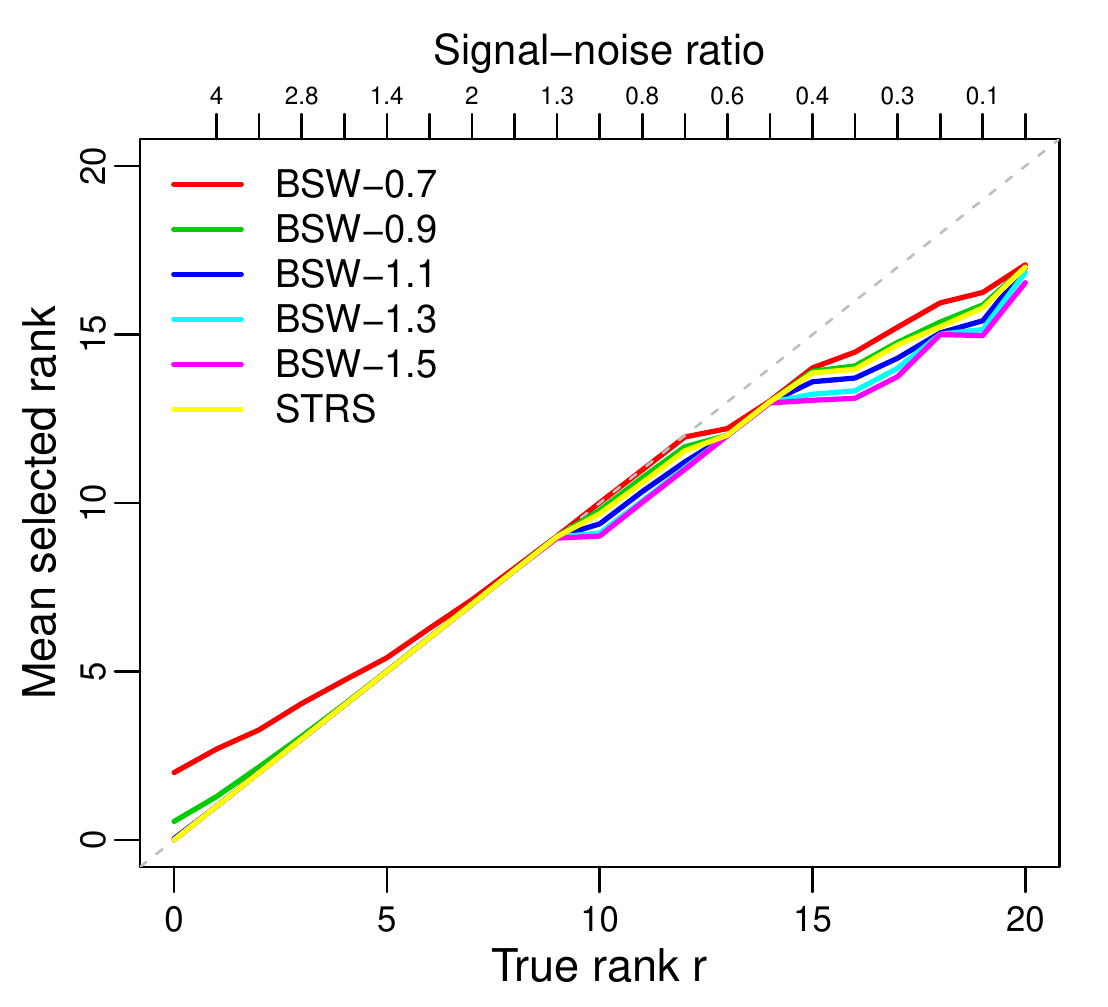}\\
			[-6pt]
			\includegraphics[width=.45\linewidth]{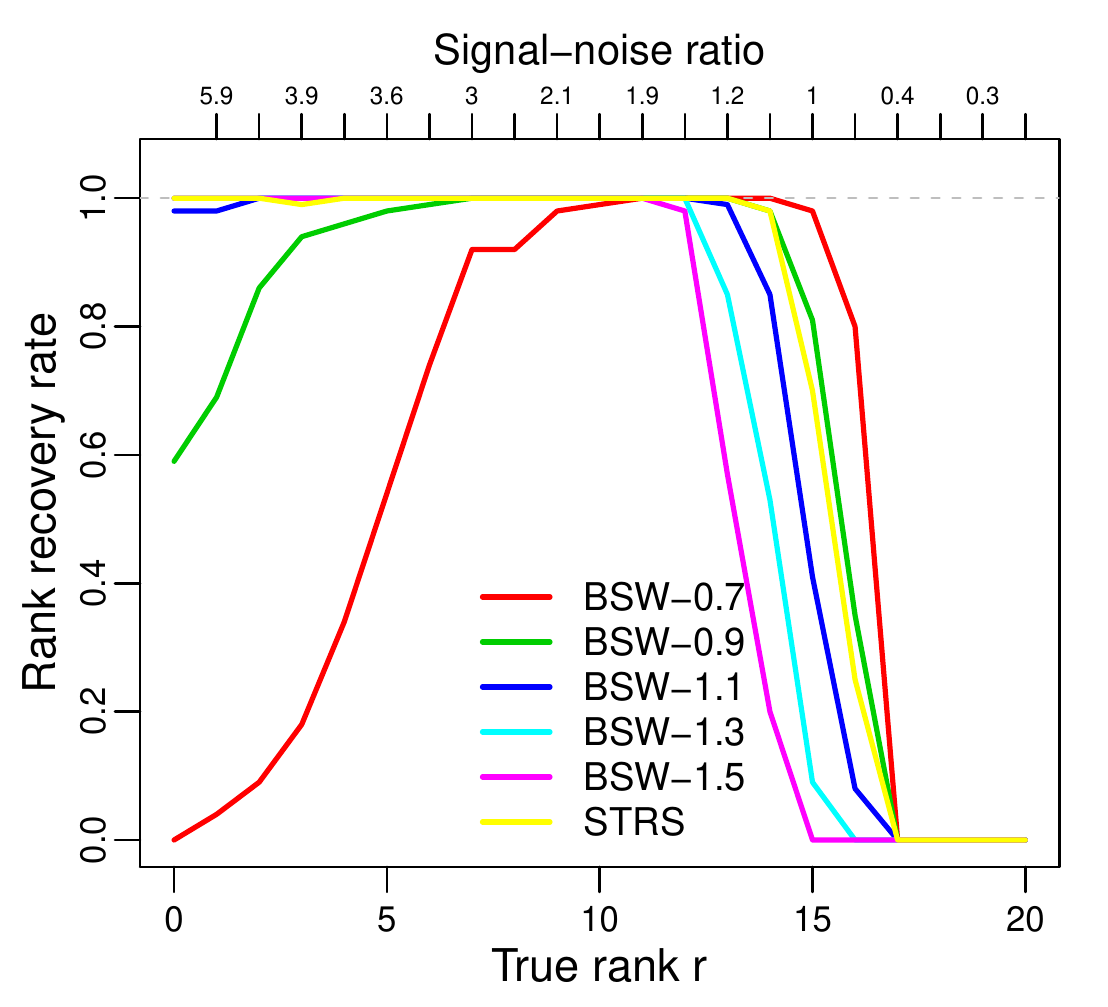} & \includegraphics[width=.45\linewidth]{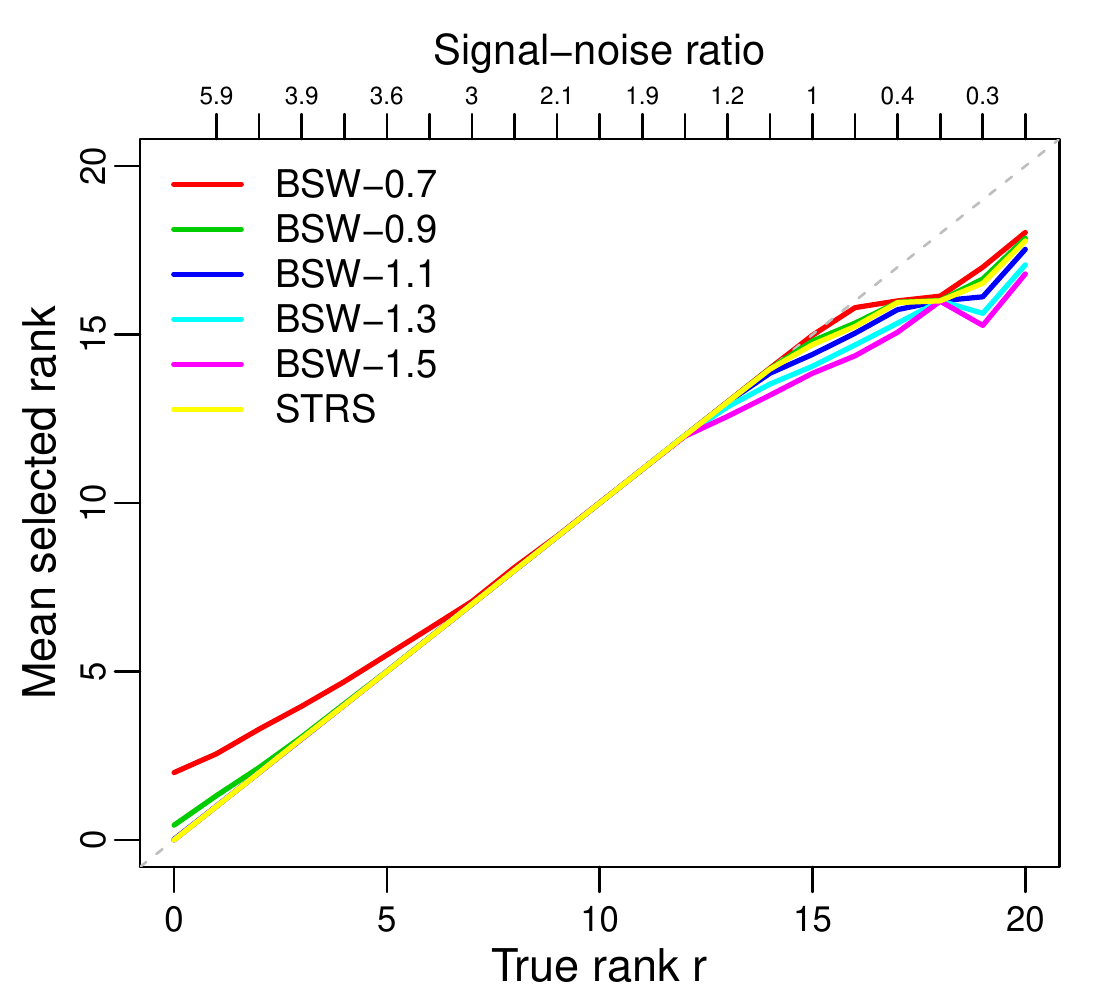}
		\end{tabular}
		\caption{Comparison of {BSW-C} and STRS in the low-dimensional setting of Experiment 1. Here $b_0$ is $0.15$ (top), $0.20$ (middle) and $0.25$ (bottom).}	\label{E1F1}
		\vspace{-3mm}
	\end{figure}

	\begin{figure}[H]
		\centering
		\begin{tabular}{cc}
			\centering
			\includegraphics[width=.45\linewidth]{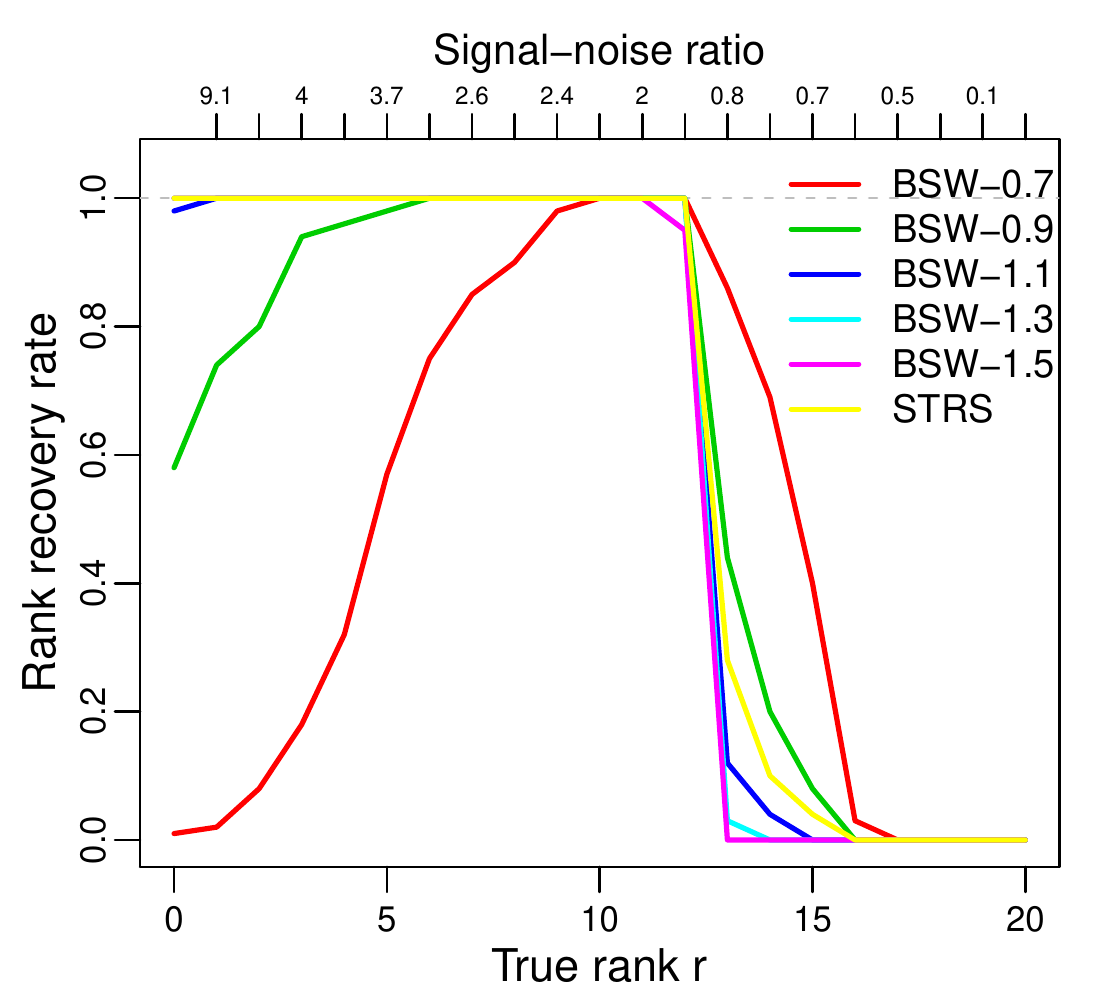} &
			\includegraphics[width=.45\linewidth]{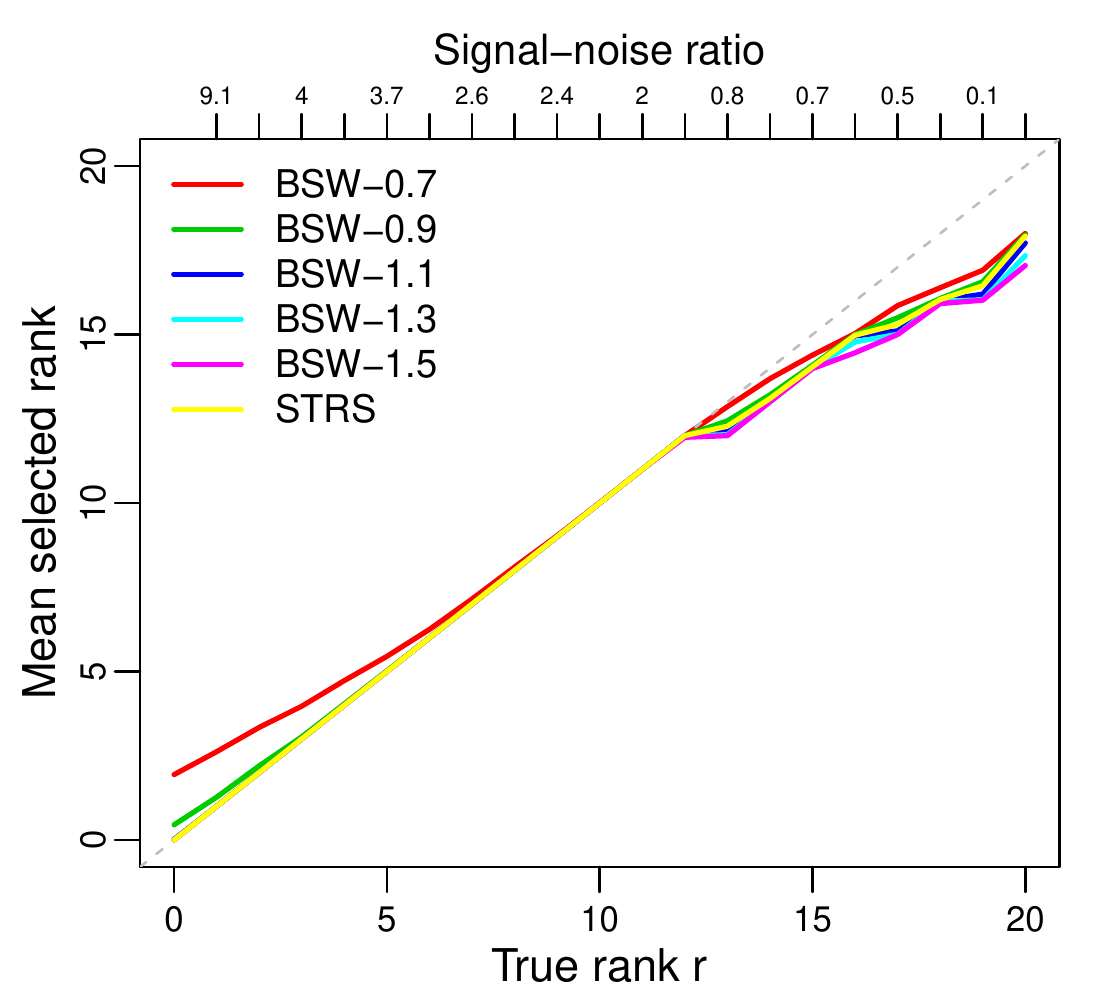}\\
			[-6pt]
			\includegraphics[width=.45\linewidth]{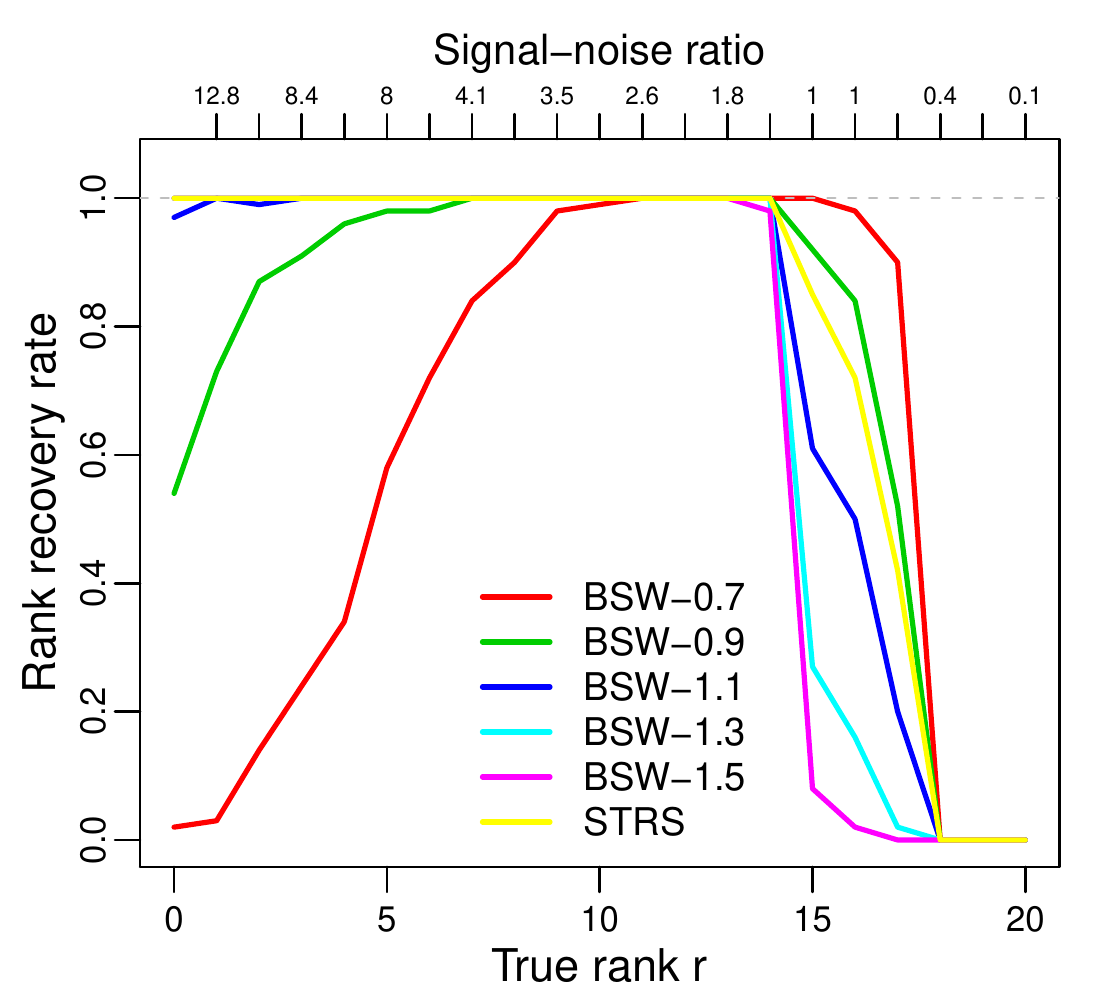} & 
			\includegraphics[width=.45\linewidth]{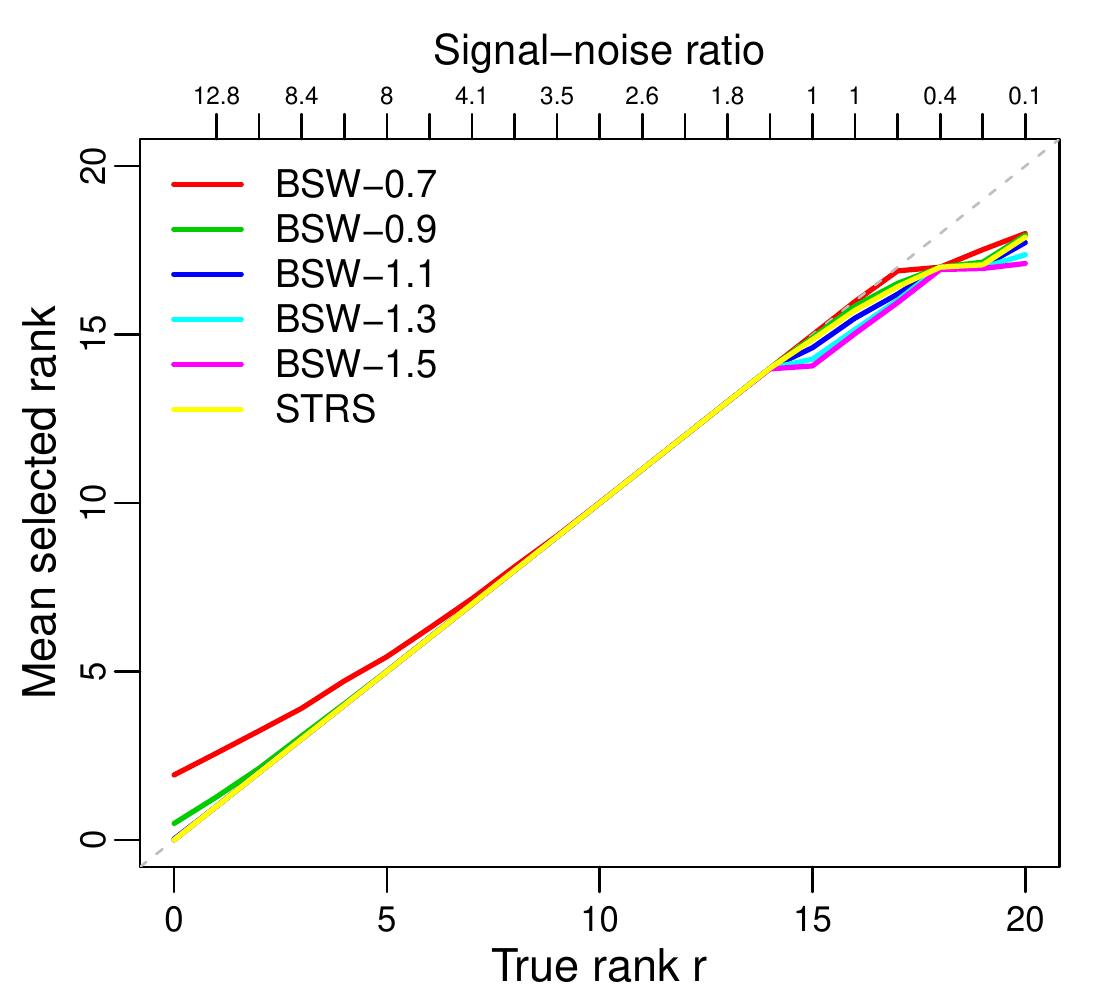}\\
			[-6pt]
			\includegraphics[width=.45\linewidth]{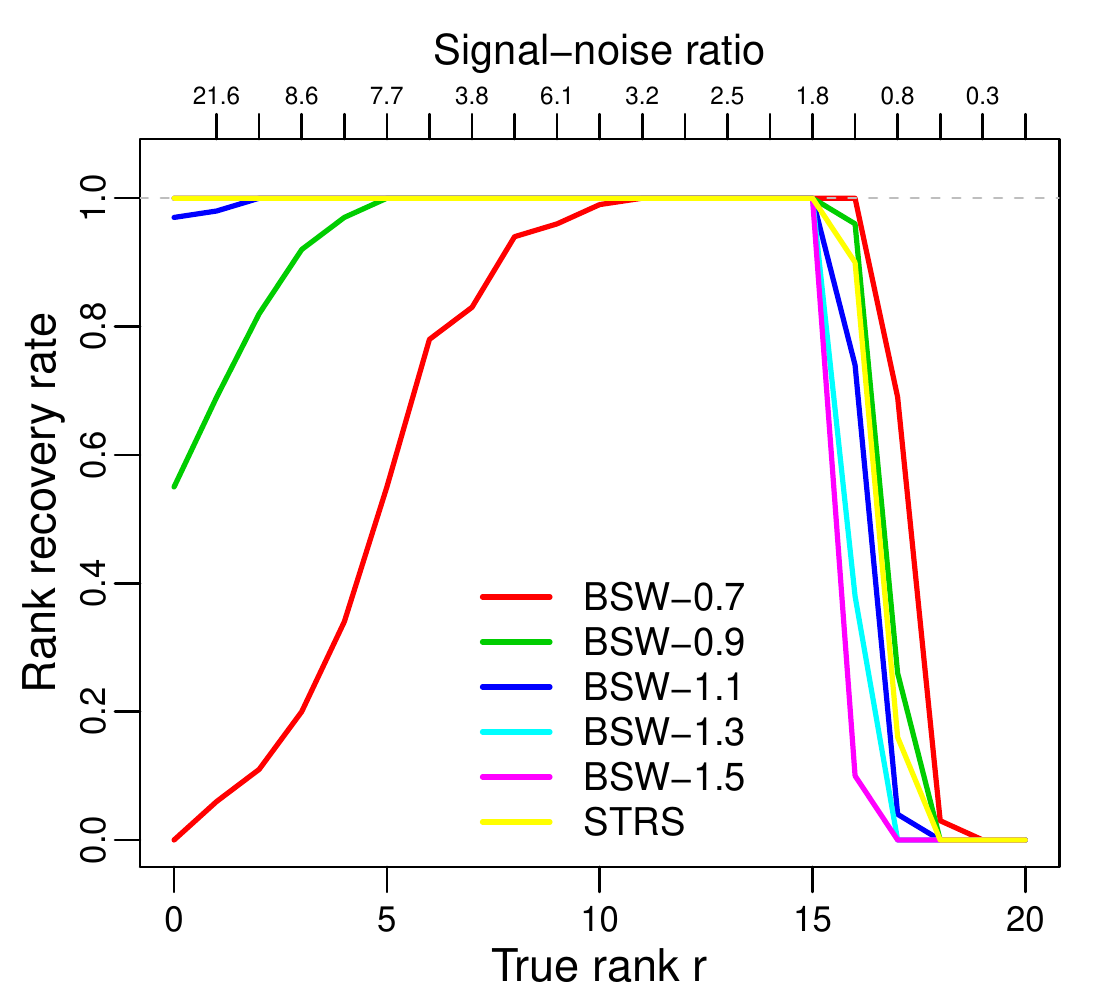} & 
			\includegraphics[width=.45\linewidth]{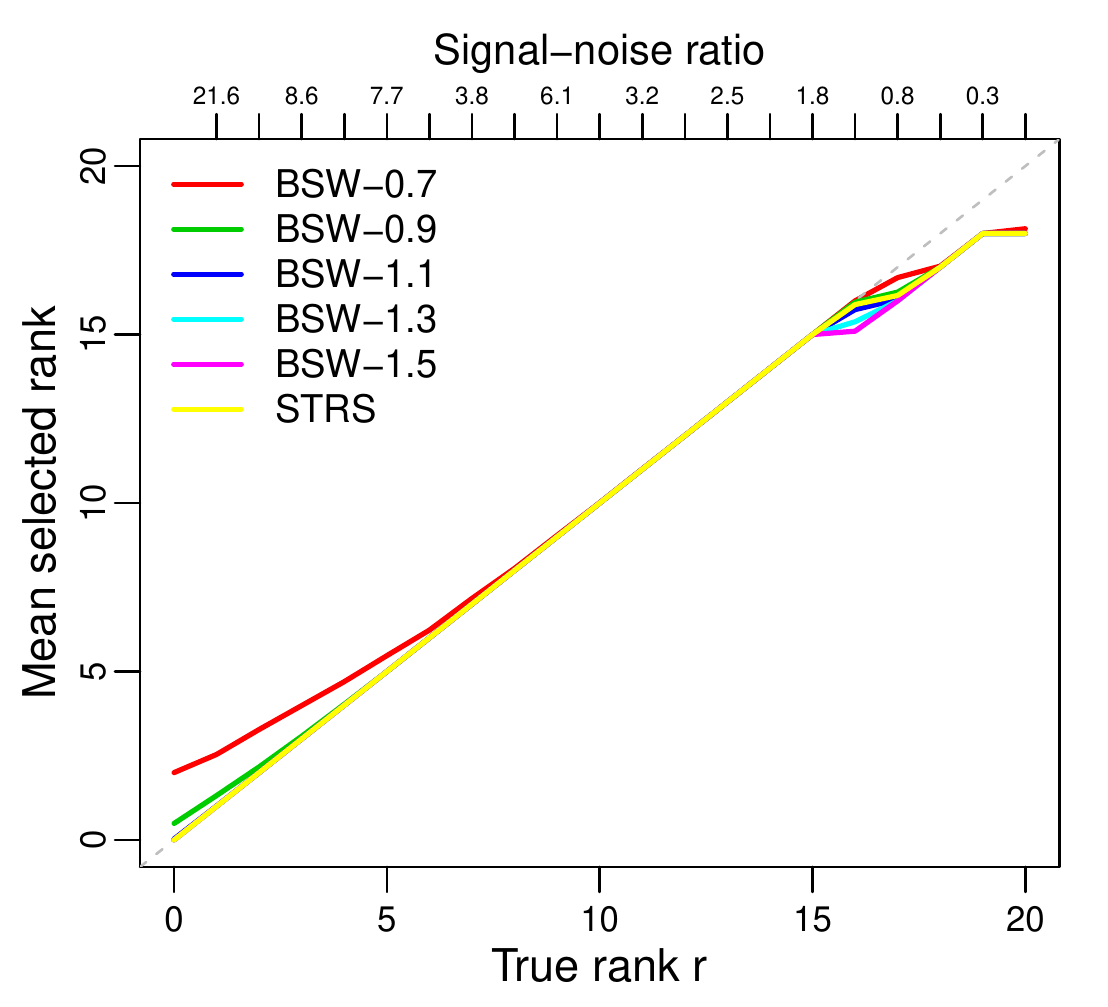}
		\end{tabular}
		\caption{Comparison of {BSW-C} and STRS in the high-dimensional setting of Experiment 1. Here $b_0$ is $0.03$ (top), $0.05$ (middle) and $0.07$ (bottom). }
		\label{E1F2}
	\end{figure}
	
	{\sc Result: } We see that, for  moderate $m$,  {BSW-C} performs  worse as   $n$ gets closer to $q$. Indeed, estimation of $\sigma^2$ is problematic for small values of  $(n-q)m$. The same problem for choosing different $C$ persists:  BSW-1.1 requires a smaller SNR, but overfits more than BSW-1.3 at small $r$, while  STRS performs perfectly as long as $r$ lies in the allowable range (its largest recoverable rank   is between $11$ and $13$, depending on the particular choice of $q$).\\
	
	In addition, we verify the feasibility of STRS when $n=q$. In this case, BSW is infeasible. We consider $p=200$, $b_0 = 0.011$, $\eta = 0.1$,   $n=q\in \{50, 100, 150\}$, $m\in \{50, 125, 200\}$ and $r\in\{0,\ldots,50\}$. In each setting,   the signal  is large enough (SNR$>3$) to eliminate the effect of the signal-to-noise ratio on rank recovery. Figure \ref{E3F2} shows that STRS recovers the rank  for all combinations of $n$  and $m$  as long as $r$ is within the recoverable range (which increases in $m$).
	
	\begin{figure}[H]
		\centering
		\vspace{-3mm}
		\begin{tabular}{cc}
			\centering
			\includegraphics[width=.45\linewidth]{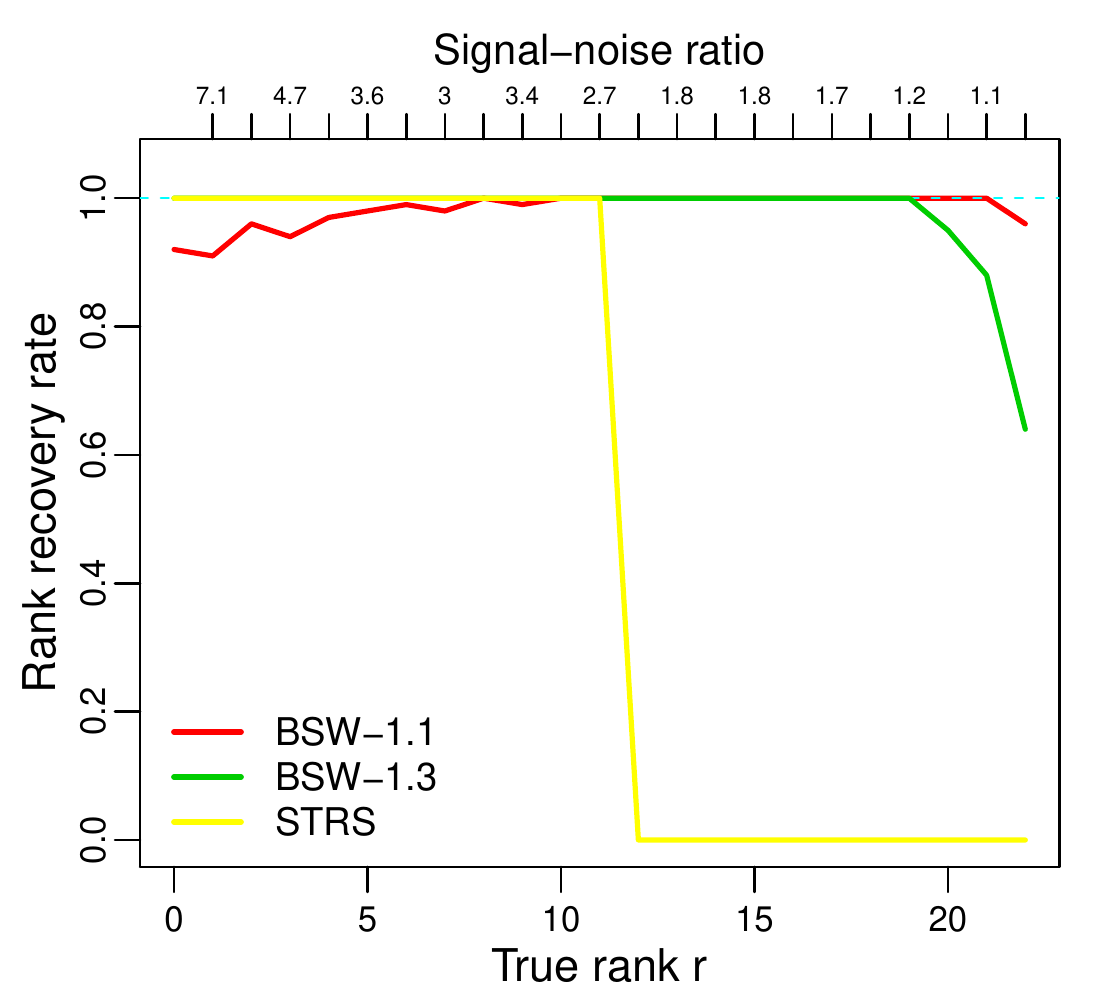} &
			\includegraphics[width=.45\linewidth]{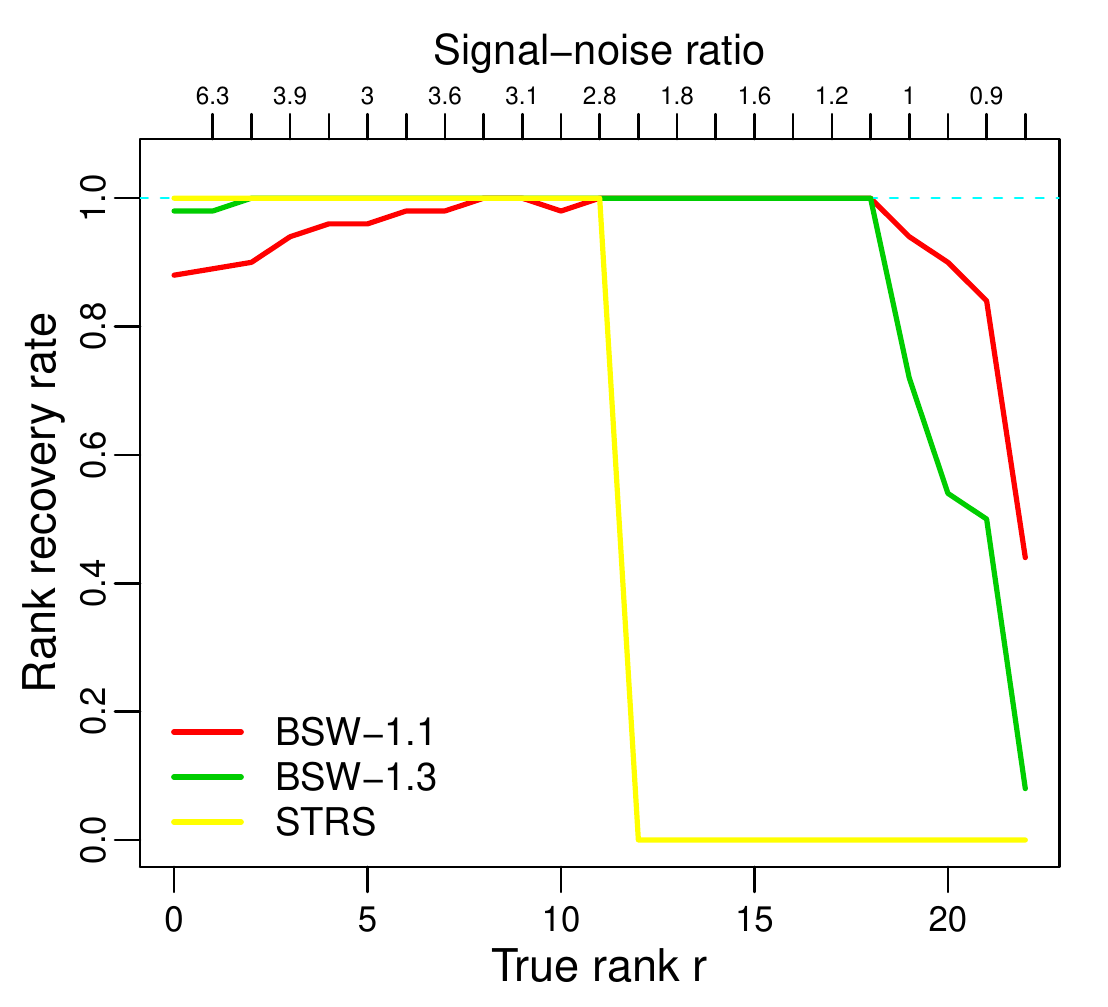}\\
			[-6pt]
			\includegraphics[width=.45\linewidth]{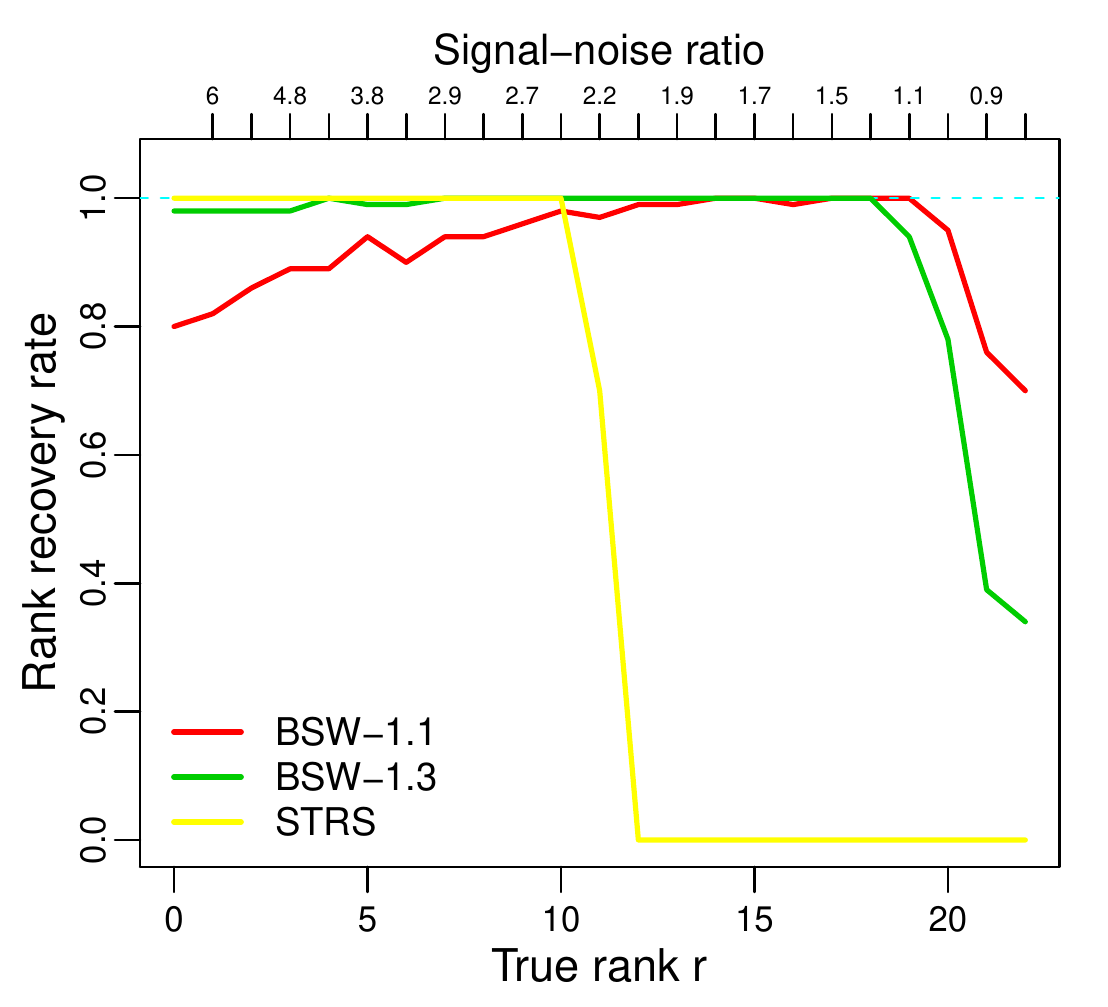} &
			\includegraphics[width=.45\linewidth]{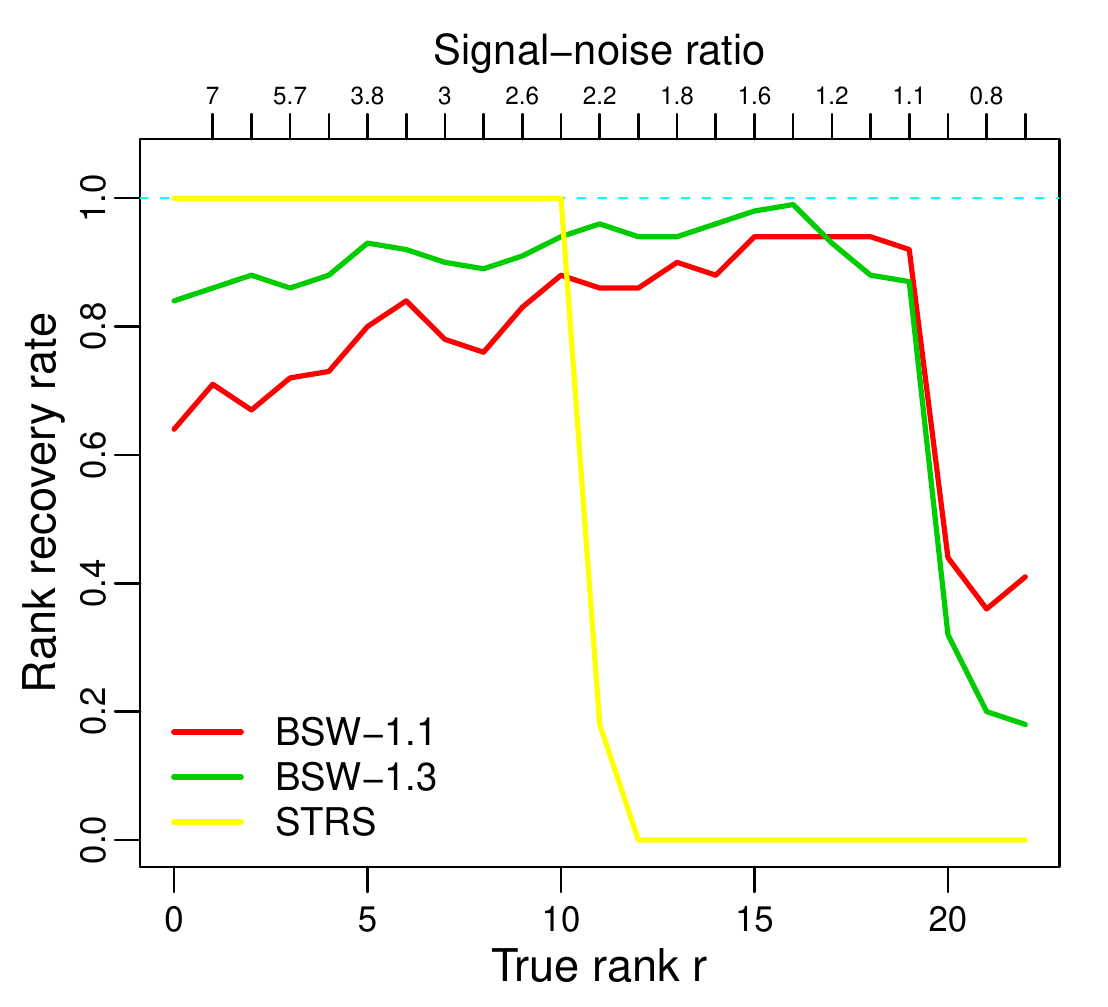}
		\end{tabular}
		\vspace{-1mm}
		\caption{Plots of Experiment 2 on rank recovery rate for BSW-1.1, BSW-1.3 and STRS. The plots  from left to right and top to bottom have   $\{143, 145, 147, 149\}$ for $q$.}
		\label{E3F1}
		\vspace{-3mm}
	\end{figure}
	
	\begin{figure}[ht]
		\vspace{-1mm}
		\centering
		\includegraphics[width=0.45\linewidth]{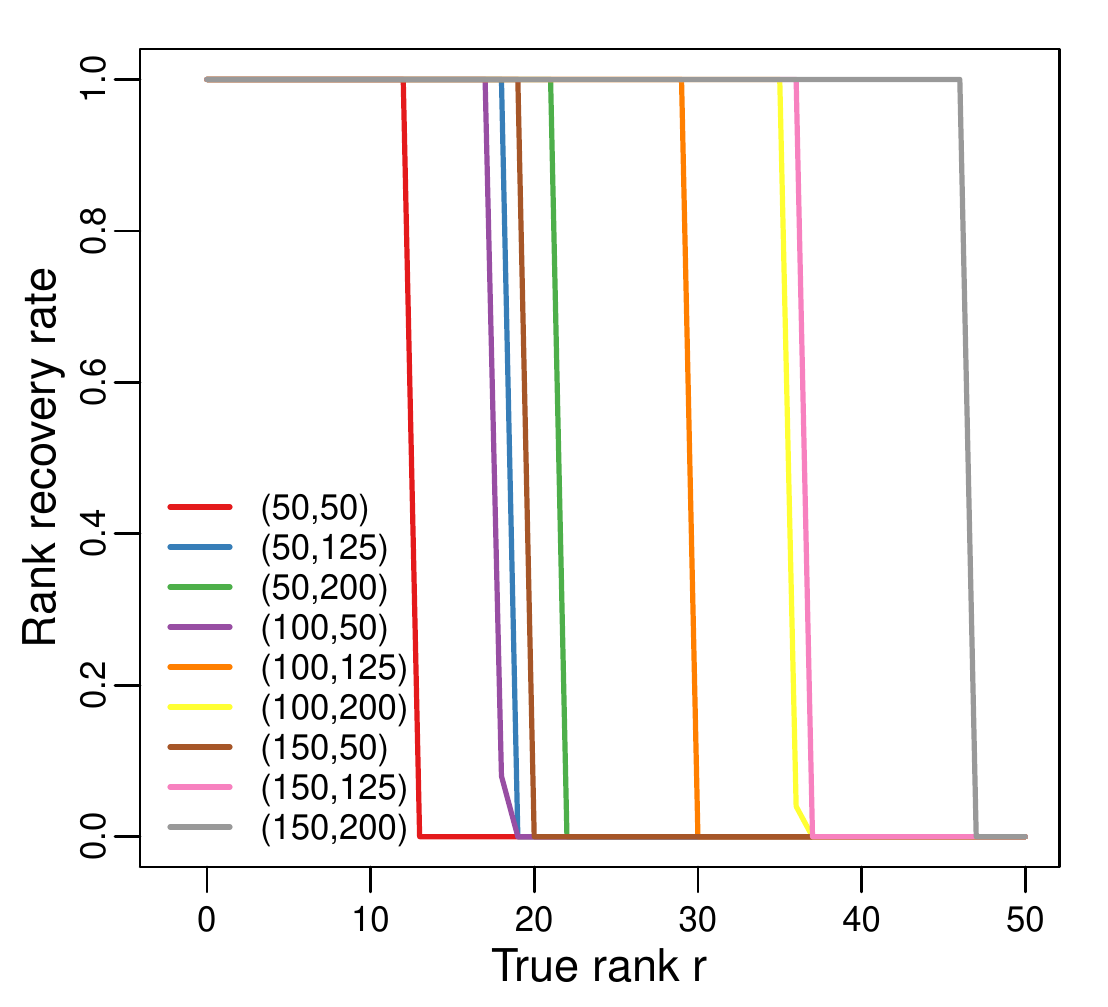}
		\caption{Rank recovery rate for STRS when $n=q$ in Experiment 2.}
		\label{E3F2}
		\vspace{-5mm}
	\end{figure}
	
	\subsection{Experiment 3}\label{sec_sim_3}
	
	Figure \ref{fig_exp2} demonstrates the advantages of STRS over GRS, stated in Theorem \ref{thm: SRS-snr} and Proposition \ref{prop: ex-rank},  in three settings. The first setting is the same low-dimensional scenario   considered in Experiment 1 with $b_0=0.25$. The second setting uses the same high-dimensional setting considered in Experiment 1 with $b_0=0.07$. The third setting focuses on the case when $nm\le \lambda_0N$, which incurs the rank constraint (\ref{rankconstraint}), and  we set
	$n=50$, $m=50$, $p=300$, $q=30$, $\eta=0.1$, $b_0=2$ and $r\in\{0,\ldots,30\}$.  In this setup, $K_{\lambda_0} = 7$. 
	
	\begin{figure}[ht]	
		\centering
		\begin{tabular}{cc}
			\includegraphics[width=.45\linewidth]{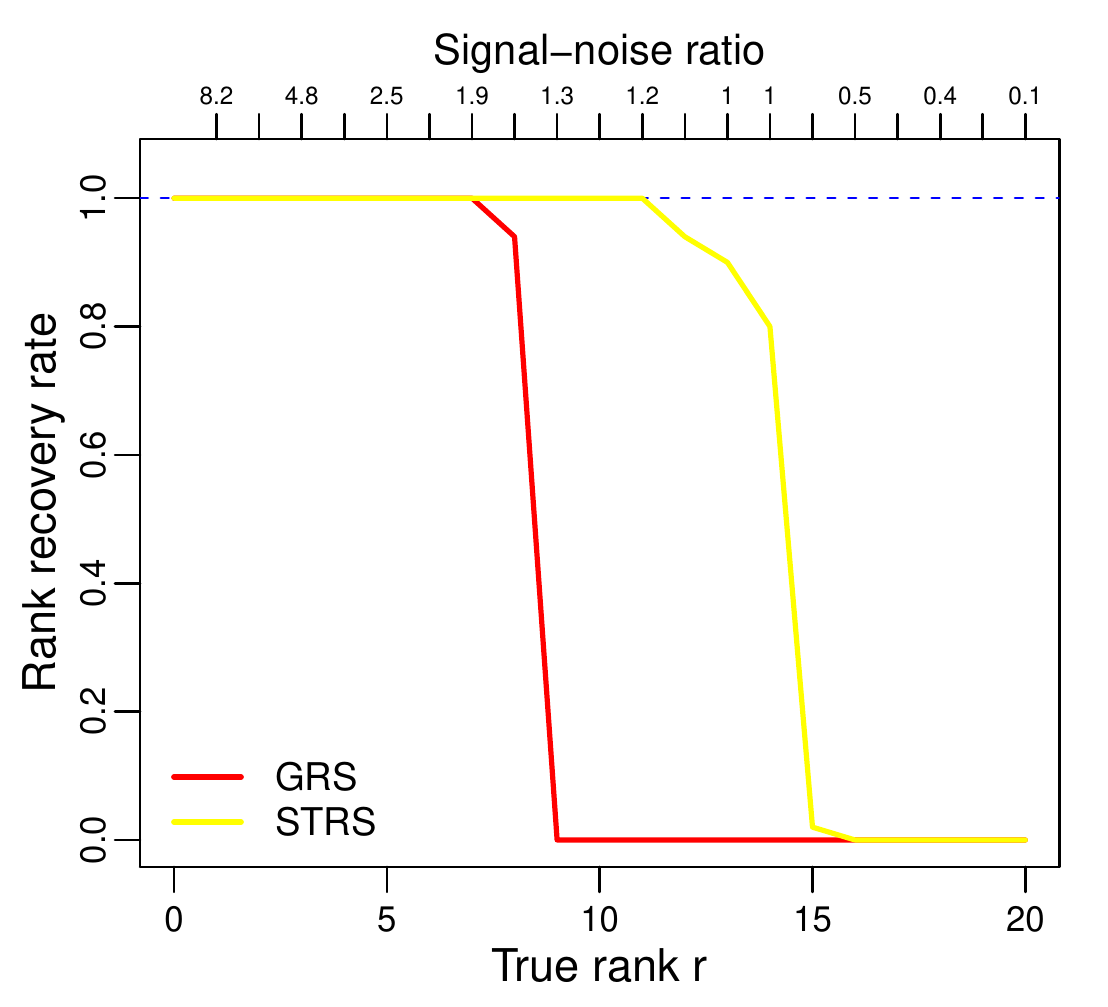} & \includegraphics[width=.45\linewidth]{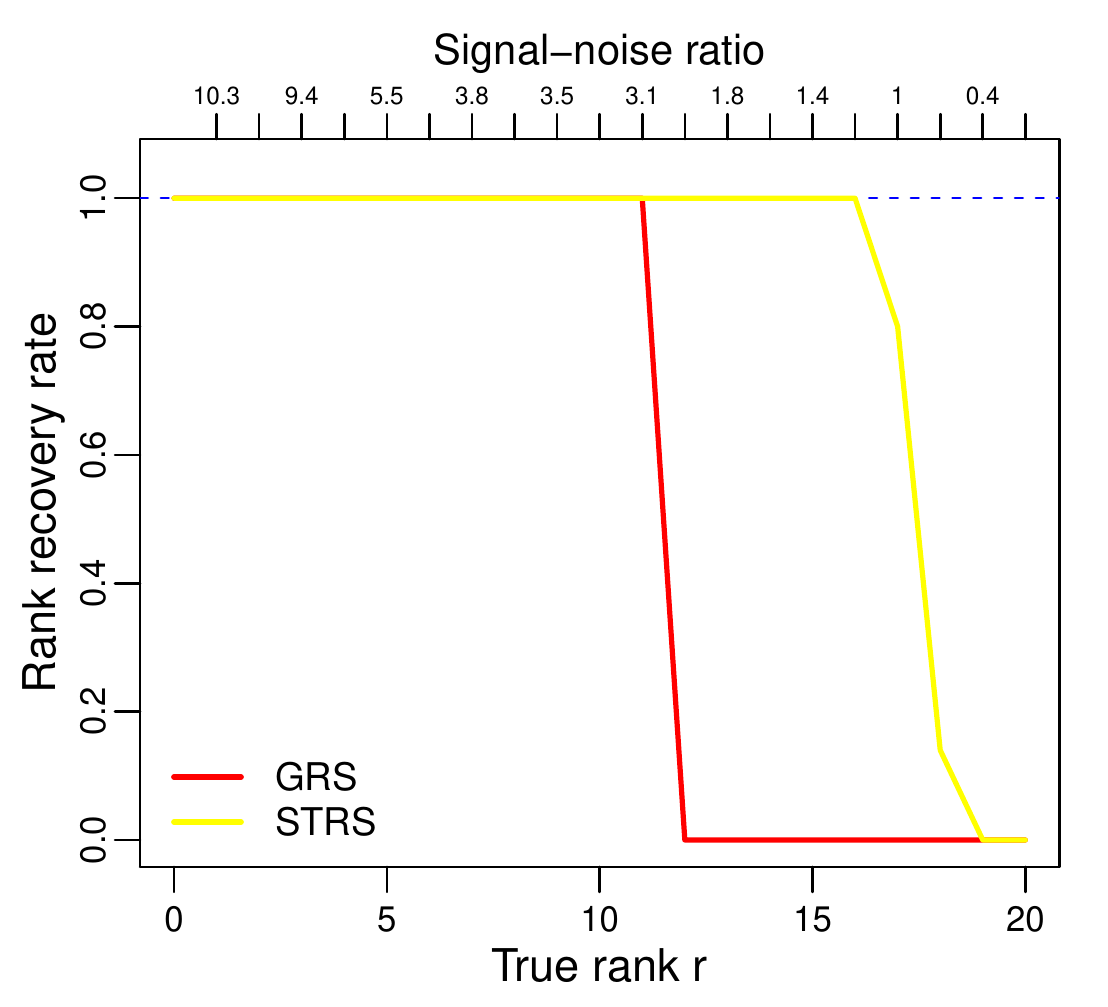}\\
			[-6pt]
			\includegraphics[width=.45\linewidth]{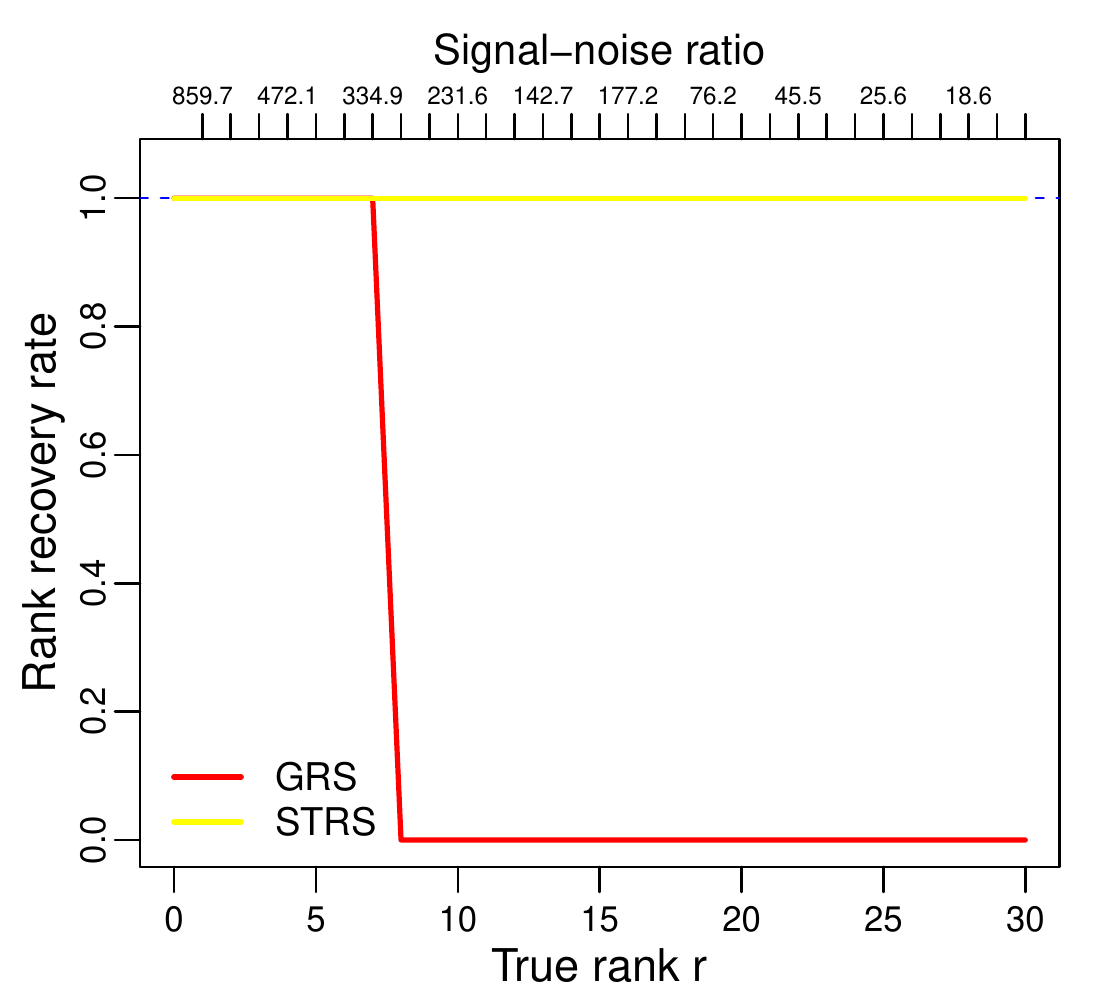}&
			\includegraphics[width=.45\linewidth]{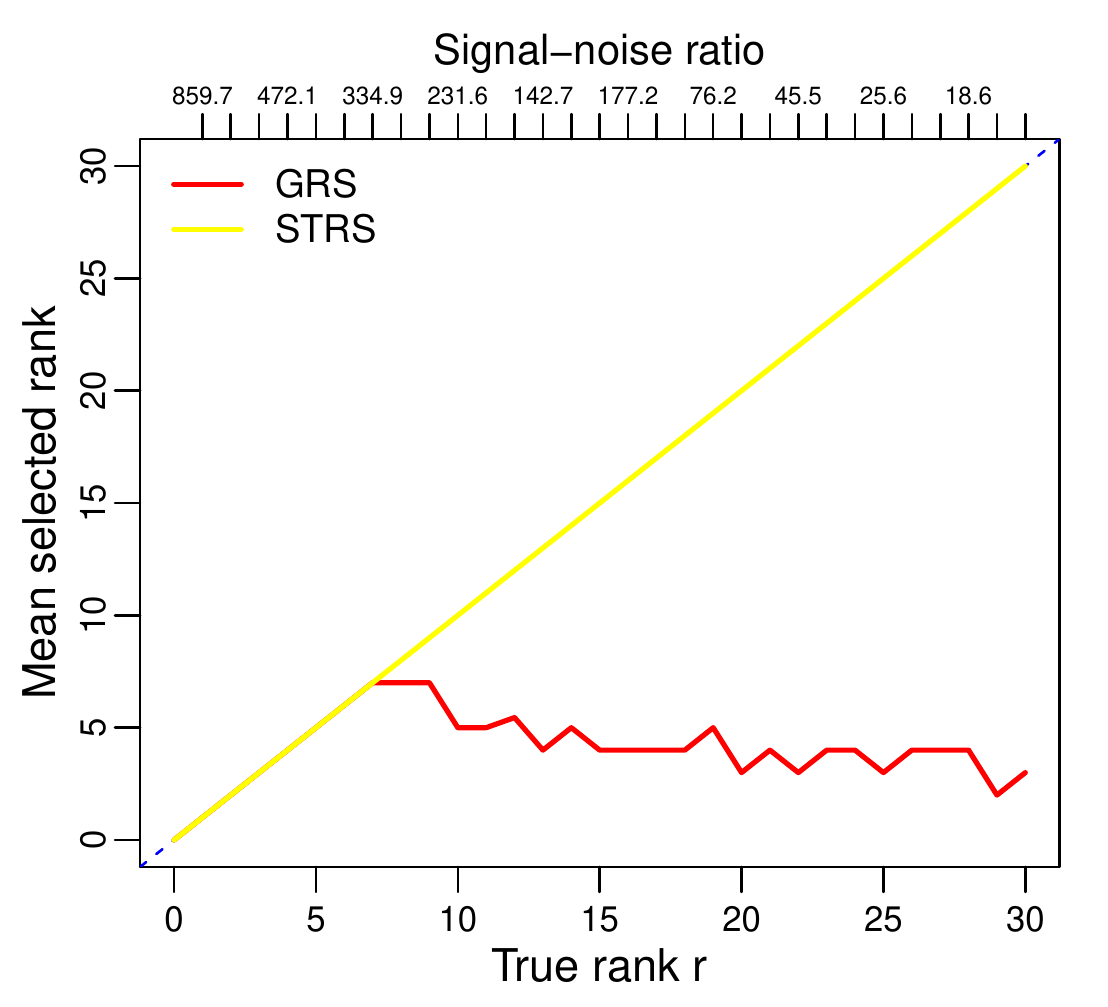}
		\end{tabular}
		\caption{Comparison of GRS and STRS in Experiment 3 in the first setting (top left), second setting (top right) and third setting (bottom). }
		\label{fig_exp2}
	\end{figure}
	\medskip
	
	{{\sc Result}: } The top two figures in Figure \ref{fig_exp2} indicate that  STRS requires a SNR of about 1, while 
	GRS needs a SNR of about $2$ for correct recovery.
	The bottom two plots  in Figure \ref{fig_exp2} show that 
	GRS fails to recover the rank $r$ if  $r>K_{\lambda_0}$, whereas STRS perfectly recovers all possible ranks. This confirms that when $nm$ is not too small compared to $(\sqrt{m}+\sqrt{q})^2$, STRS can get rid of the rank constraint (\ref{rankconstraint}). 
	(The tuning parameter $\lambda_t$  in STRS   reduces from $198$ to $83$ and $66$ in the first two cases, respectively, and from $315$ to $69$ in the third case.)
	These findings confirm our theoretical results in Section \ref{sec_SRS}.

	\subsection{Experiment 4}\label{sec_sim_4}
	We  verify the results of Section \ref{sec_extension} by comparing the performance of GRS, STRS and SSTRS for both models $Y=XA+E$ and $Y=A+E$ with errors  $E_{ij}$   generated from a $t_\nu$-distribution with various   degrees of freedom $\nu$.
	
	For model $Y =XA + E$, we consider the case $n/q\to 1$ with different $m$. We set $\eta = 0.1$ and $r\in\{0,\ldots, 15\}$ for all   settings. The first plot in Figure \ref{fig_exp5} depicts mean selected ranks of GRS and STRS when we further set $n = q = 150$, $m = 100$, $p = 250$,
	and $b_0 = 0.002$ and $\nu=6$ (degrees of freedom of the $t_\nu$ distribution). We also verify the rank consistency of SSTRS by generating $E_{ij}$ from $t_\nu$-distributions with $\nu\in\{ 6, 8,10\}$.
	The second row in Figure \ref{fig_exp5} depicts  mean selected ranks of SSTRS and  is based on 
	$n = 300\approx q=280>> m=50$, $p = 400$  and $b_0 = 0.0015$.   
	The third row in Figure \ref{fig_exp5} shows the same quantities and is based 
	on $n = 80$, $q = 60$, $p = 150$, $m = 400$ and $b_0 = 0.003$.  We  varied the closeness of $n$ and $q$, but since  the results didn't change,  we only report for one pair of $n$ and $q$ for each setting.

	For model $Y = A + E$, we   present two cases of skinny $A$ when $n=O(m^\alpha)$ and $m=O(n^\alpha)$ for some $\alpha \in (0,1)$.
	Specifically, 
	we consider $n = 500$, $m = 80$ in the first setting and $n = 80$, $m=500$ in the second one. We set $\eta = 0.1$, $b_0 = 0.25$, $r\in\{0,\ldots,20\}$ and $\nu\in\{6,8,10\}$ in both cases. The mean selected ranks of SSTRS are plotted in the last row of Figure \ref{fig_exp5}.
	\medskip
	
	{\sc Result: } In both models, all three procedures work perfectly for heavy tailed errors under very mild SNR, although   the rank constraint (\ref{rankconstraint}) impacts GRS. In addition, STRS and SSTRS
	can handle a larger range of $r$ under a milder SNR and their performance seems very stable under various error distributions.

	\begin{figure}[H]
		\centering
		\vspace{-1mm}
		\includegraphics[width=.42\textwidth]{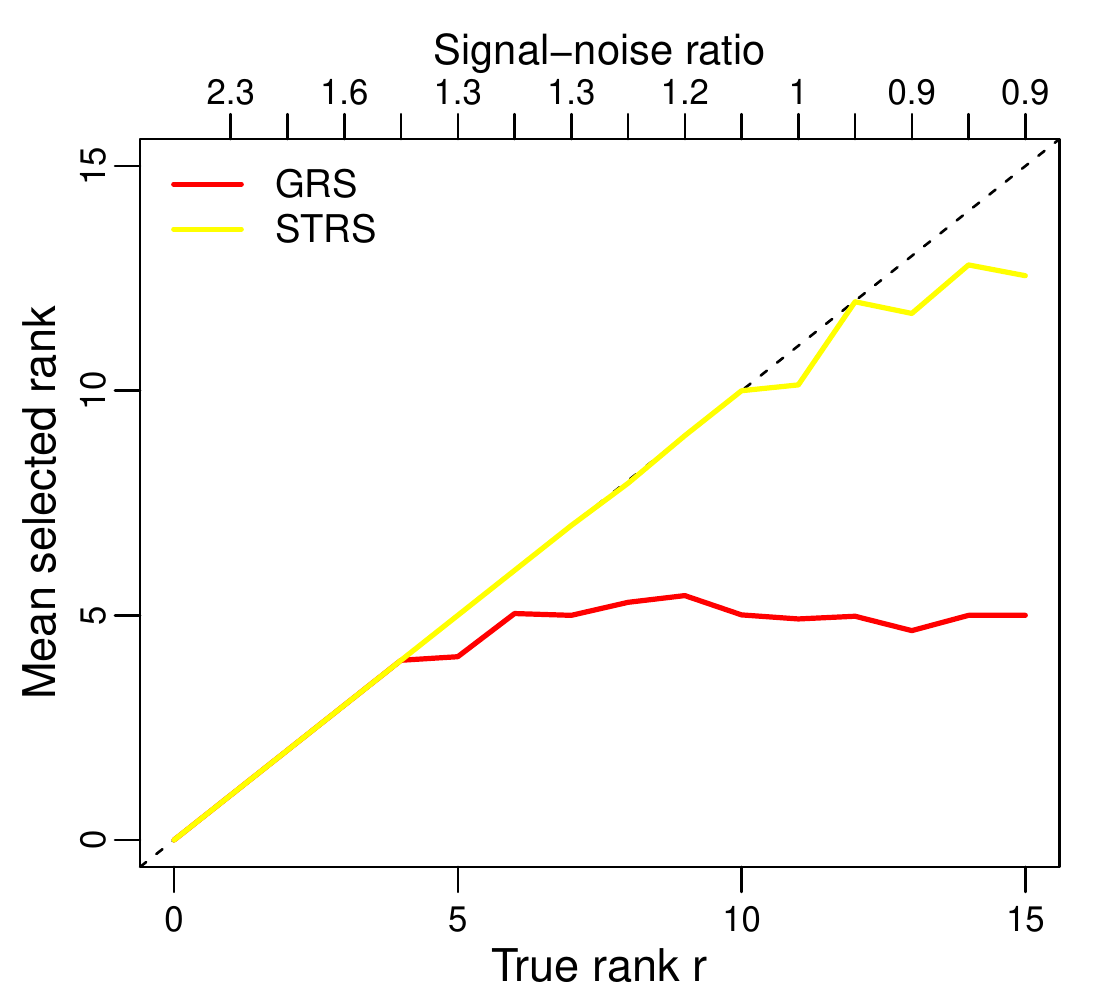}\\
		[-4pt]
		\begin{tabular}{cc}
			\includegraphics[width=.42\textwidth]{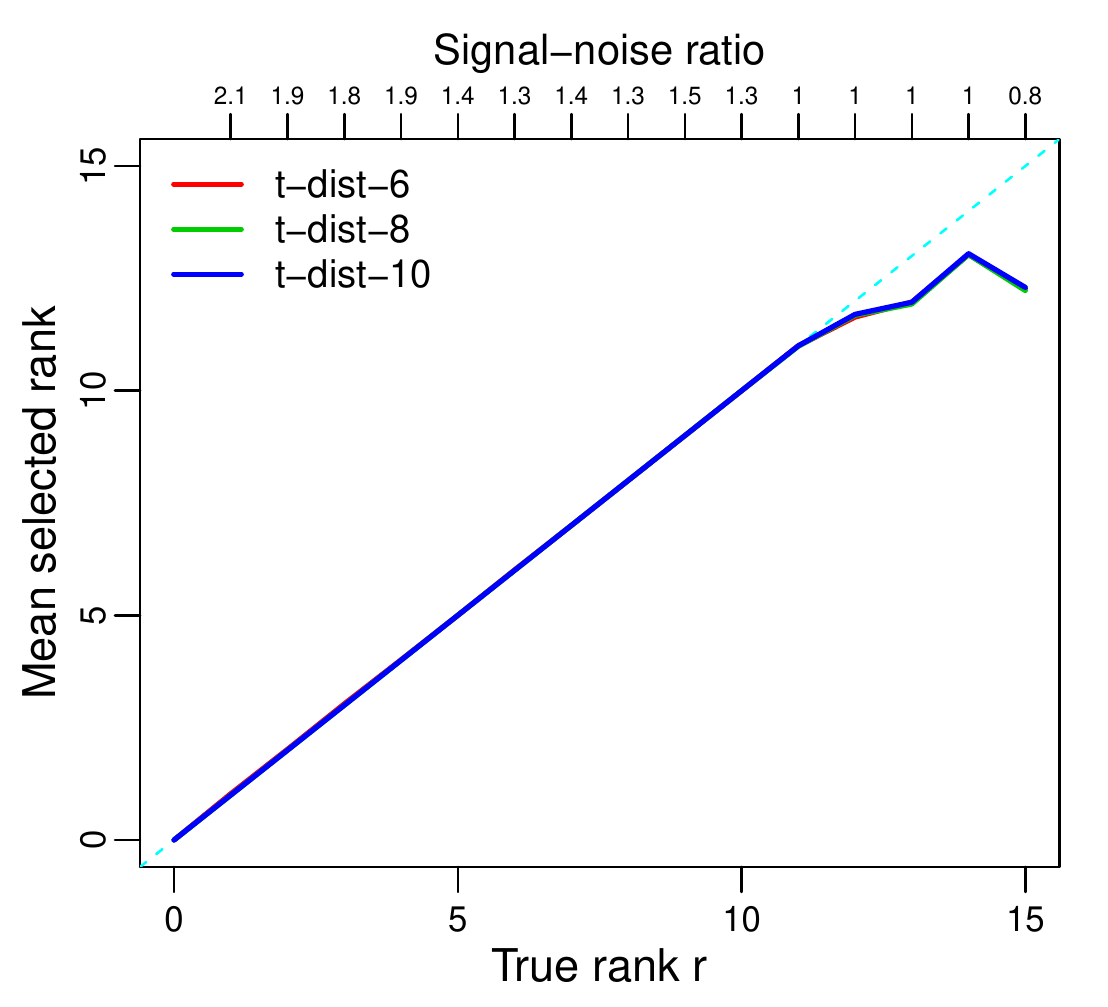} & 
			\includegraphics[width=.42\textwidth]{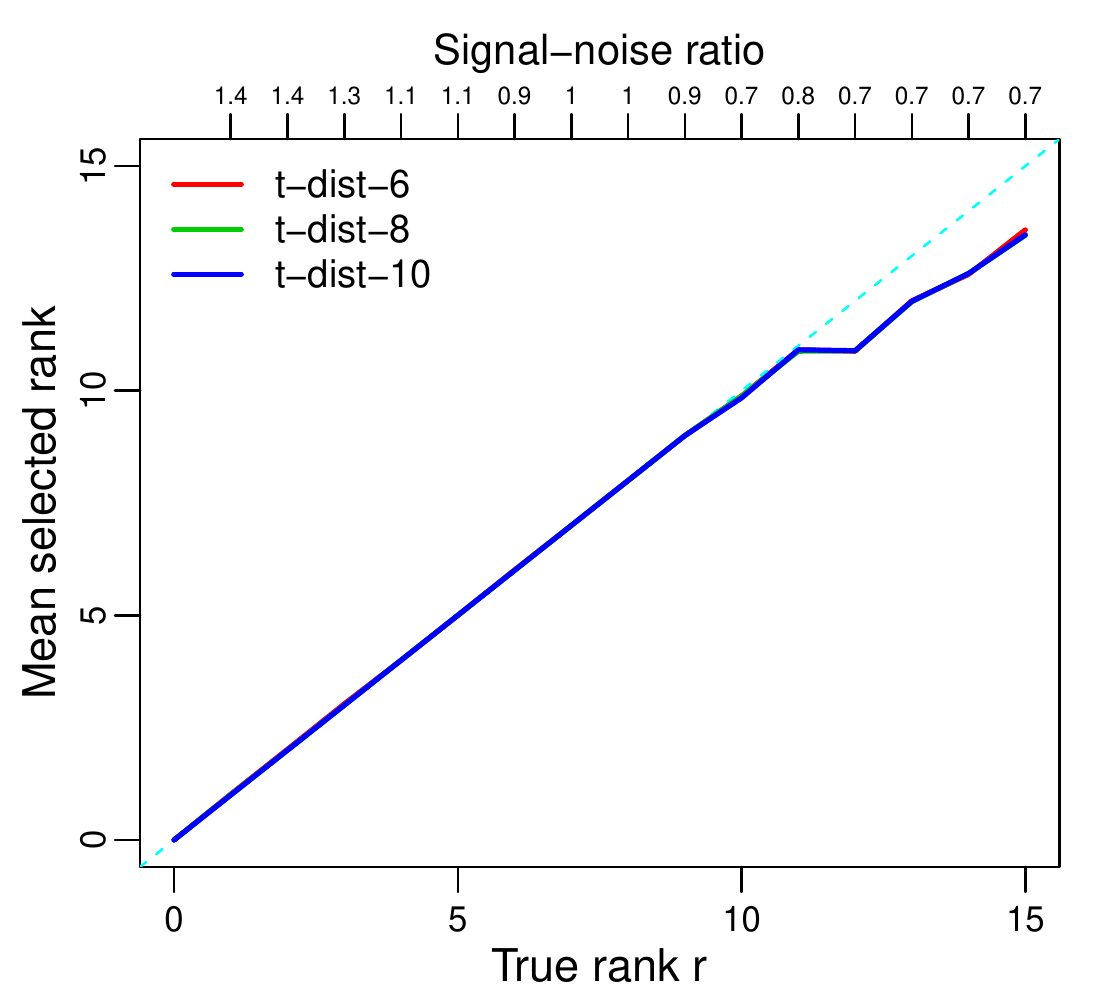}\\
			[-6pt]
			\includegraphics[width=.42\textwidth]{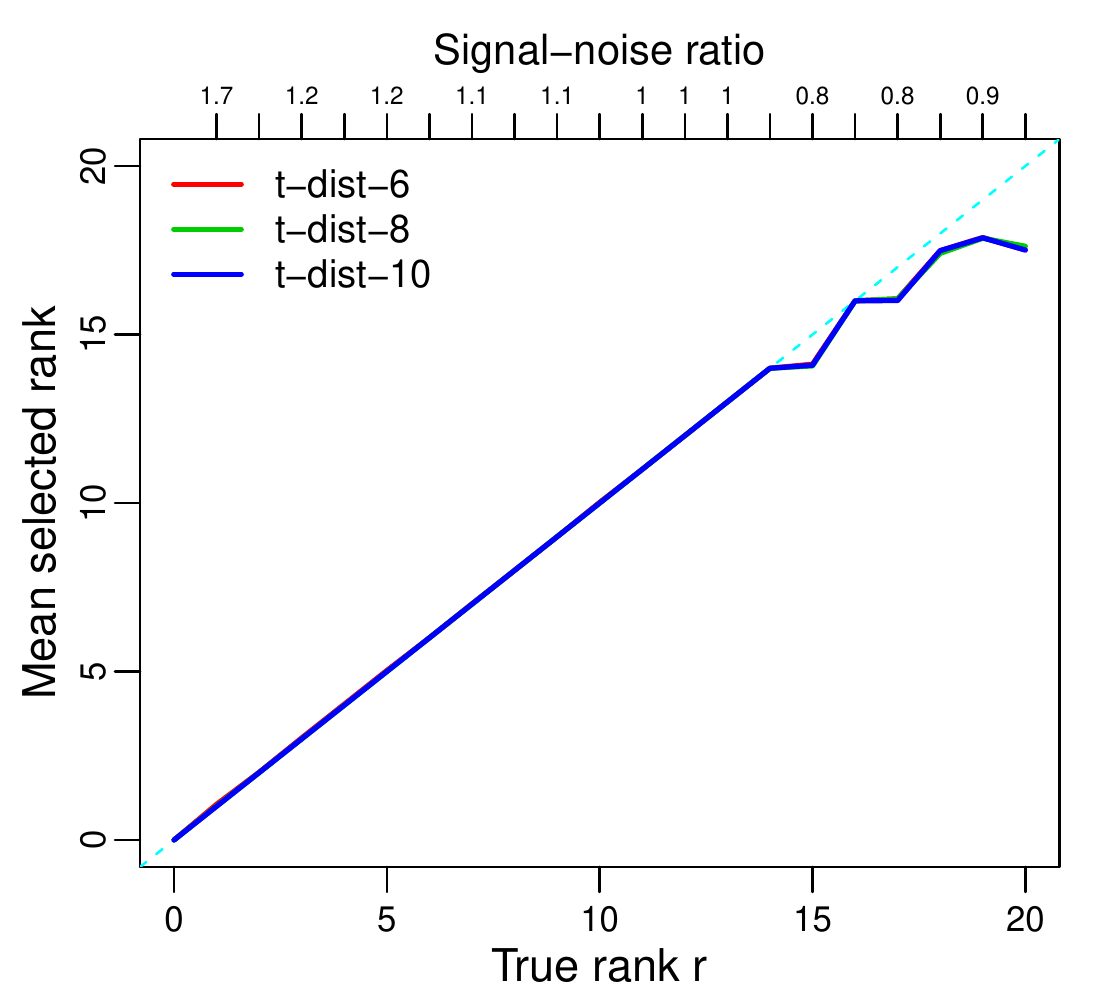} & 
			\includegraphics[width=.42\textwidth]{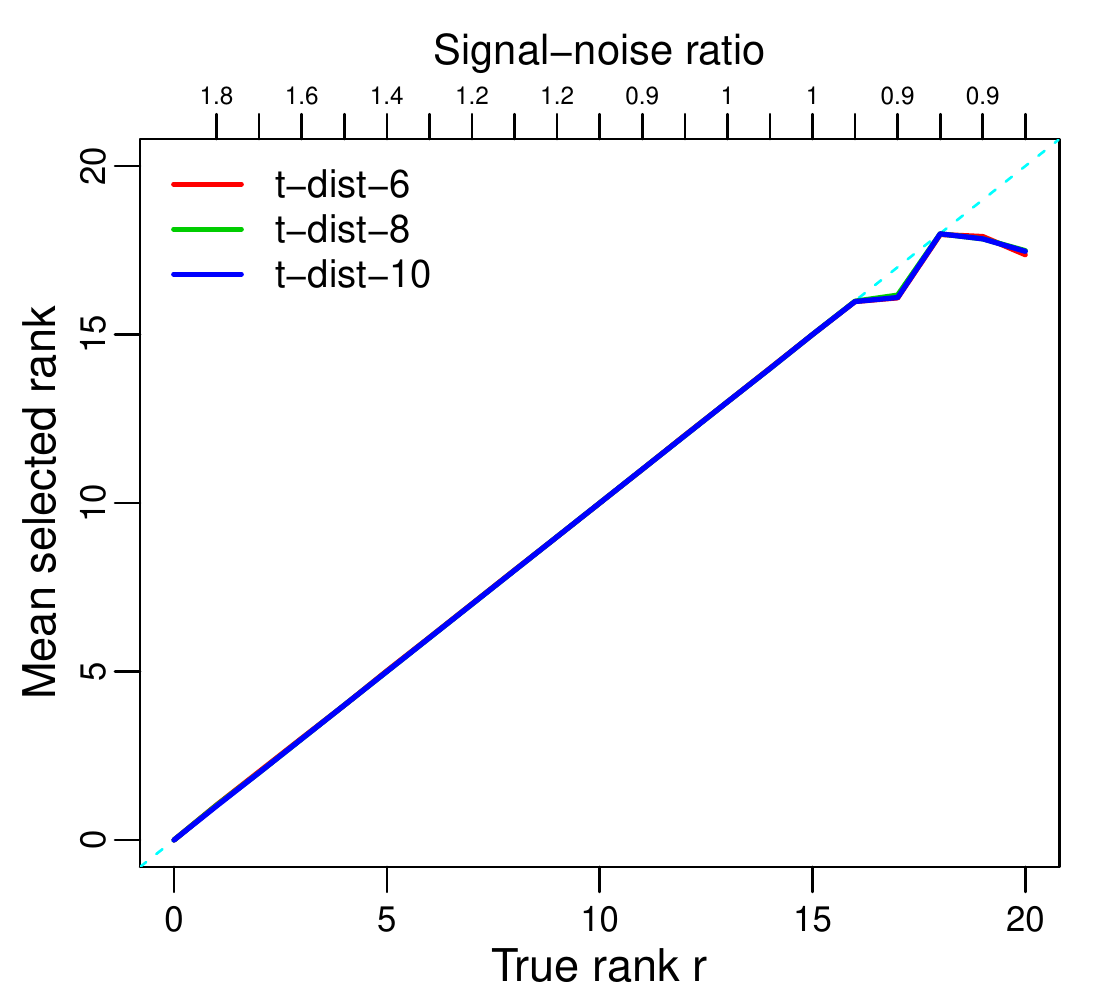}
		\end{tabular}
		\caption{Plots of mean selected ranks related to Experiment 4. The first plot compares GRS and STRS in model $Y=XA+E$. 
			The middle row evaluates SSTRS in model $Y=XA+E$ for various error distributions with
			$n = 300$, $q = 280$, $m = 50$ (left) and $n = 80$, $q = 60$, $m = 400$ (right). The bottom row plots mean selected ranks of SSTRS in model $Y = A + E$ with $n = 500$, $m = 80$ (left) and  in $m=80$, $n=500$ (right).}
		\label{fig_exp5}
	\end{figure}

	\subsection{Experiment 5}\label{sec_sim_5}
	The stable performance of GRS, STRS and SSTRS leads us to make the following conjecture: 
	\begin{equation}\label{eq_conj}
	\EE[d_j(PE)]\ \approx\ \EE[d_j(Z)],\qquad \text{ for all }j = 1,\ldots, q\wedge m
	\end{equation}
	where $Z\in \RR^{q\times m}$ has i.i.d. $N(0,1)$, $P$ is the projection matrix based on $X$ with $\text{rank}(P) = q$ and entries of $E\in \RR^{n\times m}$ are i.i.d. mean zero random variables with $\EE[E_{ij}^2] = 1$ and $\EE [E_{ij}^4] <\infty$.
	The result is striking since the projection $P$ destroys the independence of $E_{ij}$, hence one would not necessarily expect the Bai-Yin law \citep{Bai-Yin}   continue to  hold for $PE$ which only has independent columns. 
	Proving (\ref{eq_conj}) is beyond the scope of the current paper and we leave it for future research.
	Instead, we   verify this conjecture in simulations  for  two cases: (1) $n = 150$, $p = 250$, $q = 50$, $m = 50$; (2) $n=50$, $p=40$, $q=40$, $m=150$. In both cases,  $\eta\in \{0.1, 0.3, 0.5, 0.7, 0.9\}$ and we generate $E$ from $t_\nu$-distributions with degrees of freedom $\nu\in\{5,8,12\}$.
	For each setting, we generate $X$ and $P$ for a given $\eta$,  and we generate 100 pairs of matrices $E$ and   $Z$. Averaged ratios of $d_j(PE)/d_j(Z)$ are calculated for each $j$ and  Figure \ref{fig_exp6_1} shows that the ratios of $d_j(PE)/d_j(Z)$ are highly concentrated around 1.   We only report the case of $\eta = 0.9$ as the other cases gave essentially the same picture.
	
	\begin{figure}[ht]
		\centering
		\begin{tabular}{cc}
			\includegraphics[width=0.42\textwidth]{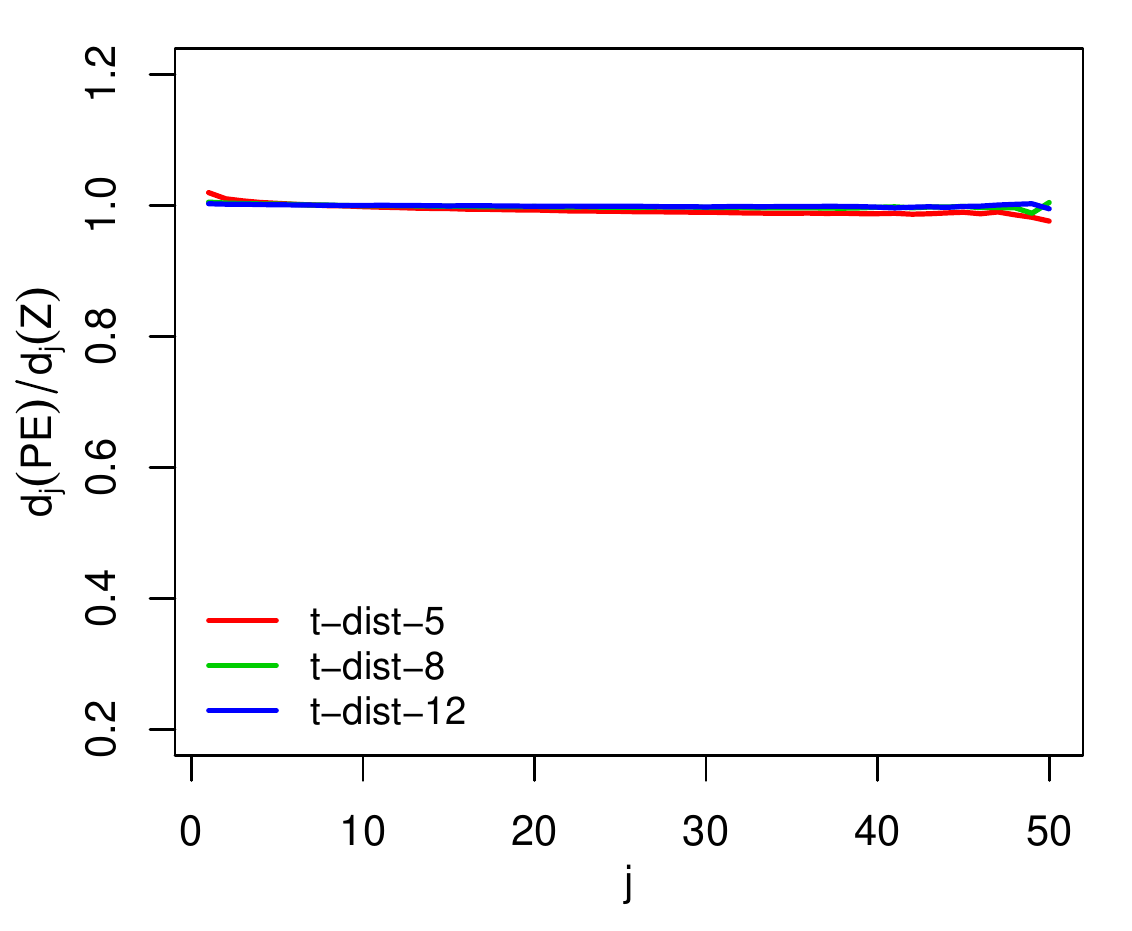}
			\includegraphics[width=0.42\textwidth]{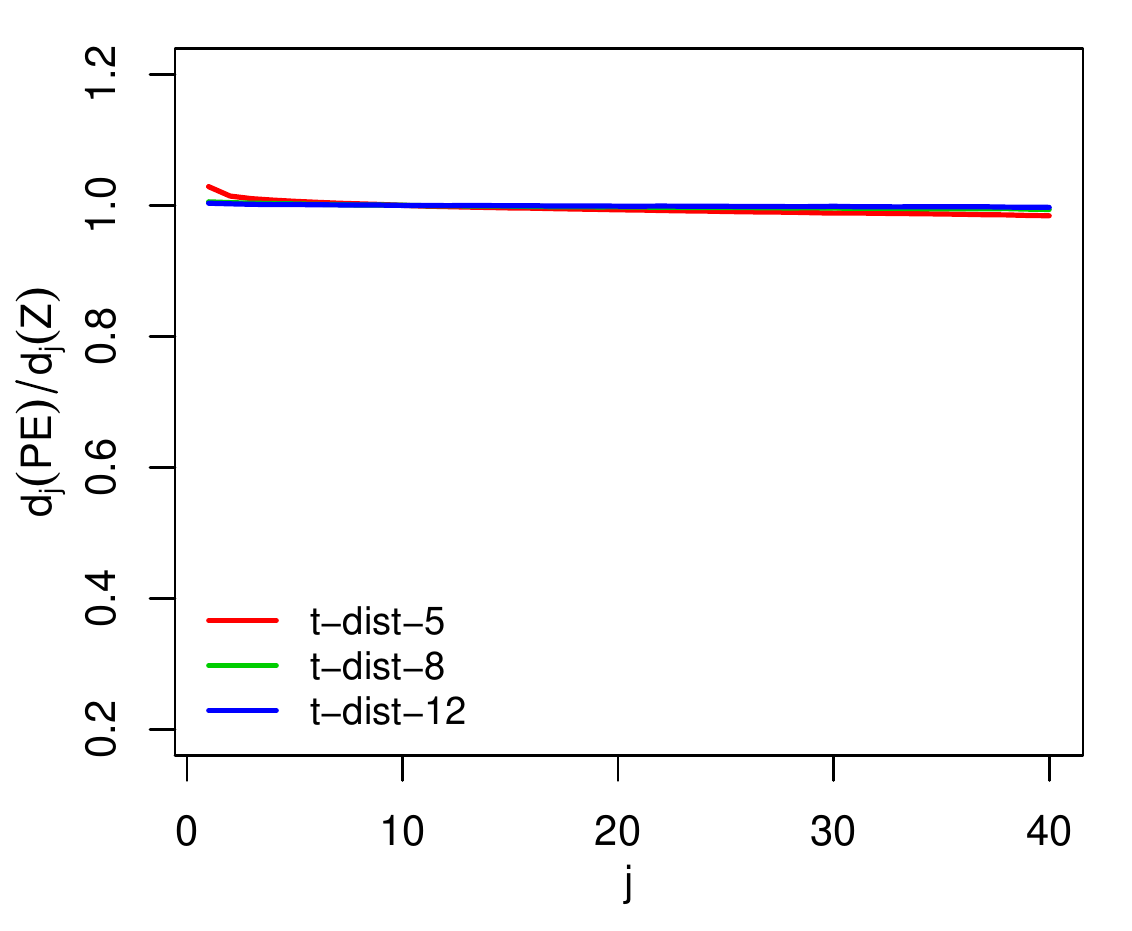}
		\end{tabular}			
		\vspace{-1mm}
		\caption{Panel of $\EE[d_j(PE)]/\EE[d_j(Z)]$ in Experiment 5 with   $n = 150$, $p = 250$, $m =q= 50$  (left) and   $n=50$, $p=q=40$, $m=150$ (right).}
		\label{fig_exp6_1}
		\vspace{-2mm}
	\end{figure} 
	
	In light of this, we further conjecture that our procedures   work in general  settings with heavy tailed error distributions. We consider both low- and high-dimensional settings to verify this claim. The   low-dimensional setting considers $n = 150$, $p = q = m = 30$ and $b_0 = 0.15$ and  the high-dimensional setting  considers $n = 100$, $p = 150$, $q = m = 30$ and $b_0 = 0.015$. We generate $E$ from $t_\nu$-distribution with  $\nu\in\{ 6, 8 , 10 \}$ and we set $\eta = 0.1$ and $r\in\{0,\ldots,20\}$ in both cases.  The plots in Figure \ref{fig_exp6_2}  show that  STRS consistently estimates the rank in both  settings under a very mild SNR ratio and  its performance is quite stable for different heavy tailed $t$-distributions.  
	\begin{figure}[H]
		\centering
		\vspace{-2mm}
		\begin{tabular}{cc}
			\includegraphics[width=0.42\textwidth]{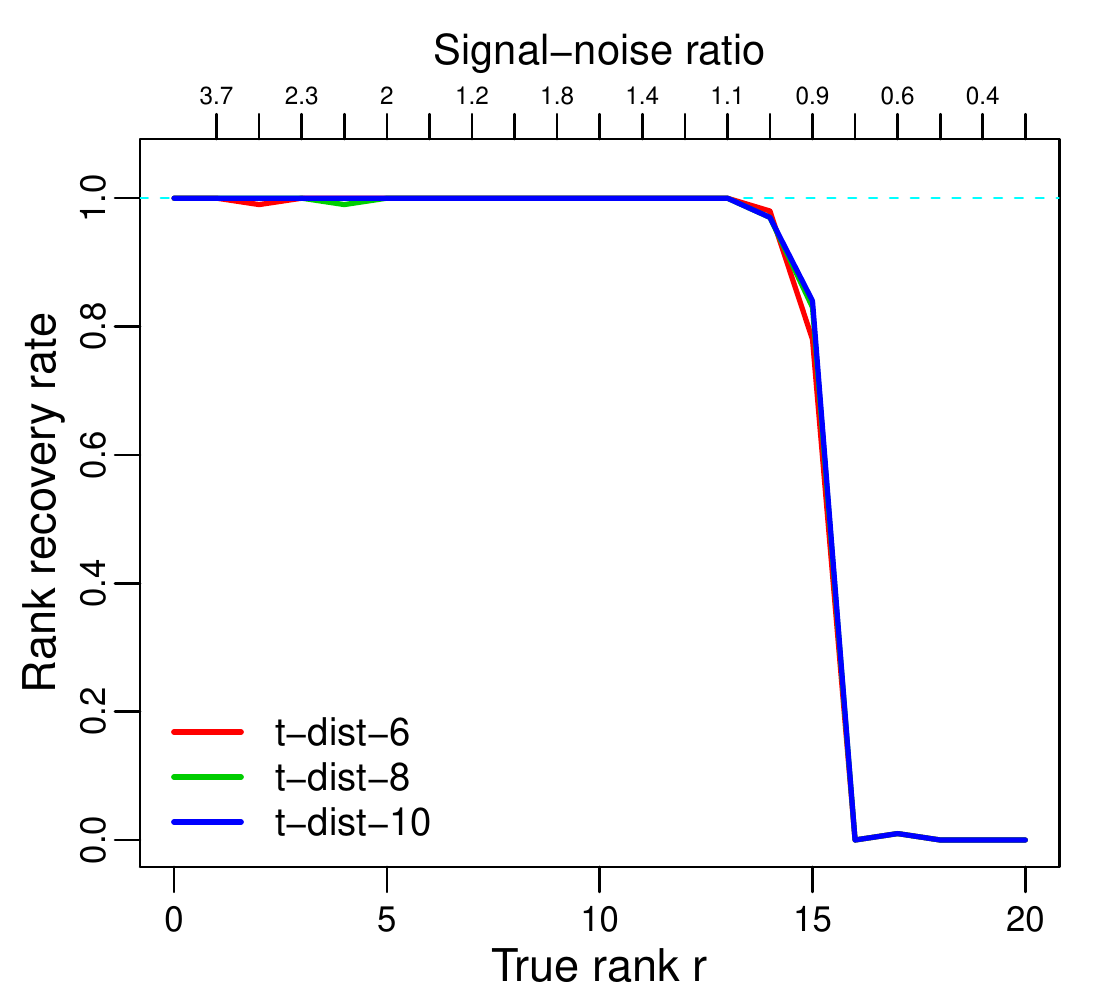}&
			\includegraphics[width=0.42\textwidth]{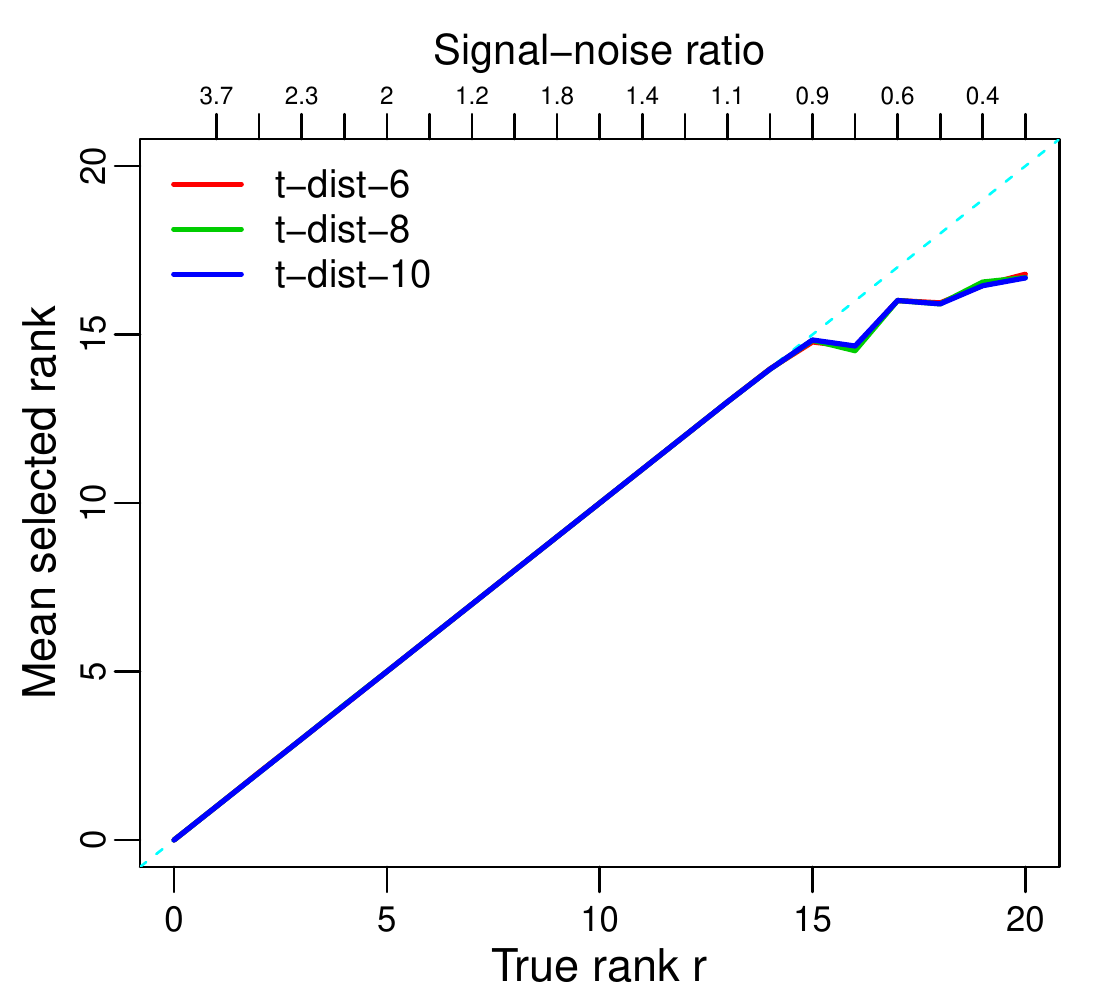}\\
			[-6pt]
			\includegraphics[width=0.42\textwidth]{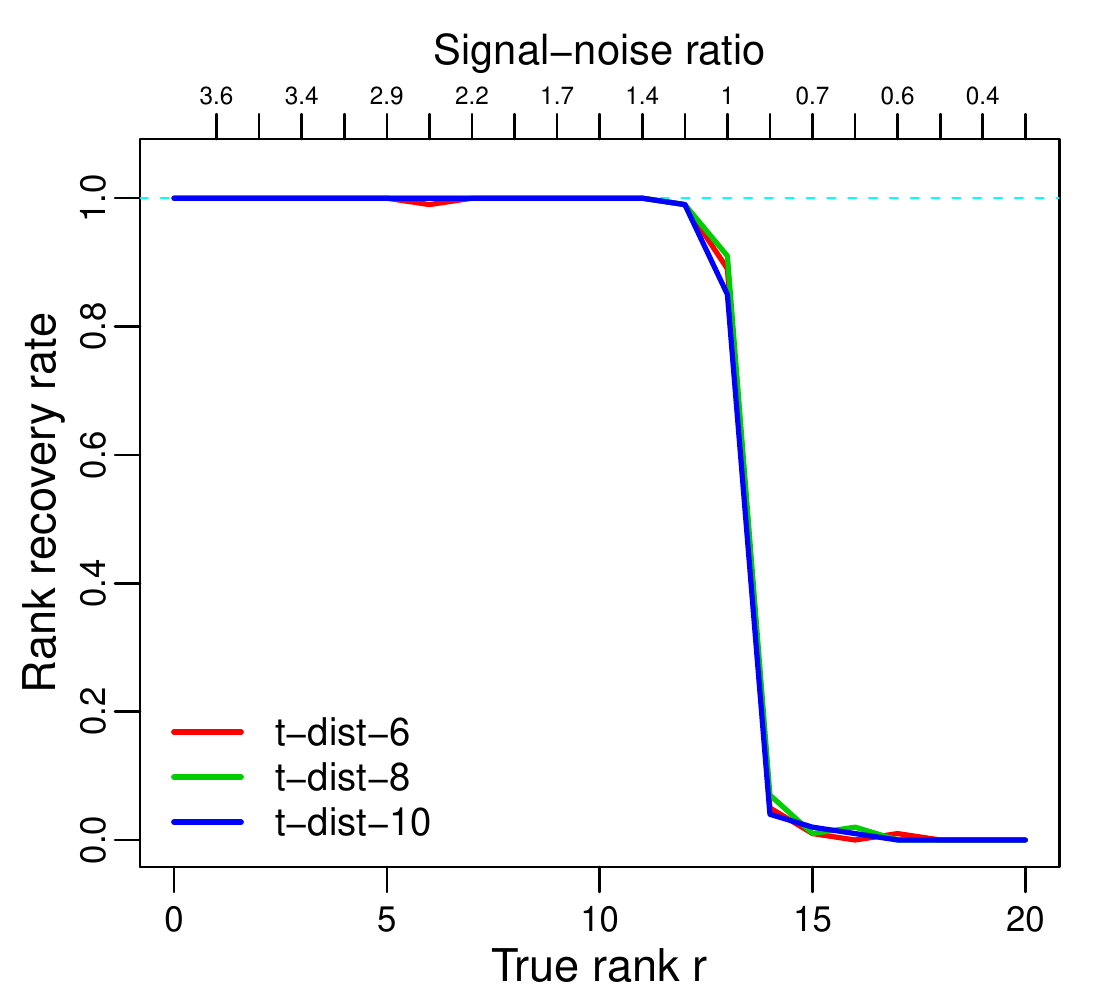}&
			\includegraphics[width=0.42\textwidth]{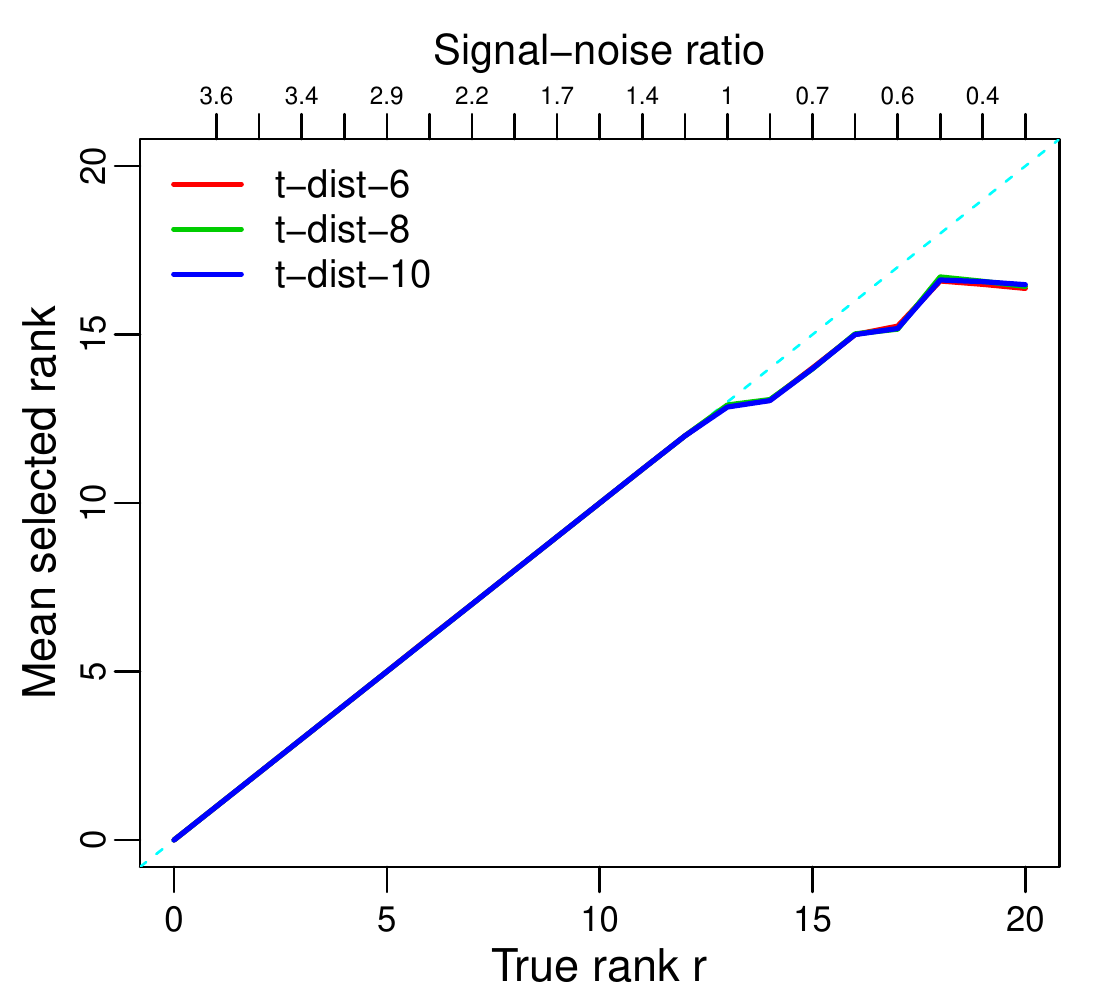}
		\end{tabular}				
		\vspace{-1mm}
		\caption{Performance of STRS for heavy tails in Experiment 5 in   the low-dimensional   (top) and the high-dimensional setting (bottom).}
		\label{fig_exp6_2}
		\vspace{-2mm}
	\end{figure} 
	
	\subsection{Conclusions of the simulation studies}\label{sec_sim_conclusion}
	\begin{itemize}
		\item
		In general, STRS outperforms {BSW-C} in both low-dimensional and high-dimensional settings.
		The performance of {BSW-C} is influenced by the true rank $r$ and there is no globally optimal tuning parameter $C$  for {BSW-C}.  STRS is   stable in general as long as the true rank $r$ lies in its allowable range. 
		\item
		In the most challenging setting of Experiment 2, when $n\approx q$ and estimation of $\sigma^2$ is problematic, the advantage of STRS over BSW-1.1 and BSW-1.3 becomes more prominent. 
		If  $n=q$,   {BSW-C} is no longer feasible, while  STRS only fails in the rare situation when $nm$ is  small compared to $m+q$ and $r$ is large. Of course, reduced rank regression only  makes  sense for relatively small $r$.
		
		\item Experiment 3 verifies that STRS has clear advantages over GRS.  It  requires a smaller signal-to-noise ratio  and allows for larger values of $r$, which confirms our theoretical result in Section \ref{sec_SRS}.
		
		\item Experiment 4 confirms our results in Section \ref{sec_extension}, that  our procedures (GRS, STRS, SSTRS) continue to consistently estimate the true rank for  heavy tailed distributions in certain settings, considered in Section \ref{sec_extension}. Moreover, Experiment 5 confirms our conjecture that STRS   works in more general settings, even if the errors are generated from  heavy tailed  distributions.	
	\end{itemize}
	
	\subsection{Tightness check of signal-to-noise condition}\label{sec_sim_tightness}
	Akin to the discussion in \cite[Section 4.2, p. 1303]{BSW}, we can empirically verify the tightness of the signal-to-noise condition in (\ref{RS2}). Specifically, from (\ref{hat k}), we have 
	\begin{eqnarray} 	\label{cc}  \{ \wh k\ne r  \} =   \{d_{r+1}(PY)\ge \sqrt\lambda\wh \sigma_r\} \cup \{d_r(PY) \le \sqrt\lambda\wh \sigma_r\}.
	\end{eqnarray}
	By using identity (\ref{cc}) and Weyl's inequality,  we observe that
	\begin{equation*} \PP  \{\wh k \ne r\} \ge
	\PP\left \{d_r(XA)+d_1(PE)  \le \sqrt \lambda\wh \sigma_r \right\}.
	\end{equation*}
	Hence we conclude that
	$\PP\{d_r(XA) \le \sqrt\lambda\wh\sigma_r-d_1(PE)\} >0$ implies $\PP\{\wh k=r\} <1$.
	This suggests $d_r(XA)$ cannot be smaller than $\sqrt{\lambda}\wh \sigma_r-d_1(PE)$. To empirically verify this conjecture, we generate different pairs of $(X, A)$ through changing $b_0$, $\eta$, $n, m, p, q$ and  $r$. For each pair of $(X, A)$,   we record the $r$th largest singular value of $XA$ as $d_r(XA)$ and we search along a  grid of $\lambda$ to find the largest $\lambda$ such that minimizing (\ref{vier}) recovers the true rank (recall that $\sqrt{\lambda}\wh\sigma_r$ is increasing in $\lambda$). Finally, we   plot $\lambda\wh\sigma_r$ and $\sqrt{\lambda}\wh\sigma_r- d_1(PE)$ against $d_r(XA)$ for all pairs of $(X, A)$   in Figure \ref{tightness}.  This plot collaborates our  conjecture that the signal-to-noise condition in (\ref{RS2}) is tight.
	\begin{figure}[H]
		\centering
		\vspace{-1mm}
		\includegraphics[width=.5\textwidth]{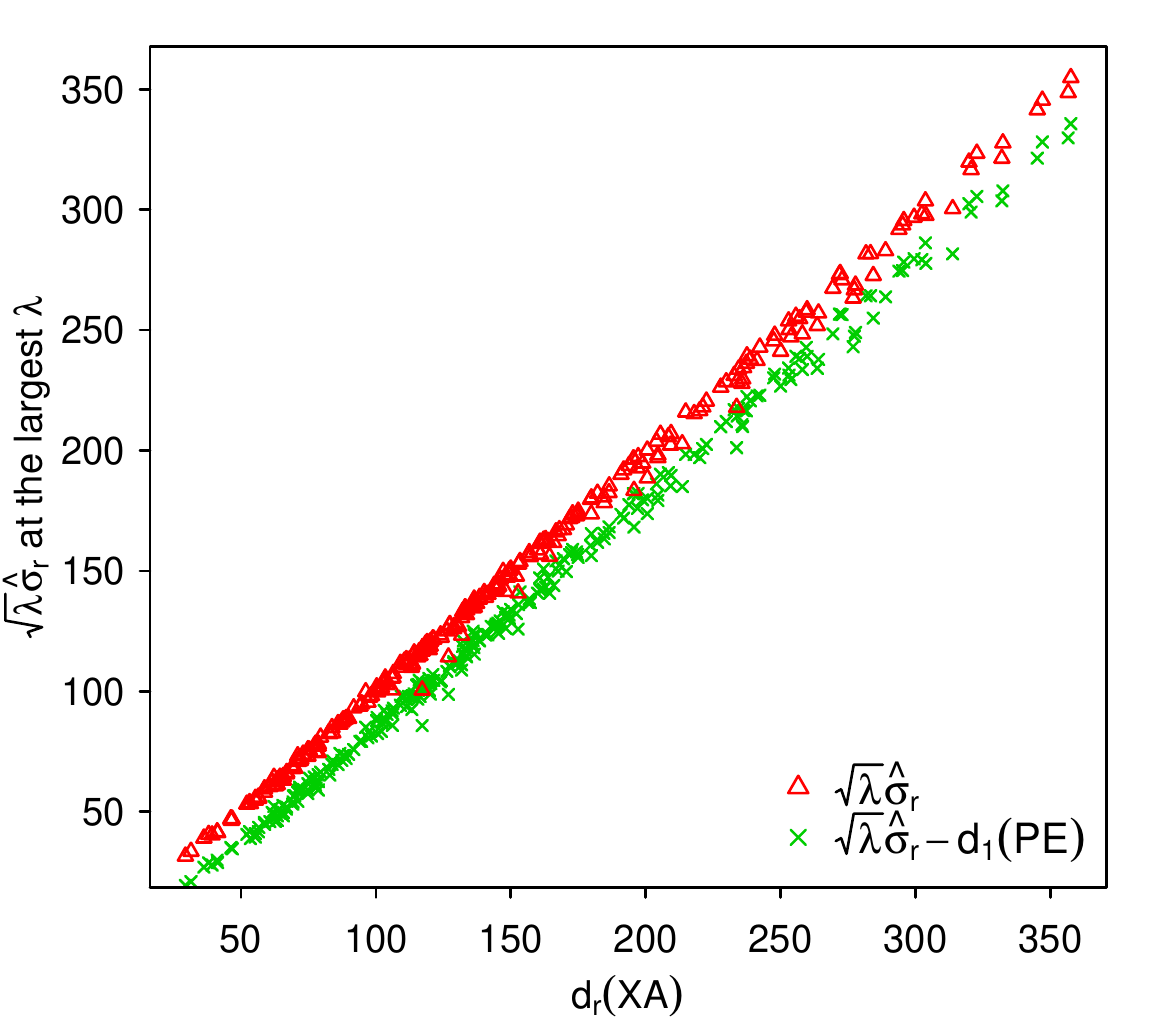}
		\caption{Plot of $\sqrt{\lambda}\wh\sigma_r$ and $\sqrt{\lambda}\wh\sigma_r-d_1(PE)$ versus $d_r(XA)$ for each pair of $(X,A)$. The value for $\lambda$ is the largest one (on a   grid) that correctly found the  true rank.}
		\label{tightness}
	\end{figure}
	
	\section*{Acknowledgements} The authors thank the Editor, Associate Editor and two referees for constructive remarks.
	Wegkamp's research was supported in part by NSF grant DMS 1712709.\\

	\begin{supplement}
		\sname{Supplement to ``Adaptive estimation of the rank of the coefficient matrix in high dimensional multivariate response regression models''}\label{suppA}
		\slink[doi]{COMPLETED BY THE TYPESETTER}
		\sdescription{The supplementary document includes the oracle inequality for the fit, additional simulation results and all proofs.}
	\end{supplement}

	\vspace{.8cm}
	\bibliographystyle{imsart-nameyear}
	\bibliography{ref}
	
	\newpage
	\appendix
	
\section{Proofs of Sections 2 \& 3}
\subsection{\bf Proof of Proposition \ref{prop1}}
We first note that, since $(PY)_i=P(PY)_i$, see \cite{GiraudBook} (page 124),  and Pythagoras' identity
\begin{eqnarray}
\| Y-(PY)_i\|^2 &=& 
\| Y-PY\|^2 + \| PY- (PY)_i\|^2 \label{identiteit}\\
&=& \| Y-PY\|^2 +  \sum_{k=i+1}^j d_k^2(PY)  +  \sum_{k>j} d_k^2(PY)\nonumber\\
&=& \| Y-(PY)_j\|^2 +  \sum_{k=i+1}^j d_k^2(PY).\nonumber
\end{eqnarray}
Consequently, with $i<j$,
\begin{eqnarray*} \wh \sigma_i^2 \ge \wh\sigma_j^2 
	&\iff& \frac{ \| Y- (PY)_i \|^2 }{ nm-\lambda i} \ge \frac{ \| Y- (PY)_j \|^2 }{ nm-\lambda j} \\
	&\iff& \frac{ \| Y- (PY)_j + \sum_{k=i+1}^j d_k^2 (PY) }{ (nm-\lambda j) + \lambda(j-i)} \ge \frac{ \| Y-(PY)_j\|^2}{nm-\lambda j}\\
	&\iff&\frac{  \sum_{k=i+1}^j d_k^2 (PY) }{   \lambda(j-i)} \ge \frac{ \| Y-(PY)_j\|^2}{nm-\lambda j}=\wh\sigma_j^2
\end{eqnarray*}
using the simple fact that $(a+b)/(c+d) \ge a/c \iff (b/d)\ge (a/c)$ for any positive numbers $a,b,c,d$. This proves (\ref{crit1}). Claim (\ref{crit2}) follows from (\ref{crit1})  by taking $i=j-1$. Finally,
\begin{eqnarray*}
	\frac{ d_{j}^2 (PY) }{ \lambda} &\le& \frac{ \| Y- (PY)_j\|^2 }{nm- \lambda j}
	\ =\ \frac{ \| Y- (PY)_{j-1} \|^2 - d_j^2(PY)}{ nm-\lambda (j-1) - \lambda}
\end{eqnarray*}
is equivalent with
\begin{eqnarray*}
	\frac{ d_{j}^2 (PY) }{ \lambda} 
	&\le& \frac{ \| Y- (PY)_{j-1} \|^2 }{ nm-\lambda (j-1) } = \wh \sigma^2_{j-1}
\end{eqnarray*}
using the above elementary manipulation again and (\ref{crit5}) follows.
This completes our proof.\qed\\

\subsection{\bf Proof of Proposition \ref{prop2}}
We first show (\ref{crit3}). 
Suppose $\wh \sigma^2_k\le \wh\sigma^2_{k-1}$. We observe
\begin{eqnarray*}
	\frac{1}{k-\ell} \sum_{j=\ell+1}^k d_j^2(PY) &\ge& d_k^2(PY)\quad \text{ by }d_1(PY)\ge d_2(PY)\ge \cdots\\
	&\ge& \lambda\wh\sigma_{k}^2 \qquad \text{by (\ref{crit2})}
\end{eqnarray*}
so that (\ref{crit1}) implies $\wh\sigma^2 _k \le \wh\sigma_\ell^2$ for all $\ell \le  k-1$. This proves the non-trivial direction of (\ref{crit3}).

Next, we show (\ref{crit4}). Suppose $\wh \sigma^2_k\ge \wh\sigma^2_{k-1}$. Then, by (\ref{crit2}) and $d_{k+1}(PY) \ge d_k(PY)$, we get
\begin{eqnarray*}
	d_{k}^2(PY)\le  \lambda\wh\sigma_{k}^2\quad \Longrightarrow \quad d_{k+1}^2(PY)\le  \lambda\wh\sigma_{k}^2 \quad \overset{(\ref{crit5})}{\Longleftrightarrow} \quad d_{k+1}^2(PY)\le  \lambda\wh\sigma_{k+1}^2.
\end{eqnarray*}
Note that the last inequality further implies $\wh \sigma_{k+1}^2 \ge \wh\sigma_{k}^2$ by (\ref{crit2}) again. Repeating the same reasoning completes our proof.\qed\\

\subsection{\bf Proof of Theorem \ref{thm:null}}
By Theorem \ref{closedform}, it suffices to show 
$d_1^2(PY) \leq \lambda \wh \sigma^2_1$. This 
is equivalent to $d_1^2(PY) \le \lambda\wh\sigma_0^2$ by criterion (\ref{crit5}) in Proposition \ref{prop1}. The latter, in turn,  is equivalent to $ d_1^2(PE) \le \lambda\wh \sigma^2$ as $XA=0$ implies
$d_1^2(PE)=d_1^2(PY)$ and $\wh \sigma^2_0 =\wh\sigma^2$.\qed\\

\subsection{Proof of Corollary \ref{c1}}
We can write $\lambda := C (\sqrt{m}+\sqrt{q})^2$ for some 
$C = (1+C_0)^2/(1-C_1)>1$ with $C_0>0$ and  $0<C_1<1$. By Theorem \ref{thm:null}, (\ref{d_1}) and (\ref{chiL}), we have
\begin{align*}
\PP\{ \wh k\ne 0\} &\le 
\PP\left\{ d_1^2(PE)\ge \lambda \wh\sigma^2 \right\}\\
&\le
\PP\left\{	d_1^2(PE)\geq (1+C_0)^2(\sqrt{m}+\sqrt{q})^2\sigma^2
\right\}  +\PP\left\{ \wh\sigma^2\le (1-C_1)\sigma^2\right\}\\
& \le  \exp \left\{- C_0^2 (\sqrt{m}+\sqrt{q})^2 /2\right\}+\exp\left\{- 
C_1^2nm /4\right\},
\end{align*}
which proves the claim.\qed\\

\subsection{\bf Proof of Theorem \ref{thm:RS1}}
By Proposition \ref{prop1}, for $\wh k\le r$, we need to show that
$
d_{r+1}^2(PY) <\lambda \wh\sigma^2_{r+1}\label{b}.	
$
We  observe that 
\begin{equation*}
d_{r+1}^2 (PY) < \lambda \wh \sigma_{r+1}^2\ \quad \iff \quad
\lambda > d_{r+1}^2 (PY)/ \wh \sigma_{r}^2
\end{equation*} by statement (\ref{crit5}). Again, on the event (\ref{RS1}), an application of Weyl's inequality and observing that $d_{r+1}(XA)=0$ yield
\begin{equation*}
\lambda ~\ge~ d_1^2(PE) / \wh\sigma_{r}^2 ~\ge~  d_{r+1}^2(PY)/\wh\sigma_{r}^2 
\end{equation*}
which is exactly what needed to be shown.\qed\\

\subsection{Proof of Proposition \ref{prop:ongelijk}}
To show (\ref{eq_sigma_r}), on the one hand, we use the Eckhart-Young theorem and the fact that $r(XA)\le r$ to deduce
$\| Y- (PY)_r\|^2 \le \| Y-XA\|^2 = \|E\|^2$. Hence 
\[\wh \sigma_r^2 \le \frac{nm}{nm-\lambda r}\wh\sigma^2 \]
On the other hand,  Weyl's inequality shows that 
$d_{r+i}(PY)\ge d_{2r+i}(PE)$
for $0\le  i \le N-r$ by defining $d_k(PE) := 0$ for $k > N$, and we obtain    	
\begin{eqnarray}\nonumber
\wh \sigma^2_{r} & = & \frac{\|Y-(PY)_{r}\|^2}{nm-\lambda r}\ =\ \frac{\|Y-PY\|^2 + \sum_{j=r+1}^{N}d_j^2(PY)}{nm-\lambda r}\\\nonumber
& \ge & \frac{\|E-PE\|^2 + \sum_{j=2r+1}^{N}d_j^2(PE)}{nm-\lambda r}\\\label{lowbdsigma}
& = & \frac{\|E-(PE)_{2r\wedge N}\|^2}{nm-\lambda r}.
\end{eqnarray}
If $2r \ge N$, then 
$$\frac{\|E-(PE)_{2r\wedge N}\|^2}{nm-\lambda r}= \frac{\|E-PE\|^2}{nm-\lambda r}\ge \frac{\|E-(PE)_{N}\|^2}{nm-\lambda N/2}.
$$
We conclude the proof by invoking (\ref{mono}) in Lemma \ref{lem: mono}. \qed\\

\begin{lemma}\label{lem: mono}
	For any given $1\le k \le r$ and $2k \le N-2$, if $\lambda$ satisfies
	\[
	\lambda ~\ge~ \frac{nm}{\|E-(PE)_{2k}\|^2/ \left[d_{2k+1}^2(PE)+d_{2k+2}^2(PE) \right]+k},
	\]
	then 
	\begin{eqnarray}\label{prev2lbd}
	\frac{\|E-(PE)_{2k}\|^2}{nm-\lambda k}\le \frac{\|E-(PE)_{2r}\|^2}{nm-\lambda r}
	\end{eqnarray}
	In particular, on the event $\{\lambda\wh\sigma^2 \ge d_1^2(PE)+d_2^2(PE)\}$, we have
	\begin{equation}\label{mono}
	\frac{\|E\|^2}{nm}\le \frac{\|E-(PE)_2\|^2}{nm-\lambda}\le \frac{\|E-(PE)_4\|^2}{nm-2\lambda} \le \cdots\le  \frac{\|E-(PE)_N\|^2}{nm-\lambda N/2 }.
	\end{equation}
\end{lemma}

\begin{proof}[Proof of Lemma \ref{lem: mono}]
	We first show (\ref{mono}) by using the same argument as in Propositions \ref{prop1} and \ref{prop2}. For any $ 0 \le k \le (N/2-1)$, we define 
	\begin{eqnarray}\label{ak}
	e_k := \frac{\|E-(PE)_{(2k)}\|^2}{nm-\lambda k}.
	\end{eqnarray}
	Observe that 
	\begin{align}\label{eqmono}\nonumber
	e_k\le e_{k+1} &\\\nonumber
	\iff&\ \frac{\|E-(PE)_{(2k)}\|^2}{nm-\lambda k} \le \frac{\|E-(PE)_{(2k+2)}\|^2}{nm-\lambda(k+1) }\\\nonumber
	\iff &\ \frac{\|E-(PE)_{(2k)}\|^2}{nm-\lambda k} \le \frac{\|E-(PE)_{(2k)}\|^2-d_{2k+1}^2(PE)-d_{2k+2}^2(PE)}{nm-\lambda k-\lambda }\\
	\iff &\ \frac{d_{2k+1}^2(PE)+d_{2k+2}^2(PE)}{\lambda} \le 
	\frac{\|E-(PE)_{2k}\|^2}{nm-\lambda k}  =e_k.
	\end{align} 
	For $k=0$, we find that
	$e_0\le e_1$   is equivalent with 
	$ d_1^2(PE)+d_2^2(PE)\le \lambda\wh\sigma^2$, which is precisely our condition on $\lambda$. From the decreasing property of singular values, (\ref{eqmono}) implies
	\begin{equation*}
	\frac{d_{2k+3}^2(PE)+d_{2k+4}^2(PE)}{\lambda} \le 
	\frac{\|E-(PE)_{2k+2}\|^2}{nm-\lambda (k+1) } = e_{k+1},
	\end{equation*}
	which is, by the same argument above, equivalent with
	$e_{k+1}\le e_{k+2}$. We  conclude that
	$e_0\le e_1\le e_2\le e_3\le \ldots$, proving
	(\ref{mono}). \emph{In fact, this proves the sequence of $\{e_k\}$ could only be either globally monotonic or decreasing first and then increasing.}
	
	Finally, since $k \le r$, (\ref{prev2lbd}) follows immediately from (\ref{eqmono}) and the monotone or two-sided monotone property of $\{e_k\}$.
\end{proof}
\vspace{1mm}

\subsection{Proof of Theorem \ref{rankGaussian1}}
From Theorem \ref{thm:RS1} and Proposition \ref{prop:ongelijk}, it suffices to show $\{2d_1^2(PE) \le \lambda\wh \sigma^2\}$ holds with high probability. The proof follows exactly the same arguments as the proof of Corollary \ref{c1} by adapting to the constant $2$.

\subsection{Proof of Theorem \ref{thm:RS2}}
To show $\wh k \ge s$, it suffices to prove $d_s^2(PY)  \ge \lambda \wh\sigma^2_s$.
Identity (\ref{identiteit}) gives
\begin{eqnarray*}
	\| Y-(PY)_s\|^2 &=& \| Y-(PY)_{r}\|^2 + \sum_{j=s+1}^{r} d_j^2(PY)\\
	&\le& \| Y-(PY)_{r}\|^2 + (r-s) d_{s+1}^2(PY).
\end{eqnarray*} 
Consequently,
\begin{eqnarray*}
	d_s^2(PY)  &\ge& \lambda\cdot \frac{ \| Y-(PY)_s\|^2}{nm-\lambda s}
\end{eqnarray*}
is implied by
\begin{eqnarray*}
	d_s^2(PY)  &\ge& \lambda\cdot \frac{ \| Y-(PY)_{r}\|^2   +    (r-s) d_s^2(PY)}{nm-\lambda s}
\end{eqnarray*}
which in turn, after a little algebra, is seen to be  equivalent to
\begin{eqnarray*}
	d_s^2(PY)  &\ge& \lambda\cdot \frac{ \| Y-(PY)_{r} \|^2} {nm-\lambda r} = \lambda \wh \sigma^2_{r}.
\end{eqnarray*}
By Weyl's inequality, on the event   (\ref{RS2}), we have
\begin{eqnarray*}
	d_s(PY) \ge d_s(XA) - d_1(PE) 
	\ge   \sqrt{ \lambda} \wh\sigma_{r},
\end{eqnarray*}
which shows $d_s^2(PY)  \ge \lambda \wh\sigma^2_s$. \qed\\

\subsection{Proof of Corollary \ref{cor: generalrank}}
The proof follows immediately by noting the events in this corollary combined with (\ref{eq_sigma_r}) imply (\ref{RS1}) and (\ref{RS2}). \qed\\

\subsection{\bf Proof of Theorem \ref{rankGaussian}}
We define   $\C = \bigl\{(1-C_1)\sigma^2\le \wh\sigma^2 \le (1+C_1)\sigma^2\bigr\}$ for any $0<C_1<1$. Choose $C' = (1+C_1)\sqrt{C}(1 + \sqrt{2(1+\delta)})$ and $C_0>0$ such that $C = (1+C_0)^2/(1-C_1)$. 
Then we have
\begin{align*}
&\PP\left\{ \lambda \le  2d_1^2(PE) / \wh \sigma^2 ~\text{ or }~ d_s(XA)\le \sqrt{\lambda}  \wh\sigma \left[ {1\over \sqrt{2} } + \sqrt{\frac{nm}{nm-\lambda r}}\ \right]   \right\}\\
& \qquad\le   \PP\left\{d_1^2(PE)\ge (1+C_0)^2(\sqrt{m}+\sqrt{q})^2\sigma^2 \right\}\\
&\qquad\quad +\PP\left\{  d_s(XA) \le C'\sigma(\sqrt{m}+\sqrt{q}) \right\} +
\PP\{\C^c\}\\
&\qquad=  \PP\left\{d_1^2(PE)\ge (1+C_0)^2(\sqrt{m}+\sqrt{q})^2\sigma^2 \right\}+		\PP\{\C^c\}\\
&\qquad \le \exp \left\{ -C_0^2(\sqrt{m}+\sqrt{q})^2/2 \right\}+2  \exp\left\{-{3 C_1^2nm}/{16}\right\}
\end{align*}
using (\ref{d_1}), (\ref{chiL}) and (\ref{chiR}). Invoking Corollary \ref{cor: generalrank} concludes the proof. \qed
\\

\section{Proofs of Section 4}\label{sec_proof_SRS}

We first present several lemmas which are repeatedly used in the following proofs. 

\begin{lemma}\label{lem_tech}
	Let $c_i$ be some positive constants for $i = s, s+1, \ldots, t$. Then, for any $\eps \in [0,1)$,
	\begin{align*}
	\sum_{i=s}^t\left(c_i   +\eps S\right)^2 &\le (1+\eps)^2\sum_{i=s}^tc_i^2,\\
	\sum_{i=s}^t\left(c_i - \eps S\right)^2 &\ge (1-\eps)^2\sum_{i=s}^tc_i^2,
	\end{align*}
	where $ S = \sqrt{ \sum_{i=s}^tc_i^2 /(t-s)}.$
\end{lemma}
\begin{proof}
	Working out the square yields
	\begin{align*}
	\sum_{i=s}^t\left(c_i + \eps S\right)^2 &= \sum_{i=s}^tc_i^2 +\eps^2\sum_{i=s}^tc_i^2  + 2\eps S\sum_{i=s}^tc_i \le (1+\eps^2+2\eps)\sum_{i=s}^tc_i^2
	\end{align*}
	using the Cauchy-Schwarz inequality. This proves the first statement. The second one follows by the same arguments.
\end{proof}
\vspace{-1mm}

\begin{lemma}(Interlacing inequality \cite{hj})\label{lem_interlacing}
	Let $A$ be a given $m\times n$ matrix, and let $A_r$ denote a submatrix of $A$ obtained by deleting a total of $r$ rows and/or columns from A. Then 
	\[
	d_{k}(A) \ge d_k(A_r)\ge d_{k+r}(A),\quad k = 1,\ldots, \min\{m, n\}
	\]
	where for $M\in \RR^{p\times q}$ we set $d_j(M) = 0$ if $j \ge \min\{p, q\}$.
\end{lemma}

\begin{lemma}\label{lem_exp_PE}
	Let $n\times m$ matrix $E$ have i.i.d. $N(0, 1)$ entries and $P$ be any $n\times n$ projection matrix with rank equal to $q$. Suppose $q\le m$. Then, one has 
	\begin{equation}\label{eq_exp_d1}
	\EE[d_1^2(PE)] \ge m
	\end{equation} 
	and, for any $2\le k\le q$,  
	\begin{equation}\label{eq_bound_exp_PE}
	\EE[d_k(PE)] \le \sqrt{m} + \sqrt{q-k+1}, \quad 
	\EE[d_k(PE)] \ge \sqrt{m} - \sqrt{k}.
	\end{equation}
	Moreover, for any $1\le k\le q$,
	\begin{equation}\label{eq_cc_PEk}
	\PP\left\{\big|d_k(PE) - \EE[d_k(PE)]\big| \ge t\right\} \le 2\exp(-t^2/2), \quad \forall\ t\ge 0.
	\end{equation}
	For the case $m<q$, similar results hold for any $1\le k\le m$, by switching $q$ and $m$.
\end{lemma}

\begin{proof}[Proof of Lemma \ref{lem_exp_PE}]
	We only prove the case when $q\le m$. The complementary case $q>m$ can be derived using the fact $d_k(M)= d_k(M^T)$ for any matrix $M$.
	
	We start by writing the eigenvalue decomposition of $P$ as $P = U\Lambda U^T$ with orthogonal matrix $U\in \RR^{n\times n}$ and diagonal matrix $\Lambda$ whose first $q$ diagonal elements are $1$ and $0$ elsewhere. Thus, for any $1\le j\le q$, $d_j^2(PE) = \lambda_j(E^TPE) = d_j^2(\Lambda U^TE)$ where $\lambda_j(M)$ denotes the $j$th largest eigenvalue of $M$. Note that $Z = \Lambda U^TE\in \RR^{n\times m}$ has $q\times m$ submatrix with i.i.d. $N(0,1)$ entries while the remaining $(n-q)\times m$ entries are all $0$. For any $1\le k\le q$, by Lemma \ref{lem_interlacing}, we have 
	$d_k(Z)\le d_1(\overline{Z}_{q-k+1})$ and $d_k(Z)\ge d_k(\overline{Z}_k)$, where $\overline{Z}_j$ denotes the matrix made of the first $j$ rows of $Z$. Then (\ref{eq_bound_exp_PE}) follows immediately from Theorem 5.32 of \cite{rv_rand_mat}. 
	
	Concentration inequality  (\ref{eq_cc_PEk}) follows from the Gaussian concentration inequality of \cite[page 221]{GiraudBook} and the fact  that each singular value function is $1$-Lipschitz with respect to the Frobenius norm.
	
	Finally, $\EE[d_1^2(PE)] \ge \EE[d_1^2(\overline{Z}_1)]$ and
	\[
	\EE\left[d_1^2(\overline{Z}_1)\right] = \EE\left[\overline{Z}_1^T\overline{Z}_1\right] = \sum_{i = 1}^{m}\EE\left[\overline{Z}_{1i}^2\right] = m
	\]
	implying (\ref{eq_exp_d1}). This completes the proof.
\end{proof}
\vspace{3mm}

The next lemma proves one-side concentration of $\|PE-(PE)_k\|$ around its mean.
\begin{lemma}\label{lem_svd_PE}
	Let $n\times m$ matrix $E$ have i.i.d. $N(0, 1)$ entries and $P$ be any $n\times n$ projection matrix with rank equal to $q$. Assume $q\le m$. Then, for any $1\le k< q$ and any $\eps \in (0,1)$, there exists a constant $C=C(\eps)>0$, such that
	\begin{equation*}
	\PP\left\{\|PE - (PE)_k\|^2 \le (1-\eps)\EE[\|PE - (PE)_k\|^2]\right\} \le  e^{-Cm}.
	\end{equation*}
	Similar results hold by switching $q$ and $m$ when $m<q$ for $1\le k< m$.
\end{lemma}

\begin{proof}[Proof of Lemma \ref{lem_svd_PE}]
	Fix $1\le k< q$. Note that 
	$$
	\|PE-(PE)_k\|^2 = \|PE\|^2 - \sum_{j=1}^kd_j^2(PE).
	$$ 
	We first study $\sum_{j=1}^kd_j^2(PE)$. From (\ref{eq_cc_PEk}) in Lemma \ref{lem_exp_PE}, for any $k\le j\le q$, we have 
	\[
	\PP\left\{d_j^2(PE) \ge (\EE[d_j(PE)]+ t)^2\right\} \le \exp(-t^2/2), \quad \forall\ t\ge 0.
	\]
	This further implies
	\begin{align*}
	\PP\left\{
	\sum_{j=1}^kd_j^2(PE) \le \sum_{j=1}^k(\EE[d_j(PE)]+t)^2
	\right\} \ge 1-k\exp(-t^2/2).
	\end{align*}
	Let $c>0$ be sufficiently small and choose $t^2 = (c^2 / k) \sum_{j=1}^k(\EE[d_j(PE)])^2$.  Invoking Lemma \ref{lem_tech} and using $\EE[d_j^2(PE)]\ge (\EE[d_j(PE)])^2 $ yield
	\begin{align*}
	&\PP\left\{
	\sum_{j=1}^kd_j^2(PE) \le (1+c)^2\sum_{j=1}^k\EE[d_j^2(PE)]\right\}\\
	&\qquad\qquad \ge \PP\left\{
	\sum_{j=1}^kd_j^2(PE) \le (1+c)^2\sum_{j=1}^k(\EE[d_j(PE)])^2\right\}\\
	&\qquad\qquad \ge 1-k\exp\left[-{c^2\sum_{j=1}^k(\EE[d_j(PE)])^2 \over 2k}\right].
	\end{align*}
	From $(\EE[d_j(PE)])^2 \ge \EE[d_j^2(PE)] - 1$ in the proof of Lemma 1 of \cite{Giraud}, observe that 
	\begin{align*}
	\sum_{j=1}^k (\EE[d_j(PE)])^2 &\ge {k\over q} \sum_{j=1}^q (\EE[d_j(PE)])^2 \ge {k\over q} \sum_{j=1}^q(\EE[d_j^2(PE)]-1)\\
	& = {k\over q}\cdot \EE[\|PE\|^2]-k = k(m-1).
	\end{align*}
	In the last equality, we use the fact that $\|PE\|^2$ has a central $\chi^2_{qm}$ distribution. This further implies
	\begin{align*}
	&\PP\left\{
	\sum_{j=1}^kd_j^2(PE) \le (1+c)^2\sum_{j=1}^k\EE[d_j^2(PE)]\right\} \ge 1-e^{-Cm}
	\end{align*}
	for some constant $C = C(c)>0$. On the other hand, since $\|PE\|^2 \sim \chi^2_{qm}$, using (\ref{chiL}) yields
	\[
	\PP\left\{
	\|PE\|^2 \ge (1-c)^2\EE\left[\|PE\|^2\right] 
	\right\} \ge 1- \exp(-c^2qm/4).
	\]
	Combining these two displays concludes our proof.
\end{proof}
\vspace{3mm}

\subsection{Proof of Proposition \ref{prop:monolbd}}
We first show that $\wh \lambda_{t+1}:= \wh \lambda_{\wh k_t}$ defined in (\ref{lbdt}) is decreasing in $\wh k_t$. To simplify notation, we write $k$ for $\wh k_t$. Since it is immediate to verify the decreasing property when $2k\ge N$, we consider  $1\le k <N/2$ only. It suffices to show $(1-\eps)\wh R_t /\wh U_t + k$ is increasing. Recall that $S_j = \EE[d_j^2(Z)]$ for any $1\le j\le N$ and $\wh R_t$ and $\wh U_t$ are defined as  
\begin{equation}\label{def_rut_supp}
\wh R_t ~:=~ (n-q) m +  \rs\sum_{j=2k+ 1}^N\rs S_j,\qquad 
\wh U_{ t}~: =~ S_1 \vee \left(S_{2k+1}+ S_{2 k+2}\right),
\end{equation}
by using $k$ in lieu of $\wh k_t$.
First, we show 
\begin{eqnarray}\label{increasingak}		
a_k := {(1-\eps)\left[(n-q)m+\sum_{j=2 k+1}^NS_j\right]\over S_{2k+1}+S_{2k+2}}+k 
\end{eqnarray}
is increasing. This follows from 
\begin{align*}
{a_{k}-a_{k-1} \over 1-\eps} &\ge  {(n-q)m+\sum_{j=2 k+1}^NS_j\over  S_{2k+1}+S_{2k+2}}-{(n-q)m+\sum_{j=2k+1}^{N}S_j\over  S_{2k-1}+S_{2k}}>0.
\end{align*}
It remains to show, if $S_1 \ge S_{2k+1}+S_{2k+2}$, the sequence 
\begin{eqnarray}\label{increasingbk}
b_k :={(1-\eps)\left[(n-q)m+\sum_{j=2 k+1}^NS_j\right] \over S_1}+k
\end{eqnarray}
is increasing in $k$. This is guaranteed by the fact 
that $b_k$ starts increasing after $k$ such that 
$S_1 \ge S_{2k+1}+ S_{2k+2}$. This proves $\wh \lambda_{t+1}$ is decreasing in $\wh k_t$. We conclude the proof of $\wh \lambda_{t+1} < \wh\lambda_t$ for any $t\ge 0$ by noting that $\wh\lambda_0$ is greater than the $\lambda$ obtained from (\ref{lbdt}) by using $\wh k_t = 0$.

Next,  we show that, for given $\wh\lambda_{t}>\wh\lambda_{t+1}$, we have $\wh k_t\le \wh k_{t+1}$. Suppose $\wh k_t > \wh k_{t+1}$, we obtain
\begin{equation*}
\frac{ \| Y- (PY)_{\wh k_{t+1}} \| ^2}{ nm -\wh\lambda_{t+1} {\wh k_{t+1}}} 
\le  \frac{ \| Y- (PY)_{\wh k_{t}} \| ^2}{ nm -\wh\lambda_{t+1} {\wh k_{t}}}
~\iff~ \frac{\sum_{j=\wh k_{t+1}}^{\wh k_t}\rs d_j^2(PY)}{\wh k_t-\wh k_{t+1}} \le \frac{ \| Y- (PY)_{\wh k_t} \| ^2}{ nm/\wh\lambda_{t+1} -{\wh k_t}}
\end{equation*}
Similarly,
\begin{equation*}
\frac{ \| Y- (PY)_{\wh k_{t+1}} \| ^2}{ nm -\wh\lambda_t {\wh k_{t+1}}} 
\ge  \frac{ \| Y- (PY)_{\wh k_t} \| ^2}{ nm -\wh\lambda_t {\wh k_t}}
~\iff~ \frac{\sum_{j=\wh k_{t+1}}^{\wh k_t}\rs d_j^2(PY)}{\wh k_t-\wh k_{t+1}} \ge \frac{ \| Y- (PY)_{\wh k_t} \| ^2}{ nm/\wh\lambda_t -{\wh k_t}}
\end{equation*}
which is a contradiction as $\wh\lambda_t> \wh\lambda_{t+1}$. 

Finally, we show $\wh k_t\le r$ for all $0\le t\le T$. We start defining the event
\begin{align}\label{eq_E}\nonumber
&\E := \left\{\bigcap_{k=1}^N\E_k\right\}\bigcap\ \left\{2d_1^2(PE) \le \wh \lambda_0\wh \sigma^2 \right\}\\
&\E_k := \left\{ 
{\|E-(PE)_{2k\wedge N}\|^2 \over d_1^2(PE) \vee \left[d_{2k+1}^2(PE) + d_{2k+2}^2(PE)\right] } \ge {(1-\eps)\wh R_t \over \wh U_t}	
\right\}
\end{align}
with $\wh \lambda_0$ chosen as (\ref{initlbd}) and $\wh R_t$ and $\wh U_t$ defined in (\ref{def_rut}). We will work on this event in the remainder of the proof. From Theorem \ref{thm:RS1} and Proposition \ref{prop:ongelijk}, we know $1\le \wh k_0 \le r$ where $\wh k_0$ is the selected rank from (\ref{k0}). Let $\wh \lambda_1$ be updated via (\ref{lbdt})  using $\wh k_0$. In order to guarantee
$\wh k_{1}\le r$, from (\ref{startpoint}) and (\ref{lowbdsigma}), it suffices to show
\begin{eqnarray}\label{lessthanr}
d_1^2(PE) \le \frac{\|E-(PE)_{(2r)\wedge N}\|^2}{nm/ \wh \lambda_1-r}.
\end{eqnarray}
On the one hand, the choice of $\wh \lambda_1$ satisfies
\begin{align*}
\wh \lambda_1 &= {nm \over (1-\eps)\wh R_t/\wh U_t + \wh k_0}\\
& \overset{\E}{\ge} \frac{nm}{\|E-(PE)_{(2\wh k_0)\wedge N}\|^2\big/\left(d_{2\wh k_0+1}^2(PE)+d_{2\wh k_0+2}^2(PE)\right)+\wh k_0}
\end{align*}
by which (\ref{prev2lbd}) in Lemma \ref{lem: mono} guarantees
\begin{equation}\label{lessthanr1} \frac{\|E-(PE)_{(2\wh k_0)\wedge N}\|^2}{nm-\wh \lambda_1\wh k_0}\le \frac{\|E-(PE)_{(2r)\wedge N}\|^2}{nm-\wh \lambda_1r}
\end{equation}
provided that $\wh k_0\le r$. On the other hand, 
\[
\wh \lambda_1  \overset{\E}{\ge} \frac{nm}{\|E-(PE)_{(2\wh k_0)\wedge N}\|^2/d_1^2(PE)+\wh k_0}
\]
which is equivalent with
\begin{equation}\label{lessthanr2}
d_1^2(PE)\le \frac{\|E-(PE)_{(2\wh k_0)\wedge N}\|^2}{nm/\wh \lambda_1-\wh k_0}.
\end{equation}
Combining (\ref{lessthanr1}) with (\ref{lessthanr2}) proves (\ref{lessthanr}). After repeating this argument, on the event $\E$, we find
\begin{eqnarray}\label{d1lessthanr}
d_1^2(PE)\le \frac{\|E-(PE)_{(2r)\wedge N}\|^2}{nm/\wh\lambda_t-r}\le \frac{\|Y-(PY)_r\|^2}{nm/\wh\lambda_t-r}
\end{eqnarray}
for any $t\ge 1$, which implies	$\wh k_{t}\le r$.

To conclude our proof, we   show that $\E$ holds with probability tending to $1$. Again, we only prove for $q \le m$ since the case of $q> m$ can be obtained similarly.
The proof of Lemma \ref{lem_exp_PE} shows that there exists a $q\times m$ matrix $Z$ with i.i.d. standard normal entries such that $d_k(PE) = \sigma d_k(Z)$ for any $1\le k\le q$. Without loss of generality, we assume $\sigma = 1$. We first consider the event $\{2d_1^2(Z) \le \wh \lambda_0\wh \sigma^2 \}$. From (\ref{eq_cc_PEk}) and using $\EE[d_1^2(Z)] \ge m$   in Lemma \ref{lem_exp_PE}, we have
\begin{equation}\label{eq_d1}
\PP\Big\{d_{1}^2(Z) \le (1+C_1)^2\EE\left[d_1^2(Z)\right]\Big\} \ge 1-\exp\left(-{C_1^2m/ 2}\right).
\end{equation}
Then,   using (\ref{chiL}), we have
\begin{align*}
\PP\left\{2d_1^2(Z) \le \wh \lambda_0\wh \sigma^2 \right\} 
& \ge 1-\PP\Big\{d_{1}^2(Z) \ge (1+C_1)^2\EE\left[d_{1}^2(Z)\right]\Big\}\\
&\qquad -  \PP\{\wh \sigma^2 \le (1-C_2)\}\\
& \ge 1 - \exp\left(-C_1^2m^2/2\right)-\exp\left(-C_2^2nm/4\right),
\end{align*}
for any $0<C_1, C_2<1$ satisfying $(1+C_1)^2/(1-C_2) = 1+\eps$.

Next we quantify the event $\E_k$ in (\ref{eq_E}) which is equivalent to
\[
\E_k := \left\{ 
{\|E-(PE)_{2k\wedge N}\|^2 \over d_1^2(Z) \vee \left[d_{2k+1}^2(Z) + d_{2k+2}^2(Z)\right] } \ge {(1-\eps)\wh R_t \over \wh U_t}	
\right\}.
\]
We will proceed to to show, that each of the random quantities concentrate around their means. We shall only consider when $n>q$ since $n=q$ is easy to obtain by using the same arguments and the fact $\|E-PE\|=0$ for $n = q$.

Fix $1\le k\le r$ and   consider two cases: 
\vspace{1mm}

(1) When $2k< q$, observe that 
$$\|E-(PE)_{2k}\|^2 = \|E-PE\|^2 + \|PE-(PE)_{2k}\|^2.$$ 
Since $\|E-PE\|^2$ has a central $\chi^2_{(n-q)m}$ distribution, inequality (\ref{chiL}) 
gives
\begin{equation}\label{chiQ}
\PP\left\{\|E-PE\|^2 \le (1-C_3)(n-q)m\right\} \le \exp\left(-C_3^2(n-q)m/4\right)
\end{equation}
for any $C_3 \in (0,1)$. This bound together with Lemma \ref{lem_svd_PE} yield 
\begin{align*}
\PP\Big\{
\|E-(PE)_{2k}\|^2 \le  (1-C_3)\wh R_t\sigma^2\Big\}&\le  \exp\left(-{C_3^2(n-q)m /4}\right)+\exp\left(-C_3'm\right)
\end{align*}
with some constant $C_3'>0$. 

On the other hand, recall that $\wh U_t = S_1 \vee (S_{2k+1} +  S_{2k+2})$ from (\ref{def_rut_supp}). 
If  $d_1^2(Z) \le d_{2k+1}^2(Z) +  d_{2k+2}^2(Z) $ such that $\wh U_t = S_{2k+1} + S_{2k+2}$, the concentration inequality (\ref{eq_cc_PEk}) in Lemma \ref{lem_exp_PE} gives
\begin{align*}
\PP\Big\{d_{2k+1}^2(Z) +  d_{2k+2}^2(Z) &\le \sum_{j=2k+1}^{2k+2}\left(\EE[d_{j}(Z)]+t\right)^2\Big\}\ge 1-2\exp(-t^2/2),
\end{align*}
for any $t\ge 0$.
Take $t^2 = (C_4^2/2)\sum_{j = 2k+1}^{2k+2}(\EE[d_{j}(Z)])^2$ with arbitrary $C_4\in (0,1)$ and invoke Lemma \ref{lem_tech} to obtain 
\begin{align*}
&\PP\Big\{d_{2k+1}^2(Z) +  d_{2k+2}^2(Z)  \le (1+C_4)^2\wh U_t\Big\}\ge 1-2\exp\left(-{C_4^2m/ 2}\right).
\end{align*}
For the exponent in the probability tail, we use $(\EE[d_j(Z)])^2\ge \EE[d_j^2(Z)]-1 = S_j -1$ from the proof of Lemma 1 in \cite{Giraud} and $S_1\ge m$ from Lemma \ref{lem_exp_PE} to obtain
\begin{align*}
\sum_{j = 2k+1}^{2k+2}(\EE[d_j(Z)])^2&\ge S_{2k+1} + S_{2k+2} - 2\ge  S_1- 2\ge m - 2.
\end{align*}
If $d_1^2(Z) \ge d_{2k+1}^2(Z) +  d_{2k+2}^2(Z) $, the concentration inequality (\ref{eq_d1}) gives a  similar result as above. Choose $1-\eps = (1-C_3)/(1+C_4)^2$ and conclude that
\begin{align*}
\PP(\E_k^c)&\le 2\exp\left (-{C_4^2m/ 2}\right)+ \exp\left(-{C_3^2(n-q)m/  4}\right)+\exp\left(-C_3'm\right),
\end{align*}
for any $2k < q$.
\vspace{1mm}

(2) If $2k \ge q$, we immediately have $\wh R_t = (n-q)m$ and $\wh U_t = S_1$. Therefore, (\ref{eq_d1}) and (\ref{chiQ}) give
\begin{align*}
\PP \{\E_k^c\} &= \PP\left\{
\frac{\|E-PE\|^2}{d_1^2(Z)} \le  \frac{(1-C_3)\wh R_t}{(1+C_4)^2\wh U_t}
\right\}\\
& \le \exp\left(-{C_4^2m/ 2}\right)+ \exp\left(-{C_3^2(n-q)m/  4}\right).
\end{align*}
Taking the union bound over all $1\le k\le r$ concludes our proof. \qed\\

\subsection{Proof of Theorem \ref{thm: SRS1}}
The choice $\wh\lambda_0= C(\sqrt{m}+\sqrt{q})^2$ is feasible by using $$\EE[d_1^2(Z)] \le (\EE[d_1(Z)])^2+ 1 \le (\sqrt{m}+\sqrt{q})^2+1$$ 
and choosing $\eps$ such that $$2(1+\eps)[(\sqrt{m}+\sqrt{q})^2+1] \le C(\sqrt{m}+\sqrt{q})^2.$$ 
This implies $\wt k \le \wh k$ from the proof of Proposition \ref{prop:monolbd}. 
Invoking Proposition \ref{prop:monolbd} once again concludes the proof of Theorem \ref{thm: SRS1}.
\qed \\

\subsection{\bf Proof of Theorem \ref{thm: signalSRS}}

We work on the event 
$$
\E' = \{\wh\sigma \le (1+\eps)\sigma\} \cap \E
$$
where $\E$ is defined in (\ref{eq_E}). From (\ref{chiR}) and the proof of Proposition \ref{prop:monolbd}, $\E'$ holds with probability tending to $1$. \\
Let $\wh k_0, \wh k_1,\ldots $ be the sequence of selected ranks from (\ref{k0}) and (\ref{kt}). 
Our assumption 
\begin{align*}
d_{k_0}(XA)&\ \ge\ C''\sigma\sqrt{\lambda_0}\ \ge\ (1+\eps)\sigma\sqrt{\lambda_0}\left({\sqrt{2} \over 2} + \sqrt{nm \over nm - \lambda_0r}\right) \\
&\ \overset{\E'}{\ge}\ d_1(PE) + {\|E\| \over \sqrt{nm/\lambda_0 - r}}~ \ge ~d_1(PE) + \frac{ \| Y-(PY)_r\|}{ \sqrt{nm/ \lambda_0 - r}}
\end{align*}
implies that $\wh k_0\in[ k_0,r]$ by Theorem \ref{thm:RS2} under the rank constraint (\ref{rankconstraint}). This, in turn, from the decreasing property of $\wh\lambda_t$ in Proposition \ref{prop:monolbd}, yields
$\wh\lambda_1\le \lambda_1$, where $\wh \lambda_1$ is computed from (\ref{lbdt}) by using $\wh k_0$. Moreover, $\lambda_1\le \lambda_0$ implies $r$ also satisfies (\ref{rankconstraint}) for $\lambda_1$. Hence,
\begin{align*}
d_{k_1}(XA) &\ \ge\   C''\sigma\sqrt{\lambda_1}\ \ge\ 
(1+\eps)\sigma\sqrt{\lambda_1}\left({\sqrt{2} \over 2} + \sqrt{nm \over nm - \lambda_1r}\right)\\
&\ \ge\ (1+\eps)\sigma\sqrt{\wh \lambda_1}\left({\sqrt{2} \over 2} + \sqrt{nm \over nm - \wh\lambda_1r}\right)
\end{align*}
and by the same argument above, we have $\wh k_1\in[ k_1,r]$.  Now the proof follows by repeating these arguments and the fact that $k_T=r$. \qed\\

\subsection{Proof of Theorem \ref{thm: SRS-snr}}
We work on the event 
$$
\Z = \{d_1(Z) = \sqrt{m}+\sqrt{q}\} 
$$
which holds almost surely as $N = q\wedge m \rightarrow \infty$ from the celebrated result of \cite{Bai-Yin}.

We only prove the result for $q\le m$ since the case $q>m$ can be  derived in a similar way.	From the first signal condition (\ref{eq_srs_q/2}), Theorem \ref{thm: signalSRS} immediately implies $\wh k_0 \ge q/2$. Without loss of generality, we assume $\wh k_0 = q/2$, and $q$ is even for notational simplicity. Step (\ref{lbdt}) implies the following update
\begin{equation}\label{eq_lbd1}
\lambda_1 = {nm \over (1-c)(n-q)m/(\sqrt{m}+\sqrt{q})^2 + q/2}.
\end{equation}
To show $\wh k_1 = r$, from (\ref{thm:RS2}), it suffices to show 
\begin{equation}\label{eq_c}
d_r(XA) \ge d_1(PE) + \sqrt{\lambda_1}\wh\sigma_r.
\end{equation}
Observe that, with high probability,
\begin{equation}\label{eq_c0}
\lambda_1 \wh \sigma_r^2 = {\|Y-(PY)_r\|^2 \over nm/\lambda_1 - r} \le {\|E\|^2 \over nm/\lambda_1 - r} \le {(1+\eps)nm(\sqrt{m}+\sqrt{q})^2\sigma^2 \over 7(n-q)m/8 - (r-q/2)(\sqrt{m}+\sqrt{q})^2}
\end{equation}
for any $\eps \in (0, 1/4)$, by choosing $c = 7/8$ and using (\ref{chiR}) in the last inequality. We consider two cases:\\
Case 1:
\begin{equation}\label{eq_c1}
{nm \over \lambda_0}\ge {1+\delta \over \delta}q \quad \iff \quad nm \ge {1+\delta \over \delta} 2C(\sqrt{m}+\sqrt{q})^2q.
\end{equation}
This implies $r\le q$ from the rank constraint (\ref{rankconstraint}). Hence, 
\[
\lambda_1 \wh \sigma_r^2 \le {(1+\eps)nm(\sqrt{m}+\sqrt{q})^2\sigma^2 \over 7(n-q)m/8 - (q/2)(\sqrt{m}+\sqrt{q})^2} \le 
{C(1+\eps)(1+\delta) \over 1 + 5\delta/16}(\sqrt{m}+\sqrt{q})^2\sigma^2
\]
by using (\ref{eq_c1}), $m/(\sqrt{m}+\sqrt{q})^2\le 1$ and $C \ge 8/7$. This further implies 
\[
d_1(PE) + \sqrt{\lambda_1}\wh\sigma_r \le C'\left[1+ \sqrt{1+\delta \over 1+5\delta/16}\ \right](\sqrt{m}+\sqrt{q})\sigma
\]
by taking $C' = \sqrt{C(1+\eps)}$. Thus, (\ref{eq_c}) holds under the assumed signal condition (\ref{eq_srs_r}).\\
Case 2: 
\begin{equation}\label{eq_c2}
{nm \over \lambda_0}\le {1+\delta \over \delta}q \quad \iff \quad nm \le {1+\delta \over \delta}\cdot 2C(\sqrt{m}+\sqrt{q})^2q.
\end{equation}
It follows that 
$$r\le {\delta \over 1+\delta}{nm\over \lambda_0} = {\delta \over 1+\delta}\cdot {nm \over 2C(\sqrt{m}+\sqrt{q})^2}.$$
Plugging the above upper bound of $r$ into (\ref{eq_c0}) yields
\begin{equation}\label{eq_c3}
\lambda_1 \wh \sigma_r^2 \le {(1+\eps)(\sqrt{m}+\sqrt{q})^2\sigma^2 \over 7/8-\delta/(2C(1+\delta))-(q/n)[7/8 - (\sqrt{m}+\sqrt{q})^2/(2m)]}.
\end{equation}
If $7/8 \le (\sqrt{m}+\sqrt{q})^2/(2m)$, we further have 
\[
\lambda_1 \wh \sigma_r^2 \le {(1+\eps)(\sqrt{m}+\sqrt{q})^2\sigma^2 \over 7/8-\delta/(2C(1+\delta))} \le {C(1+\eps)(1+\delta) \over 1+3\delta/8}(\sqrt{m}+\sqrt{q})^2\sigma^2
\]
by using $C\ge 8/7$. By repeating the same arguments before, (\ref{eq_c}) holds. Finally, we show (\ref{eq_c}) still holds when $7/8 \ge (\sqrt{m}+\sqrt{q})^2 / (2m)$. Recall that $\wh k_0 = q/2$ which implies 
\[
{q\over 2} \le {\delta \over 1+\delta} {nm\over \lambda_0}\quad \iff \quad {q\over n}\le {\delta m \over C(1+\delta)(\sqrt{m}+\sqrt{q})^2}.
\]
Combining this with (\ref{eq_c3}) gives 
\[
\lambda_1\wh\sigma_r^2 \le {C(1+\eps)(1+\delta) \over 1+\delta/8}(\sqrt{m}+\sqrt{q})^2\sigma^2.
\]
Repeating the same arguments proves (\ref{eq_c}), hence concludes the proof.\qed\\

\subsection{Proof of Proposition \ref{prop: ex-rank}}
Again, we work on the event $\Z$, defined in the proof of Theorem \ref{thm: signalSRS} and
we only prove the case $q\le m$ since the complementary case follows by the same arguments. 

Since $\wh k_0 \ge N/2$, we assume $\wh k_0 = N/2$ so this is the most difficult case. We take $N$ even for simplicity. Thus, step (\ref{lbdt}) implies the update of $\lambda$ as 
\[
\wh\lambda_1 = {nm \over (1-\eps)(n-q)m/(\sqrt{m}+\sqrt{q})^2 + q/2} ={nm \over (7/8)(n-q)m/(\sqrt{m}+\sqrt{q})^2 + q/2}  
\]   
by taking $\eps = 1/8$. On the one hand, a little algebra shows that
\[
\wh K_1 ~=~ {nm \over \wh\lambda_1} ~=~ {7(n-q)m \over 8(\sqrt{m}+\sqrt{q})^2} + {q \over 2}~ \ge~ {9\over 8}q
\]
by using our assumption
\begin{equation}\label{eq_K0}
{nm \over \lambda_0} = {nm \over 2C(\sqrt{m}+\sqrt{q})^2} \ge {3\over 4}q.
\end{equation}
with $C=8/7$ and 
$m/(\sqrt{m}+\sqrt{q})^2 \le 1$.
Hence $\wh k_2$ is selected from $[q/2,\ q]$ according to (\ref{kt}). On the other hand, similar as (\ref{eq_c0}) and by using $r\le q$, we have 
\begin{equation*}
\wh\lambda_1 \wh \sigma_r^2 \le  {(1+\eps')nm \over 7(n-q)m/8 - (q/2)(\sqrt{m}+\sqrt{q})^2}(\sqrt{m}+\sqrt{q})^2\sigma^2
\end{equation*}
with high probability for any $\eps' \in (0,1/4)$. Using (\ref{eq_K0}), $m/(\sqrt{m}+\sqrt{q})^2 \le 1$ and $C > 8/7$ yields 
\[
\wh \lambda_1 \wh \sigma_r^2 \le 12C(1+\eps')(\sqrt{m}+\sqrt{q})^2\sigma^2.
\]
Taking $C' = \sqrt{C(1+\eps')}$ concludes 
\[
d_1(PE) + \sqrt{\wh \lambda_1} \wh \sigma_r \le C'(1+2\sqrt{3})(\sqrt{m}+\sqrt{q})\sigma,
\]
which, by using our signal condition and invoking Theorem \ref{thm:RS2}, completes the proof.
\qed\\

\section{Proofs of Section 5}

The following lemmas are critical for extending the previous results to general errors with heavy tail distributions.
\begin{lemma}\label{lem_conc_E}
	Let $E\in \RR^{n\times m}$ have independent entries with mean zero, variance $\sigma^2$ and fourth moment $\gamma<\infty$. For any $\eps \in (0,1)$, we have
	\begin{eqnarray}\nonumber
	\PP\Bigl\{\bigl|\|E\|^2  - nm\sigma^2\bigr| >\eps nm\sigma^2 \Bigr\} &\le& {\gamma/\sigma^4-1 \over \eps^2nm }.
	\end{eqnarray}
\end{lemma}	

\begin{proof}[Proof of Lemma \ref{lem_conc_E}]
	Since $\|E\|^2 = \sum_{i=1}^n\sum_{j=1}^mE_{ij}^2$ with 
	$\EE[\|E\|^2] = nm\sigma^2$ and Var$(\|E\|^2) = nm(\gamma - \sigma^4)$, the result follows from
	the Bienaym\'e-Chebyshev inequality.
\end{proof}

\begin{lemma}\label{lem_op_interlacing}
	Let $E$ has independent entries with mean zero and unit variance. Assume $ n =O(m^\alpha)$ for some $\alpha \in [0,1)$. Then, for any $\eps \in (0,1)$, one has
	\begin{eqnarray}\label{eq_op_interlacing}
	(1-\eps)\sqrt{m} \le d_k(E)  \le  (1+\eps)\sqrt{m},\qquad \text{for all }1\le k\le n
	\end{eqnarray}
	with probability at least $1 - 2\exp(-c\eps^2m^{1-\alpha})$ where $c>0$ is some absolute constant. Moreover, for any $\eps\in (0,1)$, we have
	\begin{align}\label{eq_PE_k_sub}
	\|E-(E)_{k}\|^2 \ge &(1-\eps)^2(n-k)m,\qquad \text{for all $1 \le k\le n$}
	\end{align}
	with probability converging to $1$ as $m\rightarrow \infty$.
	Similar results hold for $m  = O(n^\alpha)$ with $m$ and $n$ switched.
\end{lemma}
\begin{proof}[Proof of Lemma \ref{lem_op_interlacing}]
	Let us first consider $ n =O(m^\alpha)$ for some $0\le \alpha <1$. Fix any $1\le k\le n$. Lemma \ref{lem_interlacing} implies $ d_{k}(\overline{E}_{k}) \le d_k(E) \le d_1(\overline{E}_{n-k+1})$ where $\overline{E}_{j}$ is the matrix made of the first $j$ rows of $E$ for any $1\le j\le n$. For notational simplicity, we write $F^j = (\overline{E}_{j})^T\in \RR^{m \times j}$. Notice that, $F^j$ still has independent entries, each row $F^j_{i\cdot}$ of $F^j$ has $j\times j$ identity covariance matrix and 
	$$
	\|F^j_{i\cdot}\|_2 = \sum_{\ell =1}^j \left(F^j_{i\ell}\right)^2 = \sum_{\ell = 1}^{j}E_{\ell i}^2  = j, \quad a.s.
	$$
	by the law of large number. Thus, we can invoke Theorem 5.41 in \cite{rv_rand_mat} for $F^{n-k+1}$ and $F^k$ to obtain
	\begin{eqnarray}\nonumber
	\PP\left\{d_{1}(F^{n-k+1}) \ge \sqrt{m}+t\sqrt{n-k+1}\right\}  &\le& (n-k+1)\exp(-c_0t^2),\quad t\ge 0,\\\nonumber
	\PP\left\{d_{1}(F^k) \le \sqrt{m}-t\sqrt{k}\right\}  &\le& k\exp(-c_1t^2),\quad t\ge 0,
	\end{eqnarray}
	for some constant $c_0, c_1>0$. For any $\eps \in (0,1)$, choose $t = \eps\sqrt{m}/\sqrt{n-k+1}$ in the first display and $t = \eps \sqrt{m}/\sqrt{k}$ in the second display, and get
	\begin{align*}
	&\PP\left\{(1-\eps)\sqrt{m} \le d_k(E) \le  (1+\eps)\sqrt{m}\right\}\\
	&\qquad  \ge 1 - \exp\left(-{c_0\eps^2m \over n-k+1} + \log (n-k+1)\right)-\exp\left(-{c_1\eps^2m \over k} + \log k\right)\\
	&\qquad \ge  1- 2\exp\left(-c_2\eps^2m^{1-\alpha} \right)
	\end{align*}
	for some constant $c_2> 0$. We use $n=O(m^\alpha)$ in the last inequality.
	Taking the union bound over $1\le k\le n$ concludes the proof of (\ref{eq_op_interlacing}).  
	
	If $m = O(n^\alpha)$, write $F = E^T$. Using Lemma \ref{lem_interlacing} again gives $ d_{k}(\overline{F}_{k}) \le d_k(F) \le d_1(\overline{F}_{n-k+1})$ for any $1\le k\le m$. We observe that $(\overline{F}_j)^T \in \RR^{n \times j}$ has independent entries, identity covariance matrix and the $\ell_2$ norm of each row equal to $\sqrt{j}$ almost surely. Repeating the previous arguments will prove (\ref{eq_op_interlacing}).
	
	We proceed to show (\ref{eq_PE_k_sub}) for $n=O(m^\alpha)$ only since the case $m = O(n^\alpha)$ can be easily extended. For any $1\le k \le n$, notice that $\|E-(E)_k\|^2 = \sum_{j=k+1}^nd_j^2(E)$. (\ref{eq_PE_k_sub}) follows immediately from (\ref{eq_op_interlacing}) and Lemma \ref{lem_conc_E}. This completes the proof.
\end{proof}
\vspace{-1mm}

\subsection{Proof of Theorem \ref{thm: heavy tails n = q}}
Without loss of generality, we assume $\EE[E_{ij}^2] =1$. We start by defining the following event 
\[
\E'' := \{|\|E\|^2 - nm|\le \eps nm\} \cap \{d_1(E) = \sqrt{m}+\sqrt{n} \}
\]
for any $\eps \in (0,1)$.	By Lemma \ref{lem_conc_E} and \cite{Bai-Yin}, we have  $\PP(\E'') \to 1$ as $n,m\to \infty$.
To show $\wh k \le r$, from Theorem \ref{prop:ongelijk}, it suffices to show (\ref{RS3}). This is indeed the case since
\begin{equation}\label{eq_d}
{2d_1^2(PE)\over \wh \sigma^2} \le {2d_1^2(E)\over \wh \sigma^2} \overset{\E''}{\le} {2(\sqrt{m}+\sqrt{n})^2 \over 1-\eps}= {2(\sqrt{m}+\sqrt{q})^2\over 1-\eps} (1+ o(1)) 
\end{equation} 
by using $d_1(PE) \le d_1(P)d_1(E) = d_1(E)$ in the first inequality and choosing $\lambda =  2(\sqrt{m}+\sqrt{q})^2/(1-\eps)$. On the other hand, to show $\wh k \ge s$, from Theorem \ref{thm:RS2}, we need to verify (\ref{RS2}). By using (\ref{eq_sigma_r}), (\ref{rankconstraint}) and (\ref{eq_d}), it follows that 
\[
d_1(PE) + \sqrt{\lambda}\wh\sigma_r \le \sqrt{\lambda}\left[{1\over \sqrt{2}}+\sqrt{1+\delta} \right]\wh\sigma \overset{\E}{\le} C\sigma(\sqrt{m}+\sqrt{q}) \le d_s(XA)
\]
by choosing $C = \sqrt{(1+\eps)/(1-\eps)}(1+\sqrt{2(1+\delta)})$. This completes the proof. \qed\\

\subsection{Proof of Theorem \ref{thm: heavy tails skinny}}
Without loss of generality, we assume $E_{ij}$ has unit variance.
As in the proof of Theorem \ref{thm: signalSRS}, we first need to check the following two properties for $\wh \lambda_t$ and $\wh k_t$ obtained in (\ref{initlbd_ex}) and (\ref{lbdt_ex}): (1) $\wh \lambda_t$ is decreasing; (2) $\wh k_t \le r$. Obviously, (1) is straightforward from (\ref{initlbd_ex}) and (\ref{lbdt_ex}). To show (2), we only prove the case $n=O(m^\alpha)$ since the case $m = O(n^\alpha)$ is similar to derive.

We define the following event which is analogous to (\ref{eq_E}) in the proof of Proposition \ref{prop:monolbd}. 
\begin{align*}
\E''' &:= \bigcap_{k=1}^q\left\{ 
{\|E-(E)_{(2k)\wedge n}\|^2 \over d_1^2(E) \vee \left[d_{2k+1}^2(E) + d_{2k+2}^2(E)\right] } \ge {(1-\eps)[n/2 - k]_+}	
\right\}\\ 
&~~\quad \bigcap\ \left\{2d_1^2(E) \le \wh \lambda_0\wh \sigma^2 \right\}
\end{align*}
On the event $\E'''$, $\wh k_t \le r$ follows by the same arguments used in the proof of Proposition \ref{prop:monolbd} except $PE$ and $PY$ are now replaced by $E$ and $Y$, respectively.   Thus, it suffices to show $\E'''$ holds with high probability. First, observe that, for any $\eps_1, \eps_2 \in (0,1)$, and some constant $c, c'$,
\begin{align*}
\PP\left\{2d_1^2(E) \le \wh \lambda_0\wh \sigma^2 \right\} &\ge 1 - \PP\{d_1^2(E) \ge (1+\eps_1)^2m \} - \PP\{\wh \sigma^2 \le (1-\eps_2)nm \}\\
& \ge 1- \exp(-c\eps_1^2m) - c'(\eps_2^2nm)^{-1}.
\end{align*}
We use Lemmas \ref{lem_conc_E} and \ref{lem_op_interlacing} in the second inequality. Choosing $\eps$ in $\wh \lambda_0$ such that $1+\eps = (1+\eps_1)^2/(1-\eps_2)$ proves the last event in $\E'''$. For the other intersect event, the event holds trivially if $2k \ge q$. If $2k < q$, Lemma \ref{lem_op_interlacing} guarantees $\E'''$ holds with probability tending $1$.

Finally, the results of Theorem \ref{thm: heavy tails n = q}   follow from the same arguments in the proof of Theorem \ref{thm: signalSRS} except  $PE$ and $PY$are now replaced by $E$ and $Y$, respectively.\qed\\

\subsection{Proofs of Theorems \ref{thm: special model1} and \ref{thm: special model2}}	
Recall that we use $d_1(PE)\le d_1(E)$ in the proof of Theorems \ref{thm: heavy tails n = q} and \ref{thm: heavy tails skinny}. Thus, the proofs for Theorems \ref{thm: special model1} and \ref{thm: special model2}	remain the same as those
for Theorems \ref{thm: heavy tails n = q} and \ref{thm: heavy tails skinny}. \qed\\

\section{Oracle inequality}\label{app_oracle}
While our main interest is the study of the rank estimator $\wh k$, we briefly mention an oracle inequality for our estimators
$X\wh A := (PY)_{\wh k}$ based on GRS and STRS of the mean $XA$.
\begin{thm}\label{GeneralOracle}
	Let $C>2$. 
	On the event
	$\lambda \ge Cd_1^2(PE)/\wh\sigma^2$, we have
	\[\|X\wh A - XA\| ^2 \le \frac{C+2}{C-2} \min_{0\le k\le K_\lambda} \left\{\left( \frac{C+2}{C-2} +8(\rho -1)\right) \sum_{j>k} d_j^2(XA)+3\rho \lambda \wh \sigma^2k \right\},\]
	with $\rho := (nm)/(nm-\lambda K_\lambda)$.
\end{thm}
\begin{proof}
	First we notice that
	$$\frac{\| Y- (PY)_{\wh k}\| ^2}{nm-\lambda\wh k}\leq \frac{\| Y-(PY)_k\| ^2}{nm-\lambda k} \leq \frac{\| Y-(XA)_k\| ^2}{nm-\lambda k}.$$
	The first inequality follows from the optimality of $\wh k$; the second inequality is a consequence of    Pythagoras' identity and the \cite{EckartY} inequality. Rewriting the above display yields
	\begin{eqnarray*}
		\| Y- (PY)_{\wh k}\| ^2&\leq& 
		\| Y- (XA)_k\| ^2 \left\{ 1+  \frac{2\lambda k}{nm-\lambda k} -\frac{\lambda (k+\wh k)}{nm-\lambda k} \right\},
	\end{eqnarray*}
	and after working out the squares, we obtain
	\begin{align}\nonumber
	\| (PY)_{\wh k} -XA \| ^2 &\ \leq\ \|XA-(XA)_k\|^2  +\frac{2\lambda k}{nm-\lambda k}\| Y-(XA)_k \| ^2\\\label{(4)}
	&\quad\ -\frac{\lambda (k+\wh k)}{nm-\lambda k}\| Y- (XA)_k\| ^2
	+ 2\langle E,  (PY)_{\wh k}- (XA)_k \rangle.
	\end{align}
	Since both  $(XA)_k$ and $(PY)_{\wh k}$ are in the column space of $P$, see, for instance, \cite[page 124]{GiraudBook}, we have
	$2\langle E,  (PY)_{\wh k}- (XA)_k \rangle= 2\langle PE,  (PY)_{\wh k}- (XA)_k \rangle$, and
	the norm duality and the elementary inequality $2xy\leq ax^2+y^2/a$ for $a>0$, yield
	\begin{align}
	\label{(5)}\nonumber
	2\langle E, (PY)_{\wh k}&-(XA)_k\rangle 
	\leq\  
	2d_1(PE)\cdot\sqrt{k+\wh k}\cdot\| (PY)_{\wh k} -(XA)_k \| \\
	&\leq\ (a+b)(k+\wh k)d_1^2(PE) + \frac{1}{a}\| (PY)_{\wh k}-XA\| ^2+\frac{1}{b}\| XA-(XA)_k \| ^2
	\end{align}
	for any $a,b>0$.
	Moreover, we find
	\begin{eqnarray}\label{(6)}\nonumber
	\| Y-(XA)_k \| ^2& = & \| XA-(XA)_k \| ^2 + \| E\| ^2 + 2\langle E, XA-(XA)_k \rangle\\ &\le&
	3\|  XA-(XA)_k\| ^2 +\frac{3}{2}\| E\| ^2
	\end{eqnarray}
	and
	\begin{eqnarray}\label{(7)} 
	\| Y-(XA)_k \| ^2&\geq&  
	\frac12 \| E\| ^2- \| XA-(XA)_k\| ^2
	\end{eqnarray}
	Finally, combining (\ref{(4)}) with (\ref{(5)}), (\ref{(6)}) and (\ref{(7)}) gives
	\begin{align*}
	\frac{a-1}{a}\| X\wh A-&XA\| ^2
	\leq\left[ \frac{b+1}{b}+\frac{6\lambda k}{(nm-\lambda k)}
	+\frac{\lambda(k+\wh k)}{(nm-\lambda k)}\right]\| XA-(XA)_k\| ^2\\
	&+\frac{3\lambda k}{nm-\lambda k}\|E\|^2+ (k+\wh k)\left[(a+b)d_1^2(PE) -\frac{ \lambda/2}{nm-\lambda k}\|E\|^2\right].
	\end{align*}
	By using 
	$(nm)/(nm-\lambda \ell) \le (nm)/(nm-\lambda K_\lambda)$ for all $\ell \le K_\lambda$ and setting $a=1+b$ and $b= {C}/{2}$, the claim
	follows after a little algebra.
\end{proof}
\vspace{3mm}

Contrary to the rank consistency results in Sections \ref{sec_GRS} and \ref{sec_SRS}, no lower bound on the non-zero singular values $d_j(XA)$ of the signal $XA$, nor   any assumption on $X$ is required. It is clear that our selected estimator achieves the optimal bias-variance tradeoff, as discussed in   \cite{BSW,Giraud,GiraudBook}. 
In the above oracle inequality, the quantity $$\rho:={nm \over nm-\lambda K_\lambda}$$ could be  large under some circumstances. Nevertheless, the   mild rank constraint (\ref{rankconstraint}) with $\delta>0$  guarantees the upper bound  $ \rho \le 1+\delta$. Theorem \ref{GeneralOracle} and the exponential inequalities in (\ref{d_1}) -- (\ref{chiR}) immediately yield the following corollary.

\begin{cor}\label{cor:oracle}
	Suppose $E$ has i.i.d. $N(0,\sigma^2)$ entries.
	For $ \lambda > 2(\sqrt{m}+\sqrt{q})^2$, the event
	\begin{eqnarray}
	\label{ob}
	\|X\wh A - XA\| ^2 \lesssim \rho \min_{0\le k\le K_\lambda} \left\{  \sum_{j>k} d_j^2(XA)+ \lambda  \sigma^2k \right\} \end{eqnarray}
	holds with probability tending to 1, as $nm\to\infty$ and $q+m\to\infty$.	
\end{cor}
We use the notation $\lesssim$ for inequalities that hold up to multiplicative constants. The probability of event (\ref{ob}) converges (to 1) exponentially fast in $mn$ and $m+q$.

\begin{proof}
	We can write $\lambda := C' (\sqrt{m}+\sqrt{q})^2$ for some 
	$C' = C(1+C_0)^2/(1-C_1)>2$ with $C_0>0$,  $0<C_1<1$ and $C$ equal to the one in Theorem \ref{GeneralOracle}. From inequalities (\ref{d_1}) -- (\ref{chiR}), we have
	\begin{align*}
	\PP\left\{C d_1^2(PE)  \ge \lambda \frac{\|E\|^2}{nm} \right\}
	& \le \  
	\PP\left\{	d_1^2(PE)\geq (1+C_0)^2(\sqrt{m}+\sqrt{q})^2\sigma^2
	\right\}\\
	& \quad +\  \PP\left\{ \|E\|^2\le (1-C_1)nm\sigma^2\right\}\\
	& \le \ \exp \left\{- C_0^2 (\sqrt{m}+\sqrt{q})^2 /2\right\}+\exp\left\{- 
	C_1^2nm /4\right\}.
	\end{align*}
	This proves the claim. 
\end{proof}

For GRS and STRS, $\wh k=r$ holds with overwhelming probability. On this event, Theorem \ref{thm: oracle r given} provides a cleaner and tighter bound for the fit $\|X\wh A-XA\|$.

\begin{thm}\label{thm: oracle r given}
	On the event $\wh k = r$, we have
	\[
	\|X\wh A-XA\|^2 \le 
	4 rd_1^2(PE)
	\]
\end{thm}
\begin{proof}
	From the optimality of $X\wh A$, we have 
	\[ \|Y-X\wh A\|^2 \le   \|Y-(XA)_{r}\|^2 = \| E\|^2
	\] on the event $\{\wh r=r\}$. This implies 	\[
	\|X\wh A-XA\|^2 \le  2|\langle E, X\wh A- (XA)_{r}\rangle|  \le  2\|X\wh A- XA\| ( \sqrt{r} d_1(PE) )
	\]and the result follows.
\end{proof}

Specialized to our STRS procedure, we have the following corollary.
\begin{cor}\label{cor: oracle r given}
	Assume $E$ has i.i.d. $N(0,\sigma^2)$ entries. On the event (\ref{signalcondmsrs}) in Theorem \ref{thm: signalSRS}, if we choose and update $\lambda$ according to (\ref{initlbd}) and (\ref{lbdt}), then we have 
	$\wh k = r$ with overwhelming probability. Hence, for some constant $C>4$, we have
	\[
	\PP\left\{
	\|X\wh A-XA\|^2  \le Cr(\sqrt{q}+\sqrt{m})^2\sigma^2
	\right\}\to 1
	\]
	as $n\to\infty$ and $(q\vee m)\to\infty$.
\end{cor}
The convergence rate in Corollary \ref{cor: oracle r given} is exponentially fast. The proof follows immediately from Theorem \ref{thm: signalSRS}, Theorem \ref{thm: oracle r given}, (\ref{d_1}) and $\EE[d_1(PE)] \le \sigma(\sqrt{m}+\sqrt{q})$.\\

\section{STRS with deterministic bounds}\label{app_DB}

When $E$ has i.i.d. $N(0,\sigma^2)$ entries, a deterministic bound for updating $\lambda$ in Section \ref{sec_SRS} can be derived as follows. We define 
$M:= m\vee q$ and recall that $N=\qm$. 
For any $0<\eps<1$, we choose 
$\wt \lambda_0 = 2(1+\eps)(\sqrt{m}+\sqrt{q})^2$ and let $\wt k_0$ be selected from (\ref{vier}) by using $\wt \lambda_0$. For given $\wt k_t\ge 1$ with $t \ge 0$, we set
\begin{align}\label{lbda2}
\wt \lambda_{t+1} &~ =~ 
\left\{\begin{array}{ll}
nm \Big/ \left[(1-\eps)\wt R_t / \wt U_t+\wt k_t\right], & \text{ if }2\wt k_t \le N;\\
nm \Big/ \left[(1-\eps)(n-q)m/((\sqrt{m}+\sqrt{q})^2+1)+\wt k_t\right], & \text{ if }2\wt k_t \ge N 
\end{array}
\right.
\end{align}
with
\begin{align}\label{wtR}
\wt R_{t} := \max &\bigg\{nm-\sum_{j=1}^{2\wt k_t}\left(\sqrt{M}+\sqrt{N-j+1}\right)^2\!\!\!- 2\wt k_t,\\\nonumber
&\quad (n-q)m+\!\!\!\sum_{j=2\wt k_t+1}^{N}\!\!\!\!\left(\sqrt{M}-\sqrt{j}\right)^2\!\bigg\}
\end{align}
and
\begin{equation}\label{wtU}
\wt U_t := \max\left\{\!(\sqrt{m}+\sqrt{q})^2+1, \sum_{j = 2\wt k_t+1}^{2\wt k_t + 2}\rs\left(\sqrt{M}+\sqrt{N-j+1}\right)^2+2\right\}.
\end{equation}
After updating $\wt \lambda_t$, we select $\wt k_t$ as
\[
\wt k_t := \argmin_{\wt k_{t-1} \le k \le \wt K_t}{\|Y-(PY)_{k}\|^2 \over nm - \wt \lambda_t k}
\]
where $\wt K_t := \floor{nm/\wt\lambda_t} \wedge q\wedge m$. The procedure stops when $\wt k_{t} = \wt k_{t+1}$. For this deterministic self-tuning procedure, we show that results analogous to Proposition \ref{prop:monolbd} and Theorem \ref{thm: signalSRS} are still guaranteed.  

\begin{prop}\label{prop: lbdgaussian}
	For all $t\ge 0$, we have $\wt \lambda_{t+1} \le \wt \lambda_t$ and $\wt k_{t+1}\ge \wt k_t$. Moreover, $\wt k_t \le r$ holds with probability converging to $1$ as $(q\vee m)\rightarrow \infty$ and $n\rightarrow \infty$.
\end{prop}

\begin{proof}[Proof of Proposition \ref{prop: lbdgaussian}]
	
	We first show $\wt\lambda_t$ is decreasing in $\wt k_t$. For notational simplicity, we write $k$ for $\wt k_t$ and assume $q\le m$. The case $q > m$ can be obtained similarly. The decreasing property of (\ref{lbda2}) is easily seen if $2k\ge N$. Hence we focus on the case $2k<N$ and consider two sub-cases:
	\vspace{1mm}
	
	(1) Suppose $\wt U_t =\sum_{j = 2k+1}^{2 k + 2}(\sqrt{m}+\sqrt{q-j+1})^2+2$. It suffices to show 
	\[
	A_k:= {(1-\eps)\wt R_t \over \sum_{j = 2k+1}^{2 k + 2}(\sqrt{m}+\sqrt{q-j+1})^2+2} + k
	\]
	is increasing in $k$. It has the same form as (\ref{increasingak}) in the proof of Proposition \ref{prop:monolbd}. By repeating the arguments there, we can similarly show that $A_k$ is increasing. Hence $\wt \lambda_t$ is decreasing.
	\vspace{1mm}
	
	(2) Suppose $\wt U_t = (\sqrt{m}+\sqrt{q})^2+1$.
	This is similar to (\ref{increasingbk}). Therefore, the same arguments can be used to prove that $\wt\lambda_k$ is decreasing. By the same reasoning as the proof of Proposition \ref{prop:monolbd}, we have $\wt k_{t+1} \ge \wt k_t$. Finally, the proof for $\wt k_t\le r$ follows immediately from Proposition \ref{prop:monolbd} by observing that $\wt \lambda_t \ge \wh \lambda_t$ using  Lemma \ref{lem_exp_PE}. 
\end{proof}
\vspace{-1mm}


\begin{thm}\label{signalSRS2}
	Assume $E$ has i.i.d. $N(0,\sigma^2)$  entries. For any subsequence $k_0< \cdots <k_T=r$ of $\{1,2,\ldots,r\}$ with  $T\le r-1$, 
	we let $\lambda_0=2(1+\eps)(\sqrt{m}+\sqrt{q})^2$ for any $\eps \in (0,1)$. Denote by $\lambda_t$ the updated $\lambda$ according to (\ref{lbda2}) by using $k_t$, for $t=0,1,\ldots,T-1$. On the event,
	\begin{eqnarray}\label{signalcondmsrs2}
	d_{k_t}(XA)\ge C\sigma\sqrt{\lambda_t}\left[{\sqrt{2} \over 2} + \sqrt{ nm \over nm- \lambda_t r} \right], \quad t = 0,1,\ldots ,T.
	\end{eqnarray}
	for some $C>1$, there exists $T'\le T$ such that 
	$\PP\{\wt k_{T'} = r\}\to 1$, where $\wt k_0, \wt k_1, \ldots, \wh k_{T'}$ are the selected ranks from the procedure above.
\end{thm}
\begin{proof}[Proof of Theorem \ref{signalSRS2}]
	The proof follows the same arguments as the proof of Theorem \ref{thm: signalSRS} by using Proposition \ref{prop: lbdgaussian}.
\end{proof}
\vspace{3mm}

\section{Additional simulations}
\subsection{Simulations of Monte Carlo simulations vs deterministic bounds}\label{sec_sim_STRS_DB}

Depending on whether we update  $\lambda$ by Monte Carlo simulations (MC) or by the Deterministic Bounds (DB) in Section \ref{app_DB}, we denote by STRS-MC the procedure using MC and by STRS-DB the one using DB. We compare their performance for some non-Gaussian distributions. In particular, we generate $E$ from either uniform$(-\sqrt{3},\sqrt{3})$ or $t_\nu$-distribution (d.f. $\nu= 6$). For each distribution, both low- and high-dimensional settings are considered. The low-dimensional case sets $\eta = 0.1$, $b_0 = 0.1$, $n = 300$, $m = p = q = 50$ and varies $r$ between $0$ and $15$. For the high-dimensional case, we consider $\eta = 0.1$, $b_0 = 0.003$, $n = 200$, $m = 60$, $p = 300$, $q = 30$ and $0 \le r\le 15$. The mean selected ranks and rank recovery rate of the two methods for each setting are shown in Figure \ref{fig_MC_DB}. 

\begin{figure}[ht]
	\centering
	\vspace{-3mm}
	\begin{tabular}{cc}
		\includegraphics[width=.4\linewidth]{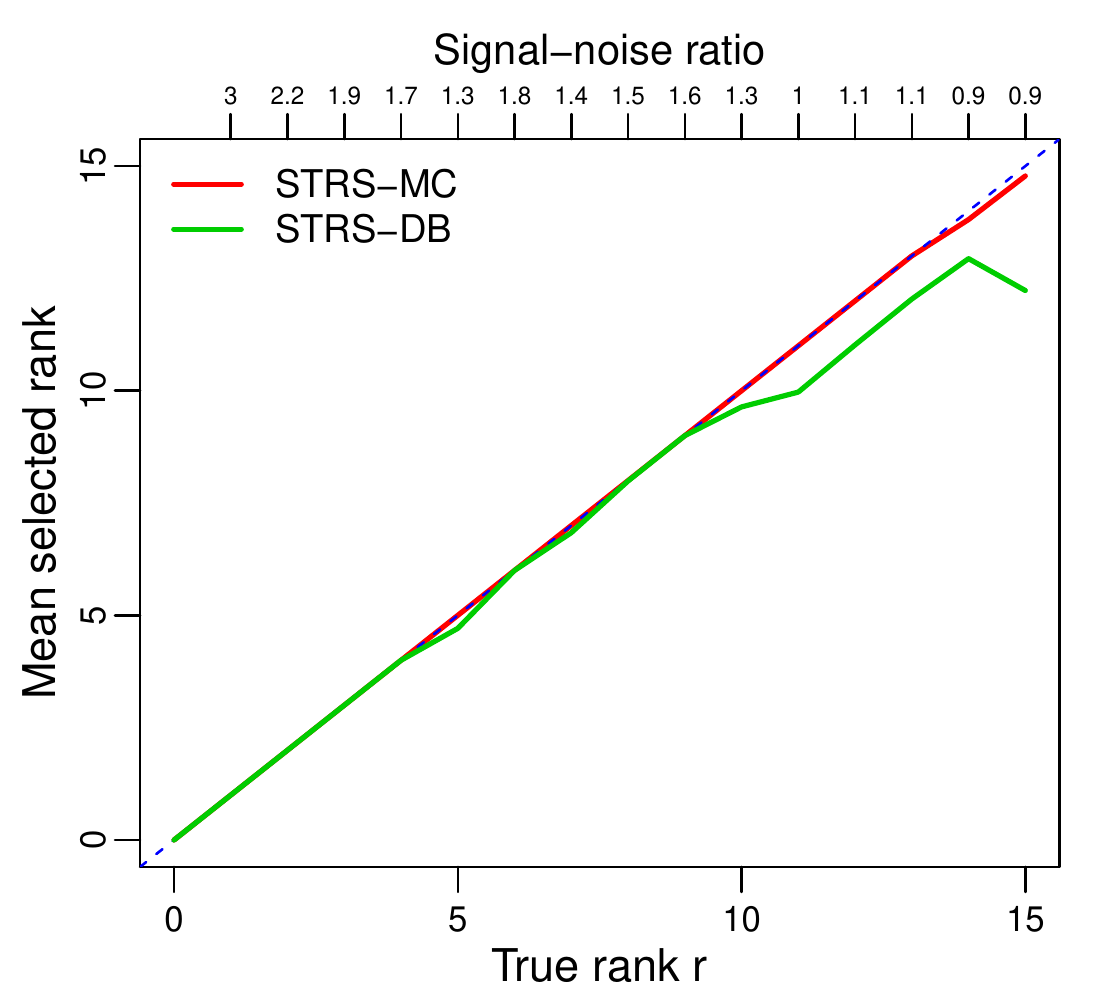} & 	\includegraphics[width=.4\linewidth]{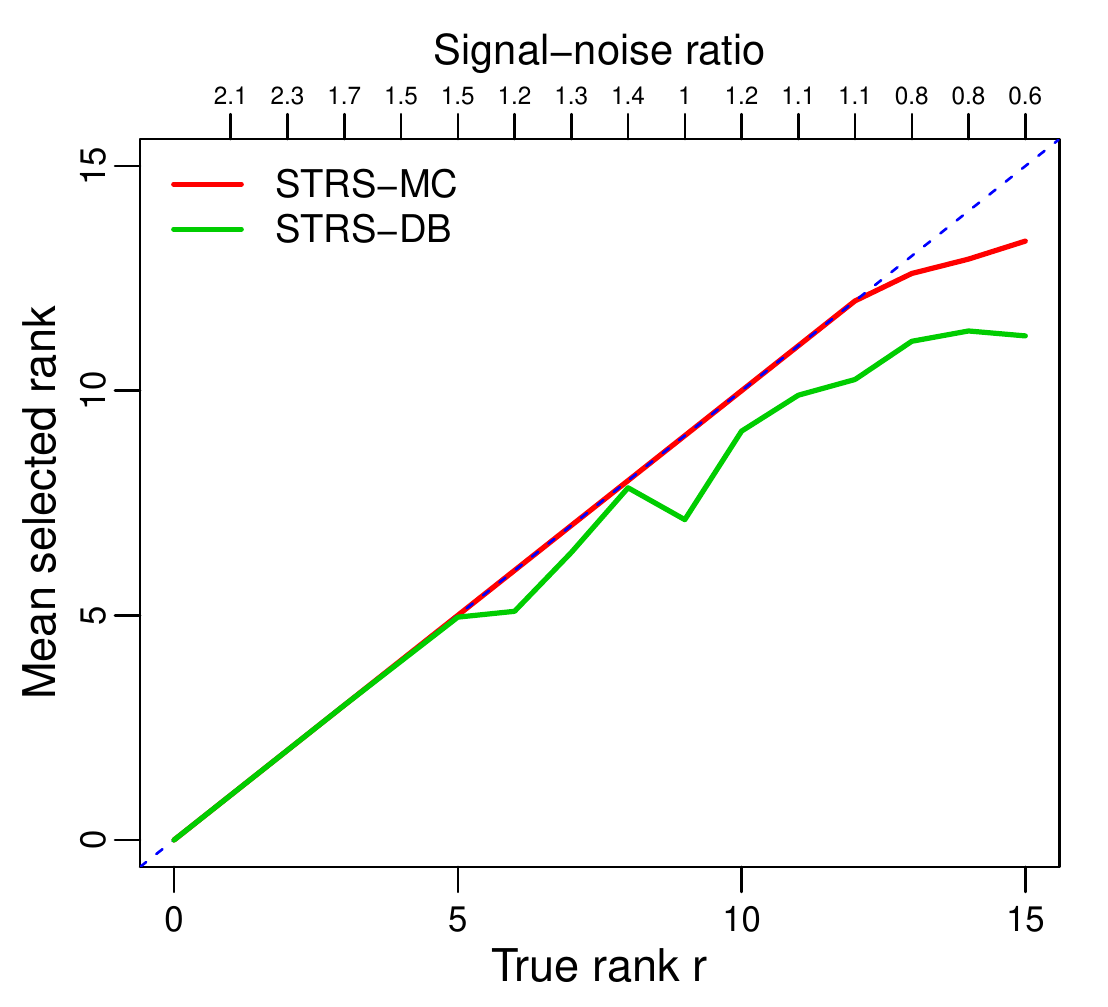}\\[-4pt]
		\includegraphics[width=.4\linewidth]{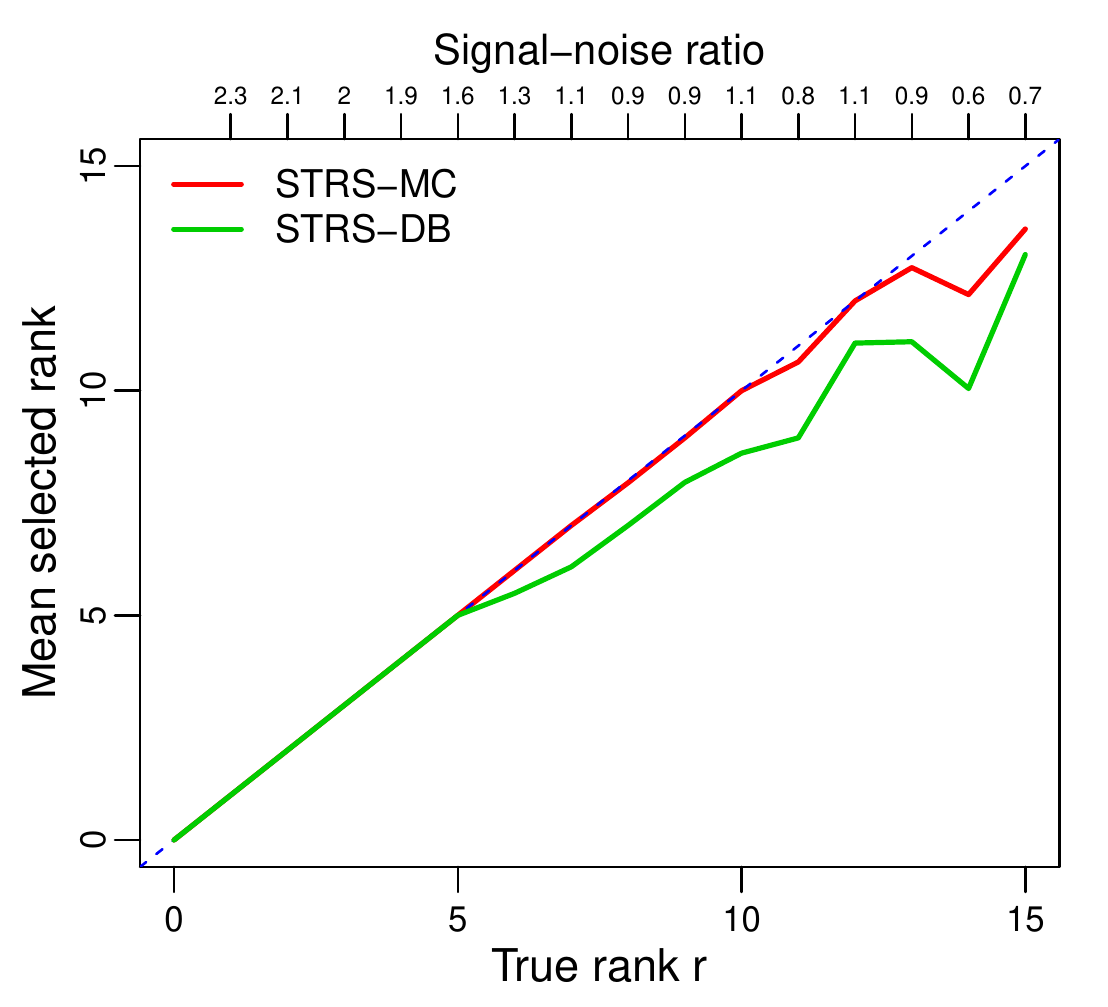} & 	\includegraphics[width=.4\linewidth]{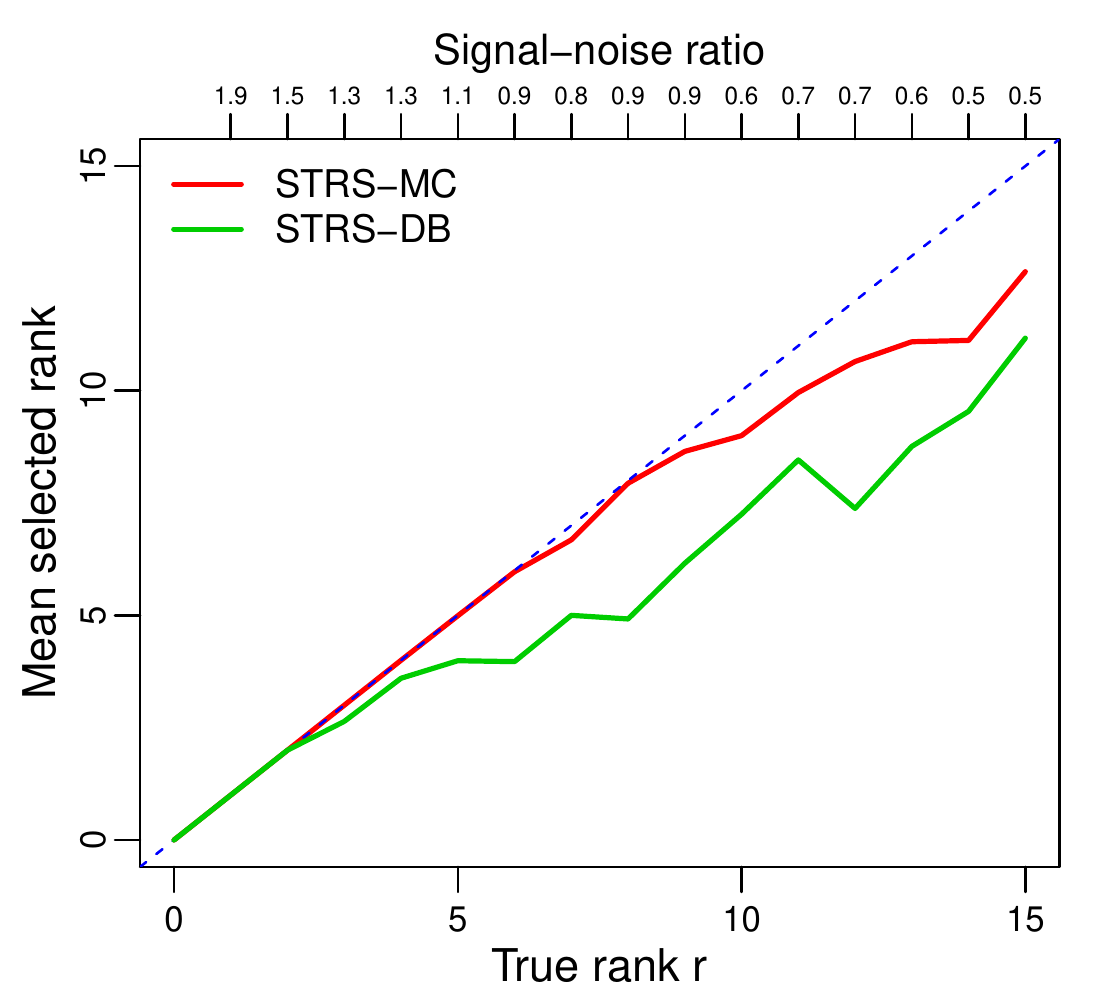}
	\end{tabular}
	\caption{ Comparison of STRS-MC and STRS-DB in the low-dimensional setting (top) and high-dimensional setting (bottom) for the uniform errors (left) and $t_6$ errors (right).}
	\label{fig_MC_DB}
	\vspace{-2mm}
\end{figure}

Both STRS-MC and STRS-DB work perfectly when the SNR is not small (say greater than $1.6$). STRS-DB seems to require a slightly larger SNR, as expected, since the deterministic bounds are upper bounds of those in STRS-MC leading to a larger updated $\lambda$.

More importantly, we emphasize that STRS-MC is using Monte Carlo simulations based on $N(0,1)$. It also supports our conjecture that STRS-MC works for  other distributions with heavier tails.

\subsection{Simulation for KF methods}\label{sec_sim_KF}
\cite{Giraud} proposed to minimize
\begin{eqnarray}\label{drie}
\frac{ \| Y- (PY)_k \|^2 }{ nm-1 - C \left( \EE [ \| G \|_{(2,k)} ] \right)^2}
\end{eqnarray}
with $ \| G \|_{2,k}^2 = \sum_{i=1}^k d_i^2 (G)$ and some  tuning parameter $C>1$. It does not require to estimate $\sigma^2$ but still needs the selection of leading constant $C$. \cite{Giraud} recommended to use $C=2$ based on \cite{BM}. We use KF to denote this method. In particular, we define KF-2 for choosing leading constant $C = 2$ and KF-CV for choosing $C$ via cross-validation. The simulation setting is slightly different from the low dimensional one in Experiment 1 and aims to give an example of an easy situation where KF fails to select the correct rank. Here we have $n=300$, $m=40$ and $p=q=35$. For illustration, we only vary the true rank from $10$ to $35$. We set $\eta=0.1$ and $b_0=20$. Note $b_0=20$ ensures a large SNR which should be the most ideal case for rank recovery. We compare the performance of KF-2, KF-CV and STRS. The plot of rank recovery and mean selected rank versus the true rank for different methods are shown in Figure \ref{KFplot}. From the result, it is clear that neither KF-2 nor KF-CV consistently recovers the true rank while STRS does. But since KF was not developed for rank recovery, nor did it claim to have this property, this does not contradict the results in \cite{Giraud}.

\begin{figure}[ht]
	\centering
	\vspace{-3mm}
	\begin{tabular}{cc}
		\centering
		\includegraphics[width=0.42\linewidth]{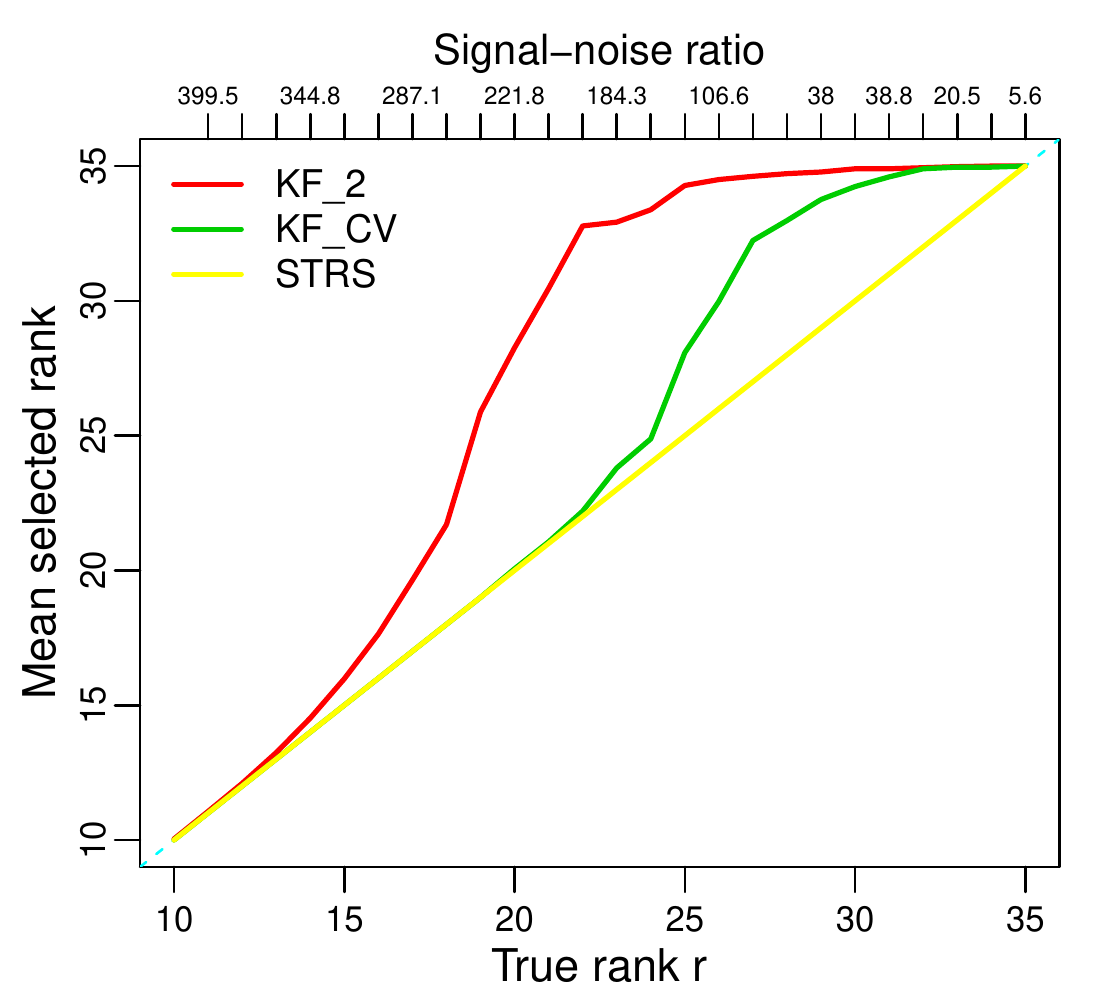}&
		\includegraphics[width=0.42\linewidth]{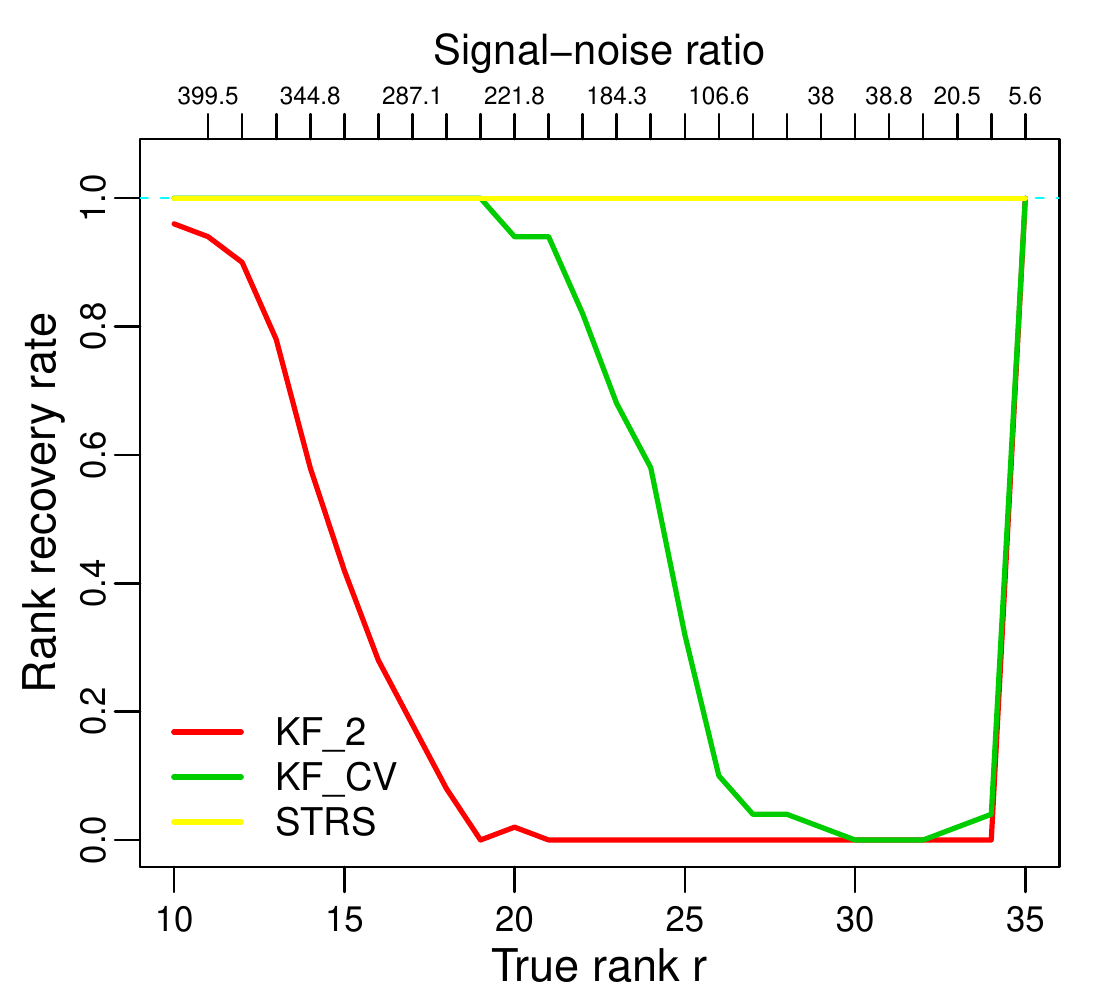}		
	\end{tabular}
	\caption{Plot of rank recovery rate and mean selected rank for KF and STRS. }
	\label{KFplot}
	\vspace{-3mm}
\end{figure}

\subsection{Simulations to  compare   the error  $\|X\wh A-XA\|$} 
\label{fitpred}

In this section, we compare three methods, STRS-MC, BSW-1.3 and KF-2, based on two criterions: (1) the fit $\|X\wh A-XA\|/\sqrt{nm}$; (2) the selected rank.

In the low-dimensional setting, when $X^TX$ is invertible, we also compute the prediction error $\|\wh A-A\|/\sqrt{pm}$. We consider both situations when the model is correctly specified and when the model approximately holds. 

\subsubsection{Exact low rank model}\label{elr}
We first consider the scenario when $A$   has an exact low rank structure, in both low- and high-dimensional settings. In the low-dimensional setting, we choose $\rho = 0.1$, $n = 200$, $m = p = q = 50$, $r = 10$ and $b_0\in\{0.02, 0.022, 0.024, \ldots, 0.046\}$. In the high-dimensional setting, we specify $\rho = 0.1$, $n = 150$, $p = 300$, $m = q = 50$, $r = 10$ and $b_0\in \{0.0015, 0.0016, \ldots, 0.003\}$. Different grids of $b_0$ are chosen to maintain  similar signal-to-noise ratio. 
The random additive errors in both settings are generated from $N(0,1)$. Within each setting, we repeat $100$ times. The averaged results are reported in Figure \ref{fig_correct} demonstrating how  criterions (1) and (2) of the three methods vary with $b_0$ in both low- and high-dimensional settings. 
STRS-MC dominates the other two methods as it produces a smaller error and   it always selects a rank closer to the true rank than the other two methods.
\begin{figure}[H]
	\centering
	\vspace{-3mm}
	\begin{tabular}{ccc}
		\includegraphics[width =.4\textwidth]{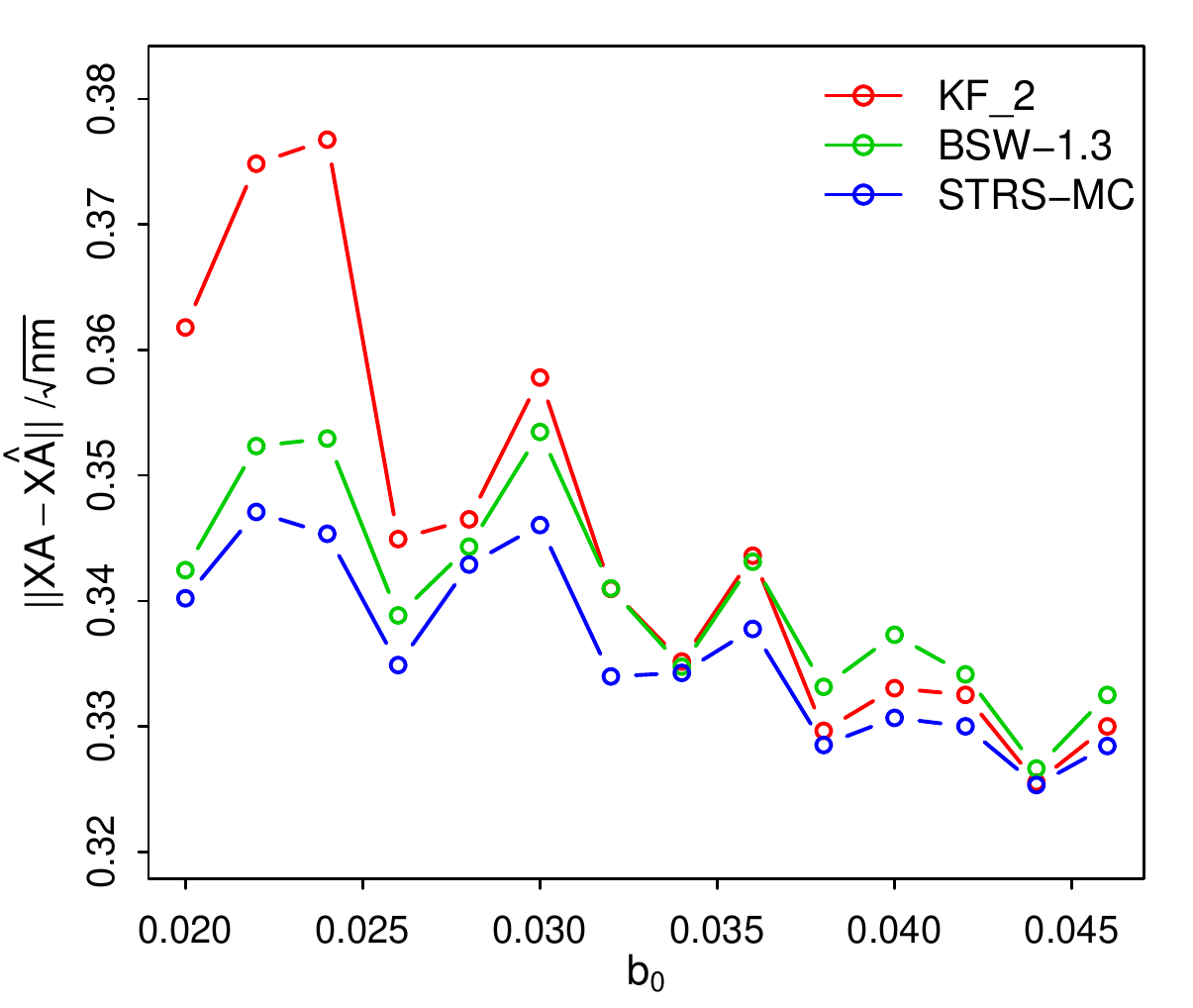} &
		\includegraphics[width =.4\textwidth]{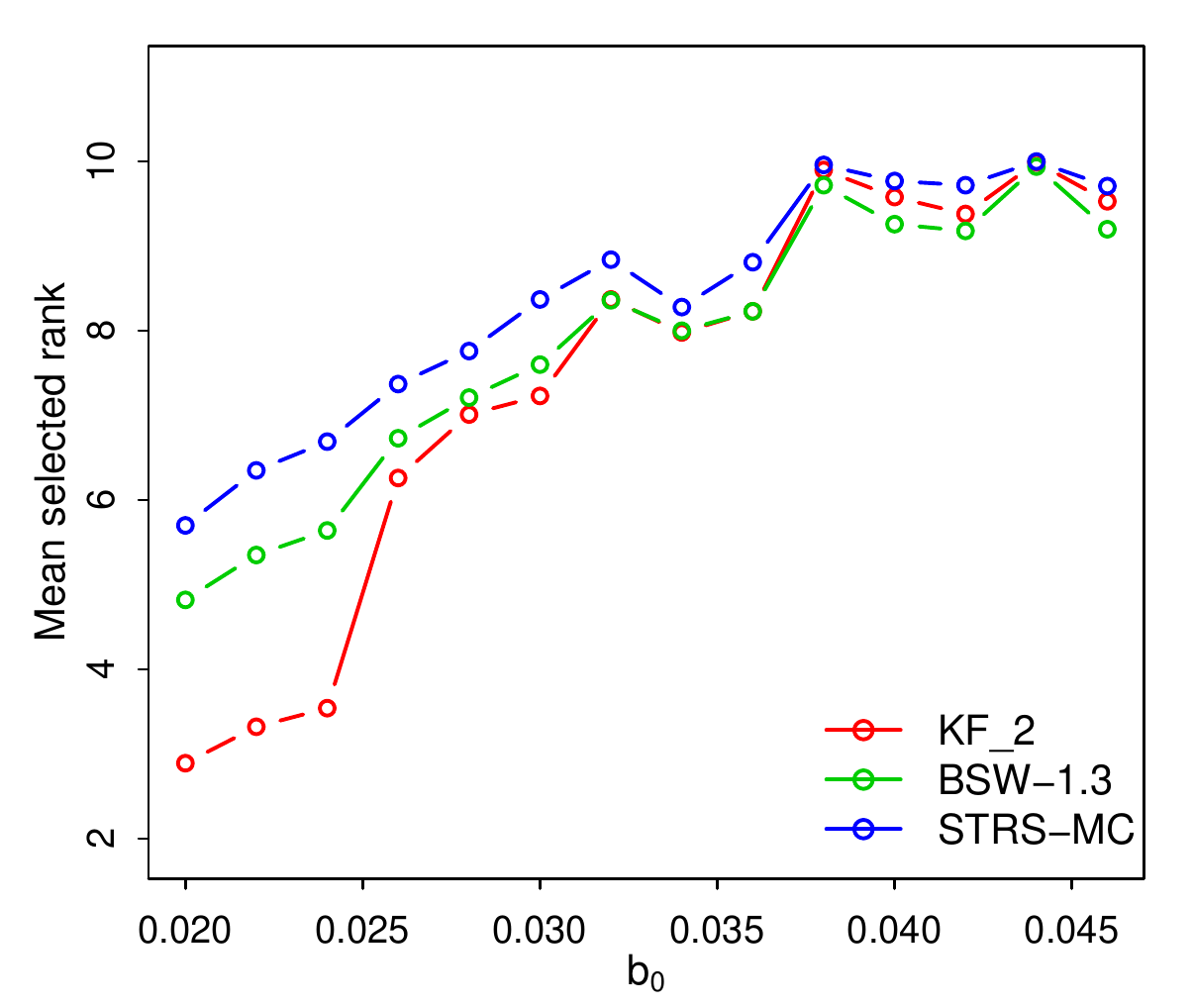}
		\vspace{-2mm}
	\end{tabular}
	\begin{tabular}{cc}
		\includegraphics[width =.4\textwidth]{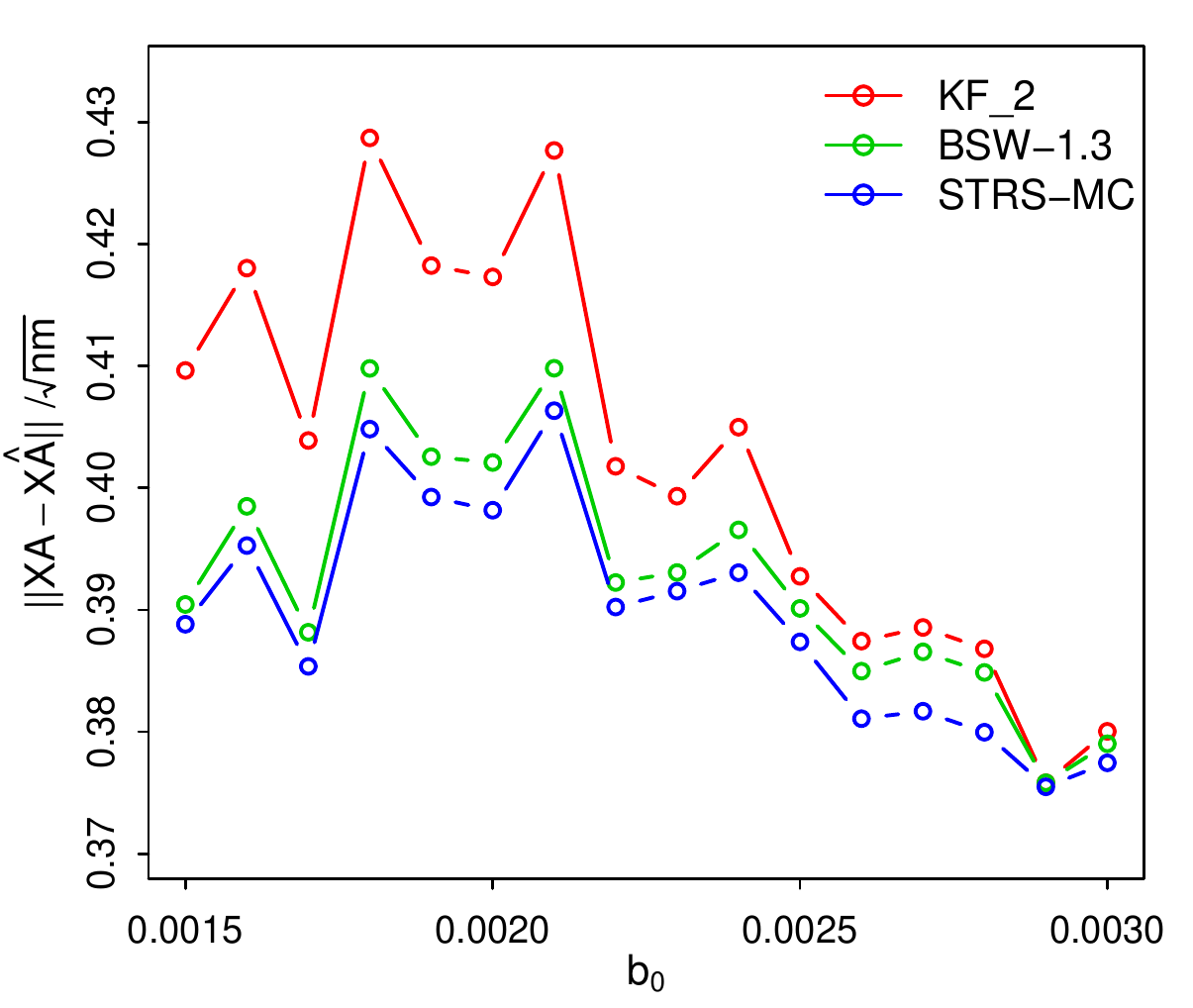}&
		\includegraphics[width =.4\textwidth]{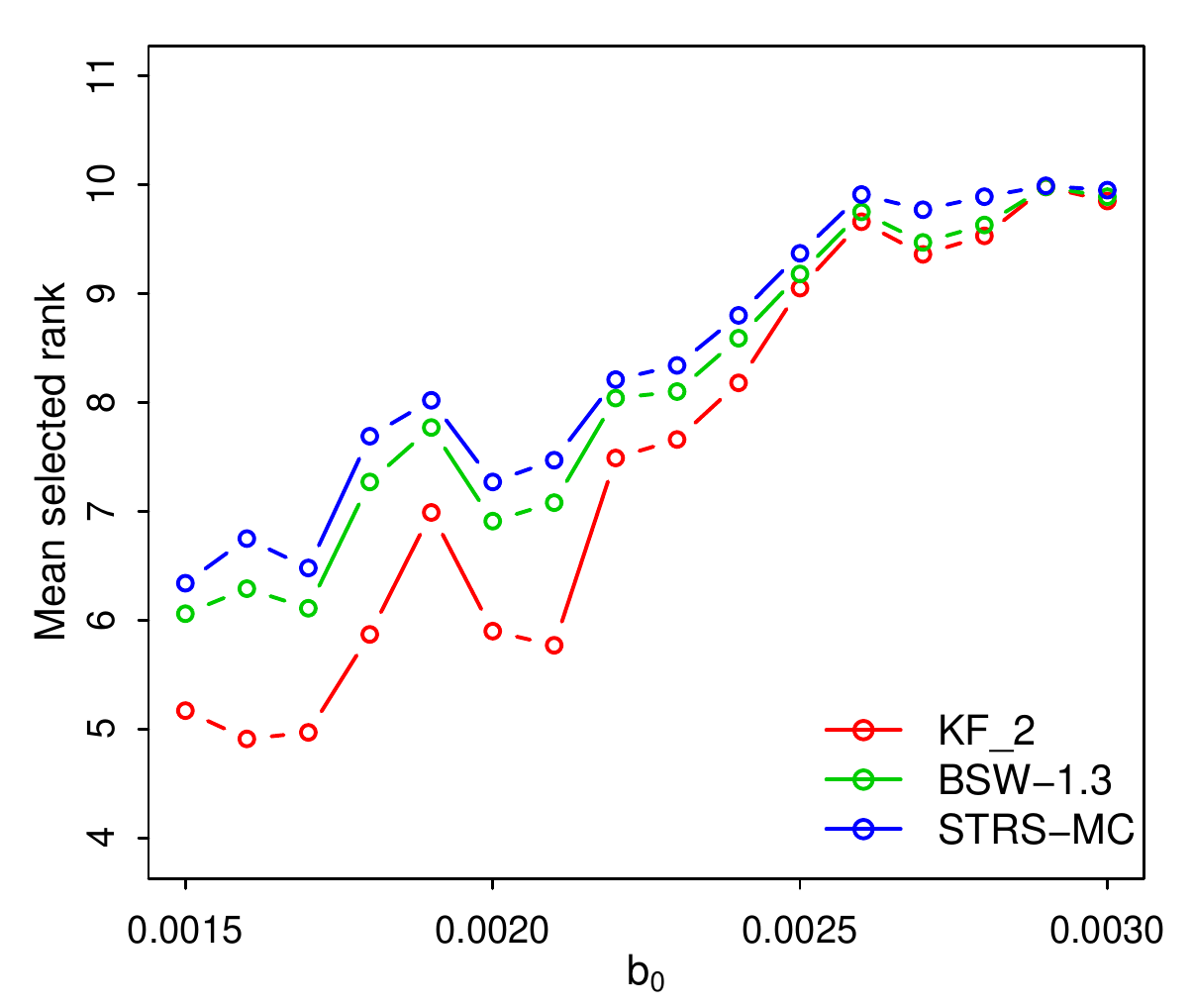} 
	\end{tabular}
	\vspace{-1mm}
	\caption{Criterions (1) and (2) of STRS-MC, KF-2 and BSW-1.3 as $b_0$ varies in the low-dimensional setting (first row) and  high-dimensional setting (second row).}
	\label{fig_correct}
	\vspace{-2mm}
\end{figure}

\subsubsection{Approximate low rank model}\label{apr}

We devote this part to testing the performance of different approaches when the model is mis-specified in that  $A$ doesn't have an exact low rank structure, but rather  has a small effective rank with non-zero decaying singular values. 
To generate this type of $A$, we first generate an exact low rank matrix $A^*$ with specified rank $r$ based on our data generating mechanism described in Section \ref{SimSetup}. 
Next, we compute its singular value decomposition $A^* = UDV^T$ with $D = \textrm{diag}(d_1, \ldots, d_r, 0, \ldots, 0)$ and add polynomial decaying noise   to the (zero valued) singular values $d_j$ for $j\ge r$. Specifically, we take $d_j = d_r\cdot \gamma(j-r+1)^{-\beta}$ for $j\ge r+1$, with $\gamma\in (0,1)$ and some positive integer $\beta$. The matrix $A = U\wt DV^T$ is our approximate low rank matrix with $\wt D = \textrm{diag}(d_1,\ldots, d_r, d_{r+1},d_{r+2},\ldots)$. Since the results are similar for different choices of $\gamma$ and $\beta$, we only present the results of $\gamma = 0.8$ and $\beta=  1$. The rest of our simulation setup stays the same as  in the  exact low rank model case above. Criterions (1) and (2) for both low- and high-dimensional settings are shown in Figure \ref{fig_misspec}. 
We find the same conclusion as before when the model is mis-specified:  STRS-MC continues to  outperform the other two methods by yielding the smallest $\|X\wh A-XA\|$ and selecting a rank closer to the effective rank (of 10).

\begin{figure}[H]
	\centering
	\vspace{-2mm}
	\begin{tabular}{ccc}
		\includegraphics[width =.4\textwidth]{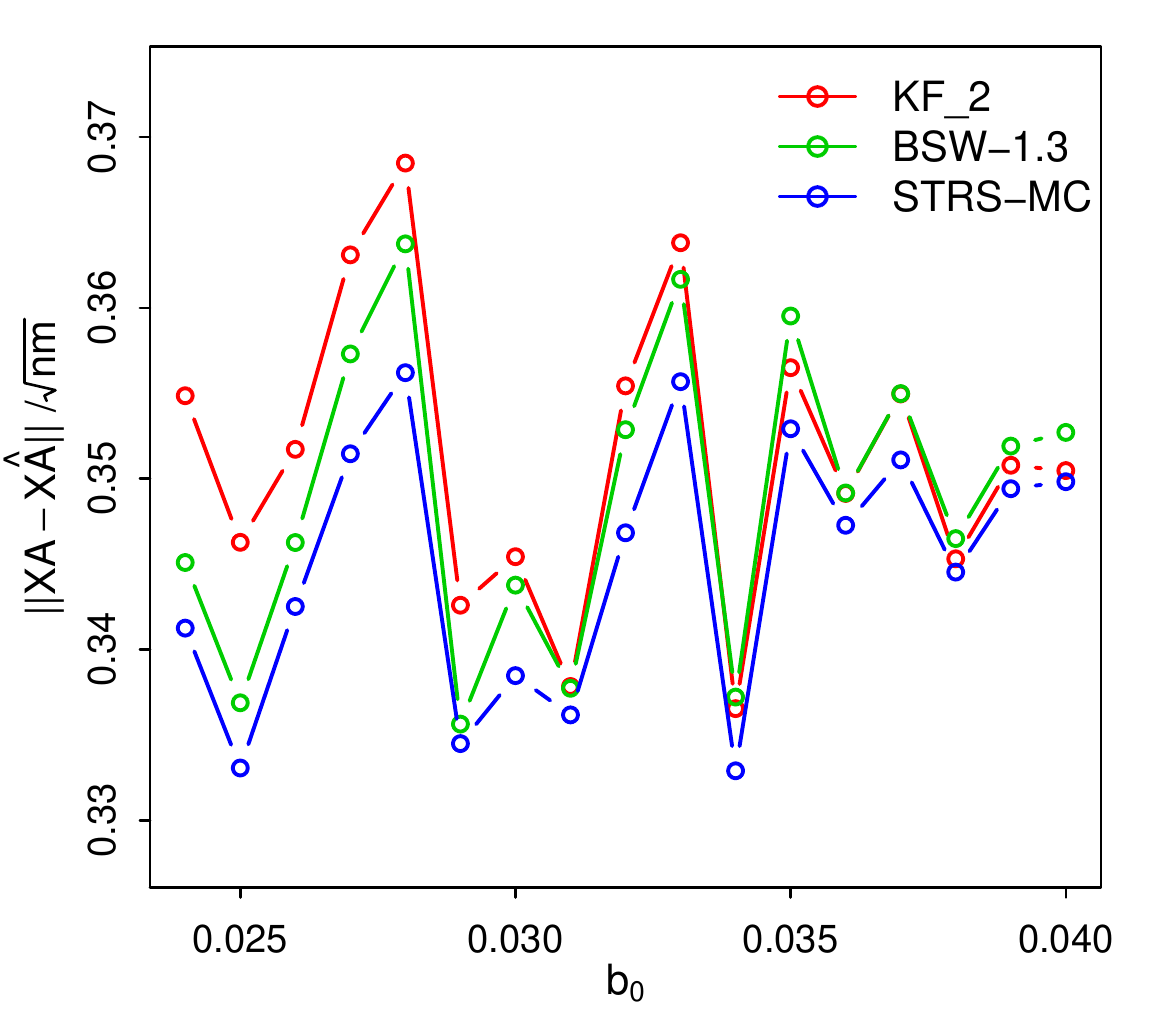} &
		\includegraphics[width =.4\textwidth]{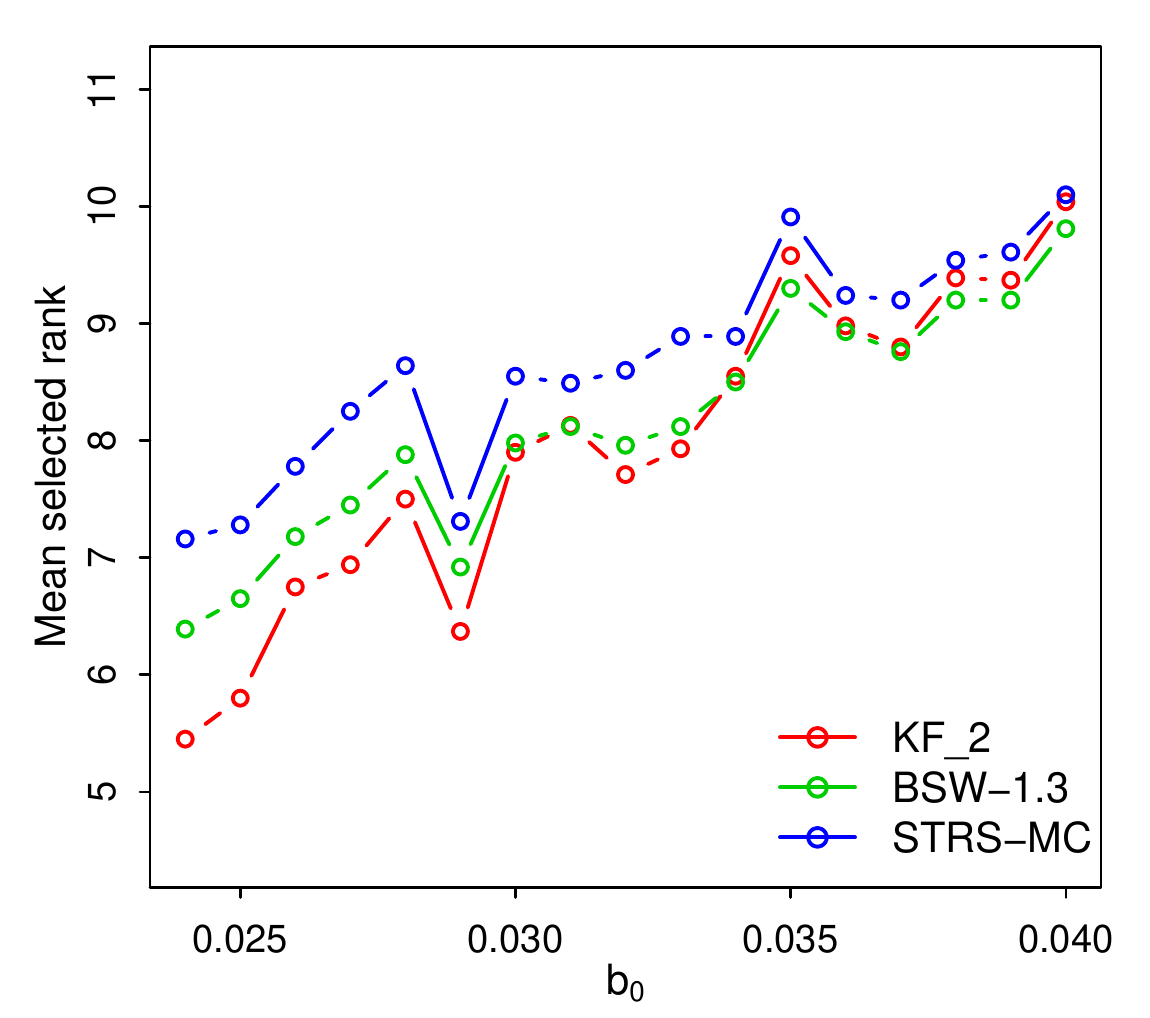}
		\vspace{-2mm}
	\end{tabular}
	\begin{tabular}{cc}
		\includegraphics[width =.4\textwidth]{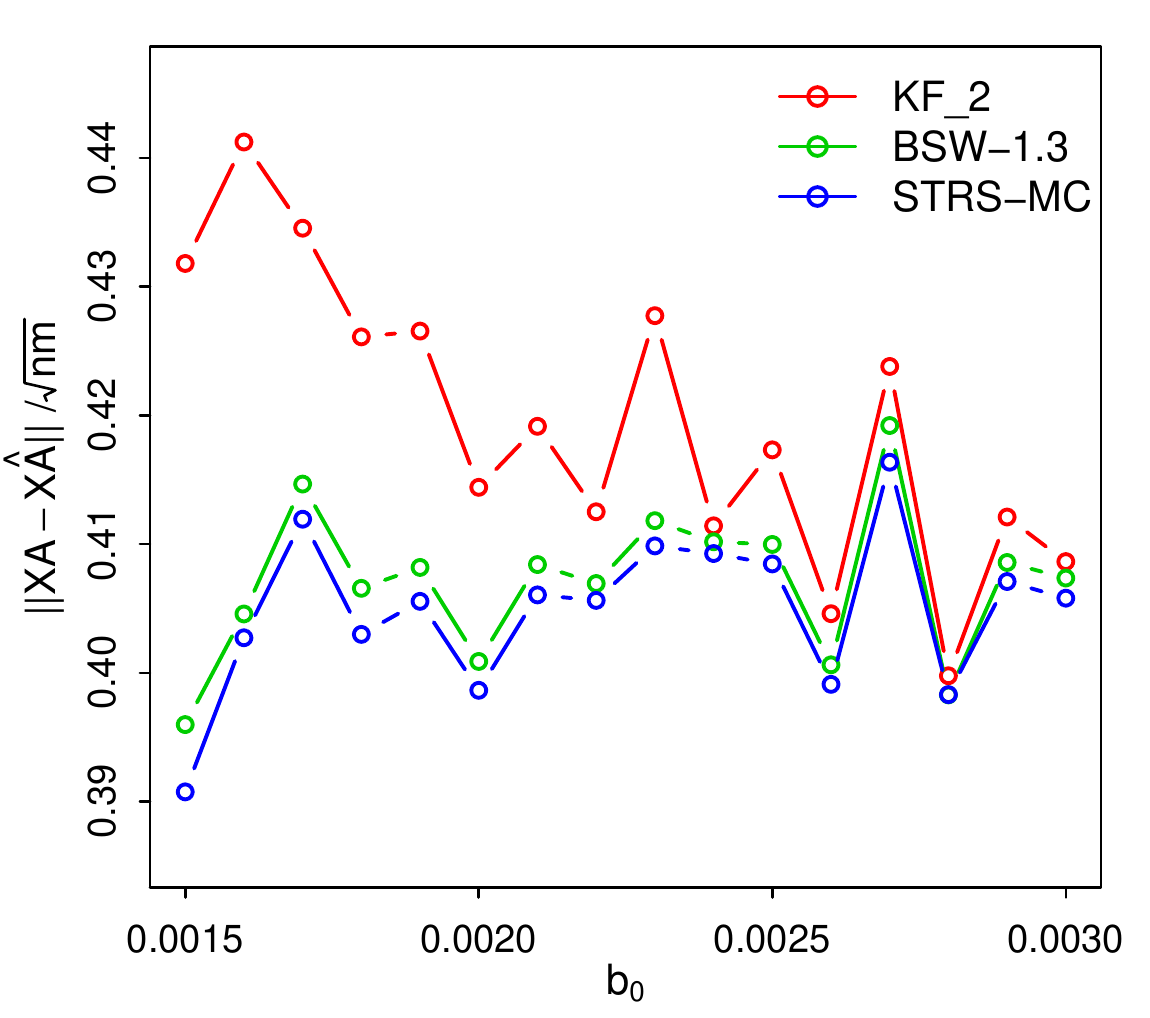}&
		\includegraphics[width =.4\textwidth]{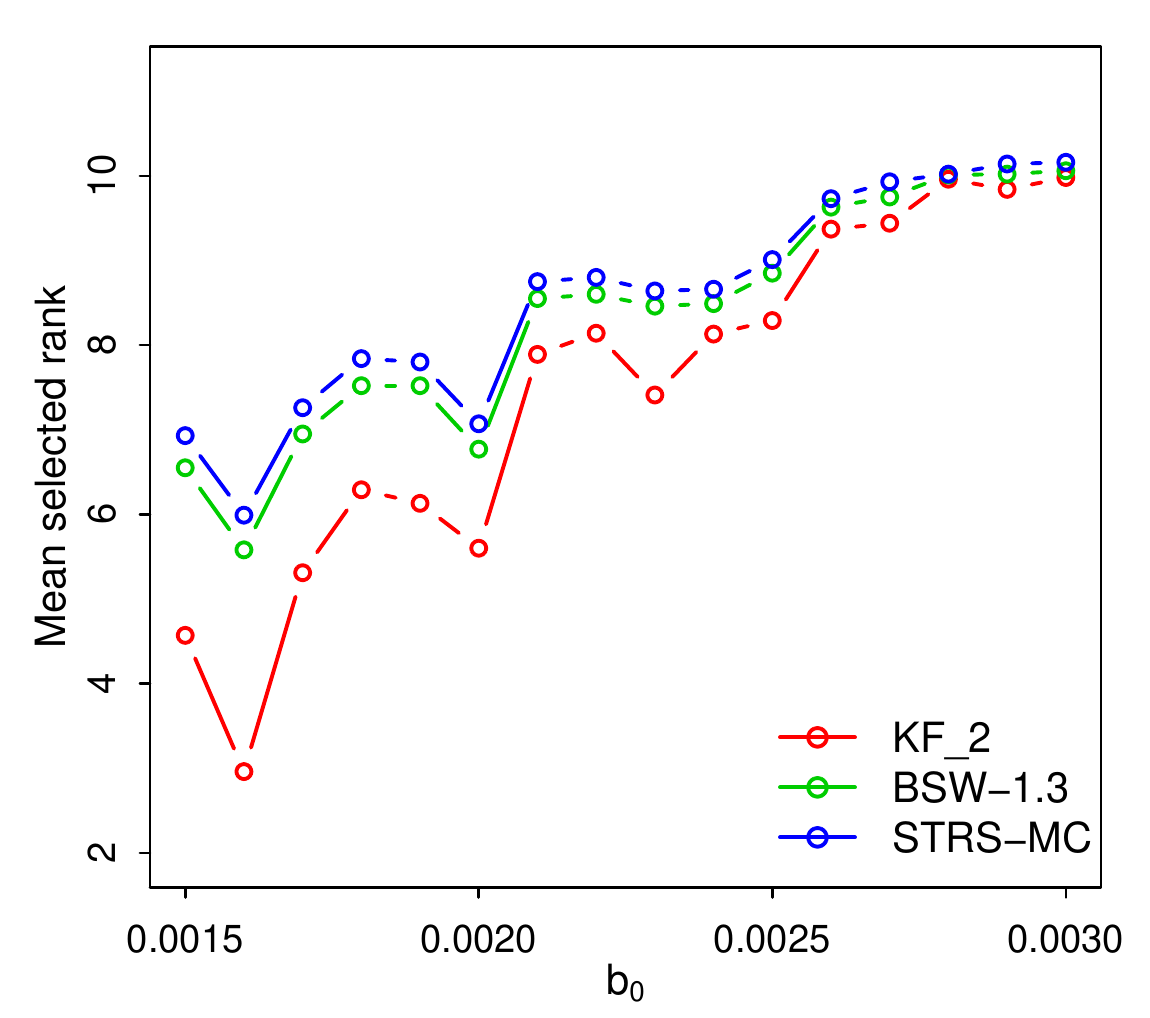} 
	\end{tabular}
	\vspace{-1mm}
	\caption{Criterions (1) and (2) of STRS-MC, KF-2 and BSW-1.3 as $b_0$ varies in  the low-dimensional setting (top) and the  high-dimensional setting (bottom).}
	\label{fig_misspec}
\end{figure}

\subsubsection{The prediction error $\|\wh A-A\|$}
We compute the prediction error $\|\wh A-A\|/\sqrt{pm}$ in the low-dimensional setting for both the exact low rank model and the approximate low rank model. Figure \ref{fig_pred} shows that STRS-MC dominates the other two methods, demonstrating the importance of selecting a better rank. 

\begin{figure}[H]
	\centering
	\begin{tabular}{cc}
		\includegraphics[width =.42\textwidth, height = 0.243\textheight]{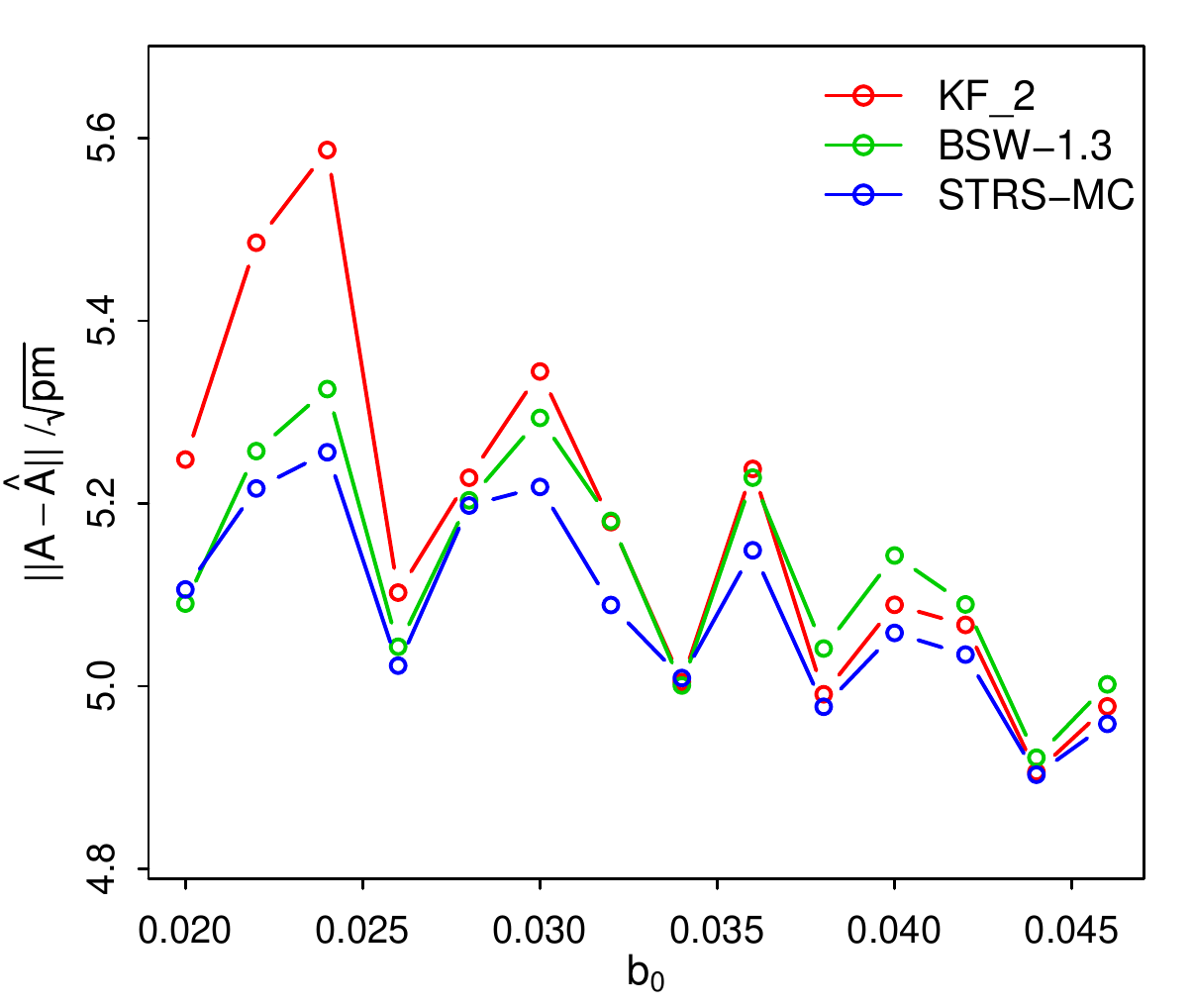}&  
		\includegraphics[width =.42\textwidth]{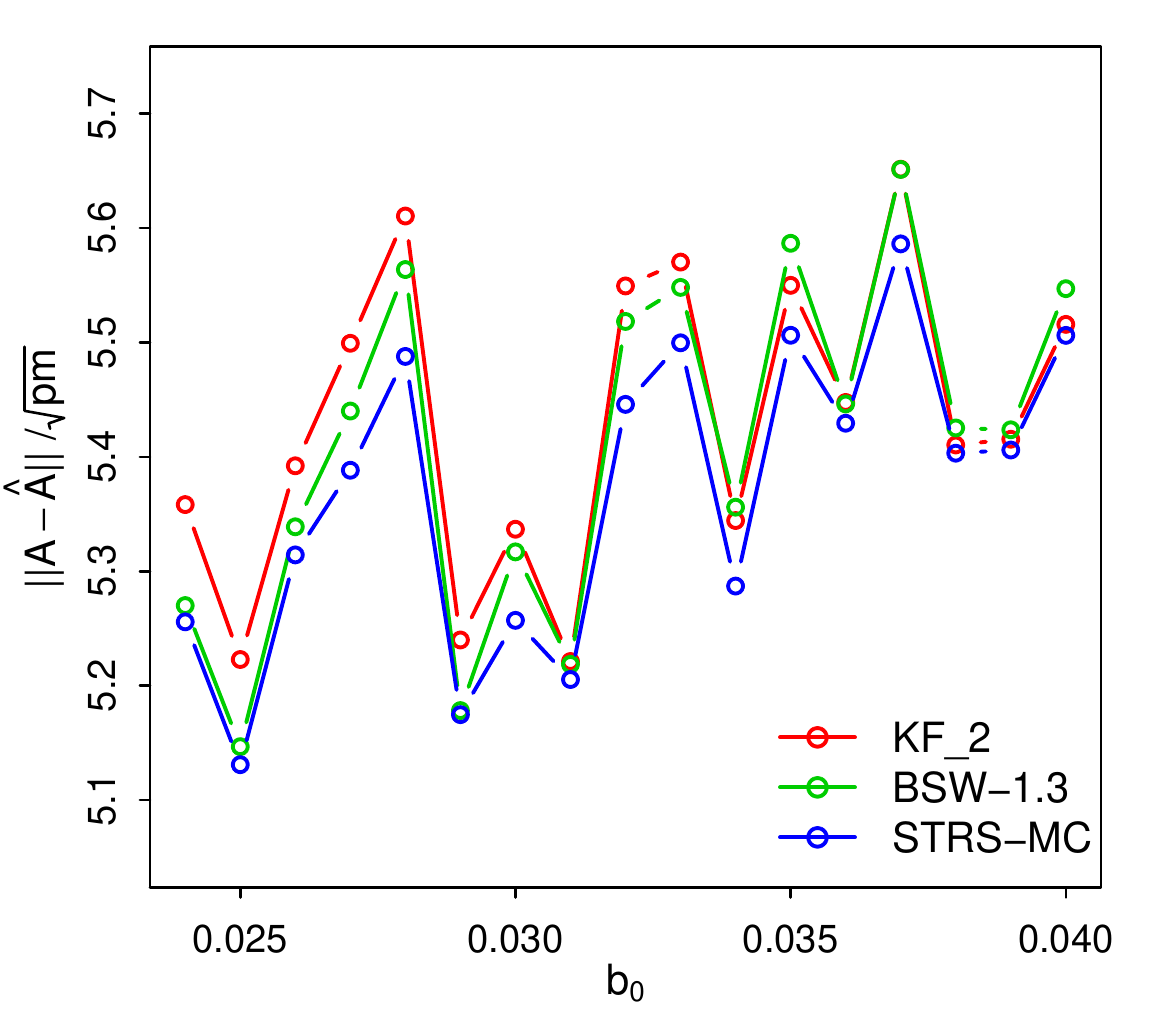}
	\end{tabular}
	\caption{Plots of $\|\wh A-A\|$ in the exact low-rank model (left) and  in the approximate low-rank model (right).}
	\label{fig_pred}
\end{figure}
	
\end{document}